\documentclass[a4paper,UKenglish,cleveref, autoref, thm-restate]{tgdk-v2021}
\usepackage{graphicx} % Required for inserting images
\usepackage[final]{todonotes}
\usepackage[table,x11names]{xcolor}

\usepackage[inline]{enumitem}
\usepackage[commandnameprefix=always]{changes}

\setaddedmarkup{\color{black}#1}
\setdeletedmarkup{} % remove
\sethighlightmarkup{#1}
\usepackage{url}
\usepackage{multicol}
\usepackage{makecell}
\usepackage{url}

\newcommand{\fx}[0]{façade-x}
\newcommand{\FX}[0]{Façade-X}

\newcommand{\IRI}[0]{$\mathcal{I}$}
\newcommand{\BN}[0]{$\mathcal{B}$}
\newcommand{\LIT}[0]{$\mathcal{L}$}
\newcommand{\VAR}[0]{$\mathcal{V}$}

\newcommand{\iri}[1]{\texttt{<#1>}}
\newcommand{\bn}[1]{\texttt{\_:#1}}
\newcommand{\lit}[1]{\texttt{"#1"}}
\newcommand{\var}[1]{\texttt{?#1}}

\newcommand{\subject}[1]{\texttt{subj}_{#1}}
\newcommand{\predicate}[1]{\texttt{pred}_{#1}}
\newcommand{\object}[1]{\texttt{obj}_{#1}}

\newcommand{\ax}[1]{$\tt{#1}$}

\newcommand{\sfx}[1]{#1\textsuperscript{FX}}
\newcommand*{\mline}[1]{%
\begingroup
    \renewcommand*{\arraystretch}{1.1}%
   \begin{tabular}[c]{@{}>{\raggedright\arraybackslash}p{2cm}@{}}#1\end{tabular}%
  \endgroup
}

\title{Towards a theory of Façade-X data access: satisfiability of SPARQL basic graph patterns.} %TODO Please add

\author{Luigi {Asprino}}{Università Telematica San Raffaele, Rome, Italy \and \url{https://luigi-asprino.github.io} }{luigi.asprino@uniroma5.it}{https://orcid.org/0000-0003-1907-0677}{}%TODO mandatory, please use full name; only 1 author per \author macro; first two parameters are mandatory, other parameters can be empty. Please provide at least the name of the affiliation and the country. The full address is optional. Use additional curly braces to indicate the correct name splitting when the last name consists of multiple name parts.

\author{Enrico {Daga}\footnote{Corresponding author}}{Knowledge Media Institute (KMI), The Open University, United Kingdom \and \url{http://www.enridaga.net} }{enrico.daga@open.ac.uk}{https://orcid.org/0000-0002-3184-5407}{EU funded projects (a) MultiPod: grant agreement 101178821, \url{https://cordis.europa.eu/project/id/101178821}; and (b) Climate Misinformation Surveillance and Analysis with Multidimensional GIS CHIST-ERA project.}

\authorrunning{L. Asprino and E. Daga} %TODO mandatory. First: Use abbreviated first/middle names. Second (only in severe cases): Use first author plus 'et al.'

\Copyright{Luigi Asprino and Enrico Daga} %TODO mandatory, please use full first names. TGDK license is "CC-BY";  http://creativecommons.org/licenses/by/4.0/

 %TODO mandatory; please add comma-separated list of keywords

\category{} %optional, e.g. invited paper

\relatedversion{} %optional, e.g. full version hosted on arXiv, HAL, or other respository/website
\Volume{?}
\Issue{?}
\Article{?}
\DateSubmission{July 2025; minor revisions submitted on December 2025}
\SectionAreaEditor{}
\usepackage{thmtools}

\newtheorem{algorithm}{Algorithm}[]

\theoremstyle{remark}
\newtheorem{ex}{Example}[]

\lstdefinestyle{logic}{
    breaklines=true,
    mathescape, 
    columns=fullflexible,
    basicstyle=\ttfamily,
    numbers=left,
    stepnumber=1,
    showspaces=false,
    showstringspaces=false,
    tabsize=4,
    numbersep=4pt,
    numberstyle=\footnotesize\ttfamily,
    framesep=16pt,
    framerule=1pt,
    xleftmargin=16pt,
    backgroundcolor=\color{white},
    commentstyle=\color{green}\itshape,
    keywordstyle=\color{blue}\bfseries\itshape,
    stringstyle=\color{red},
    escapechar=|
}
\lstdefinestyle{rdf}{
    breaklines=true,
    mathescape, 
    columns=fullflexible,
    basicstyle=\tt,
    numbers=left,
    stepnumber=1,
    showspaces=false,
    showstringspaces=false,
    tabsize=4,
    numbersep=4pt,
    numberstyle=\footnotesize\ttfamily,
    framesep=16pt,
    framerule=1pt,
    xleftmargin=16pt,
    backgroundcolor=\color{white},
    commentstyle=\color{green}\itshape,
    keywordstyle=\color{blue}\bfseries\itshape,
    stringstyle=\color{red},
    escapechar=|
}
\lstdefinestyle{algorithm}{
    breaklines=true,
    mathescape, 
    columns=fullflexible,
    basicstyle=\footnotesize\ttfamily,
    numbers=left,
    stepnumber=1,
    showspaces=false,
    showstringspaces=false,
    tabsize=4,
    numbersep=4pt,
    numberstyle=\footnotesize\ttfamily,
    framesep=16pt,
    xleftmargin=16pt,
    backgroundcolor=\color{white},
    commentstyle=\color{green}\itshape,
    keywordstyle=\color{blue}\bfseries\itshape,
    stringstyle=\color{red}
}
\colorlet{punct}{red!60!black}
\definecolor{background}{HTML}{EEEEEE}
\definecolor{delim}{RGB}{20,105,176}
\colorlet{numb}{magenta!60!black}

\lstdefinelanguage{json}{
    basicstyle=\normalfont\ttfamily,
    numbers=left,
    numberstyle=\scriptsize,
    stepnumber=1,
    numbersep=8pt,
    showstringspaces=false,
    breaklines=true,
    frame=lines,
    backgroundcolor=\color{background},
    literate=
     *{0}{{{\color{numb}0}}}{1}
      {1}{{{\color{numb}1}}}{1}
      {2}{{{\color{numb}2}}}{1}
      {3}{{{\color{numb}3}}}{1}
      {4}{{{\color{numb}4}}}{1}
      {5}{{{\color{numb}5}}}{1}
      {6}{{{\color{numb}6}}}{1}
      {7}{{{\color{numb}7}}}{1}
      {8}{{{\color{numb}8}}}{1}
      {9}{{{\color{numb}9}}}{1}
      {:}{{{\color{punct}{:}}}}{1}
      {,}{{{\color{punct}{,}}}}{1}
      {\{}{{{\color{delim}{\{}}}}{1}
      {\}}{{{\color{delim}{\}}}}}{1}
      {[}{{{\color{delim}{[}}}}{1}
      {]}{{{\color{delim}{]}}}}{1},
}

\begin{document}

\maketitle

%TODO mandatory: add short abstract of the document
\begin{abstract}
Data integration is the primary use case for knowledge graphs.
However, integrated data are not typically graphs but come in different formats, for example, CSV, XML, or a relational database.
\FX{} is a recently proposed method for providing direct access to an open-ended set of data formats. 
The method includes a  meta-model that specialises RDF to fit general data structures.
This model allows to express SPARQL queries targeting data sources with those structures.
Previous work formalised Façade-X and demonstrated how it can theoretically represent any format expressible with a context-free grammar, as well as the relational model.
A reference implementation, SPARQL Anything, demonstrates the feasibility of the approach in practice.
It is noteworthy that \FX{} utilises a fraction of RDF, and, consequently, not all SPARQL queries yield a solution (i.e. are satisfiable) when evaluated over a \FX{} graph.
In this article, we consolidate \FX{}  and we study the \textit{satisfiability} of basic graph patterns.
The theory is accompanied by an algorithm for deciding the satisfiability of basic graph patterns on \FX{} data sources.
\chadded{Furthermore, we provide extensive experiments with a proof-of-concept implementation, demonstrating practical feasibility, including with real-world queries.}
Our results pave the way for studying query execution strategies for \FX{} data access with SPARQL and supporting developers to build more efficient data integration systems for knowledge graphs.
\end{abstract}

\section{Introduction}\label{sec:introduction}Knowledge graphs use a graph-based model to integrate, manage, and extract value from diverse data sources at scale~\cite{hogan2021knowledge}. 
Thus, data integration is the primary application of Knowledge Graphs (KG).
The W3C standard SPARQL 1.1~\cite{w3c:sparql11} serves as the authoritative language for engaging with knowledge graphs serialised in RDF~\cite{w3c:rdf11}. 
SPARQL provides functionalities for selecting, filtering, and aggregating data into a table format. 
Additionally, SPARQL allows the projection of results into an RDF structure using the CONSTRUCT query type. 
Consequently, research has investigated the use of SPARQL across a variety of applications beyond mere querying, particularly in the integration of diverse data sources~\cite{cyganiak2015tarql,lefranccois2017sparql,michel2018sparql}.
In practice, methods often depend on enhancing SPARQL to retrieve data from non-RDF formats.

Recent studies~\cite{daga2021facade} suggest utilizing an intermediate RDF model called \FX{}, which allows its components to be seamlessly converted into RDF Knowledge Graphs. %\todo{FX components converted into files?} 
This approach enables the development of software that offers \textit{direct access} to source data as RDF, freeing knowledge engineers from the burden of managing diverse formats and associated languages they depend on\chadded{.}
\chadded{\FX{} uniquely addresses }\chdeleted{--} \textit{re-engineering}, which means transforming resources by minimizing domain-specific concerns and concentrating on the syntactical meta-model, and allow\chreplaced{s users}{ing them} to focus on \chdeleted{the semantic enhancement -- }\textit{remodelling}, that is, transforming the original domain model into a new representation~\cite{daga2021facade}.
%Façade-X provides consistent access to a diverse array of data formats \textit{as-if} they were RDF, including widely used formats like CSV, JSON, HTML, and XML. It has been effectively utilized in various intricate situations, such as extracting content from websites, merging data from different sources, and even constructing knowledge graphs from music scores~\cite{rattaknowledge}.
In theory, \FX{} can represent any format defined by a context-free grammar, as well as the relational data model~\cite{asprino2023knowledge}.

The \FX{} method informs the SPARQL Anything software and open source project, which also serves as its reference implementation\footnote{\url{http://sparql-anything.cc}}.  
%The method can be accessed via a command-line interface, through a server, and through code written in Java and Python~\cite{ratta2024pysparql}\footnote{See also the comprehensive online documentation: \url{https://sparql-anything.readthedocs.io/}}.\todo{Ha solo l'obiettivo di far pubblicità? o c'è un altra motivazione per uqesta frase?}

% \begin{figure}
%     \centering
%     \begin{lstlisting}[language=json]
% [
%   {
%     "name": "Friends",
%     "genres": [
%       "Comedy",
%       "Romance"
%     ],
%     "language": "English",
%     "status": "Ended",
%     "premiered": "1994-09-22",
%     "summary": "Follows the personal and professional lives of six twenty to thirty-something-year-old friends living in Manhattan.",
%     "stars": [
%       "Jennifer Aniston",
%       "Courteney Cox",
%       "Lisa Kudrow",
%       "Matt LeBlanc",
%       "Matthew Perry",
%       "David Schwimmer"
%     ]
%   },
%   {
%     "name": "Cougar Town",
%     "genres": [
%       "Comedy",
%       "Romance"
%     ],
%     "language": "English",
%     "status": "Ended",
%     "premiered": "2009-09-23",
%     "summary": "Jules is a recently divorced mother who has to face the unkind realities of dating in a world obsessed with beauty and youth. As she becomes older, she starts discovering herself.",
%     "stars": [
%       "Courteney Cox",
%       "David Arquette",
%       "Bill Lawrence",
%       "Linda Videtti Figueiredo",
%       "Blake McCormick"
%     ]
%   }
% ]
%     \end{lstlisting}
%     \caption{Example JSON file}
%     \label{fig:example:json}
% \end{figure}

\begin{ex}
\label{running_example}

\begin{figure}[h!]
    \centering
    \begin{lstlisting}[]
email,name,surname
laura@example.com,Laura,Grey
craig@example.com,Craig,Johnson
mary@example.com,Mary,Jenkins
jamie@example.com,Jamie,Smith
    \end{lstlisting}
    \caption{Example CSV file}
    \label{fig:example:csv}
\end{figure}

Suppose having the \chreplaced{CSV}{JSON} file shown in Figure~\ref{fig:example:csv} as input\footnote{Also available at \url{https://sparql-anything.cc/example1.csv)}}.
Using SPARQL Anything, a SPARQL user can retrieve the surnames of people whose first name is 'Laura' by running the query\chreplaced{ in Listing~\ref{ex:running:query}.}{:}
\begin{lstlisting}[language=sparql,caption={A query for retrieving the surname of people whose first name is 'Laura'.},label={ex:running:query}]
PREFIX xyz: <http://sparql.xyz/facade-x/data/>

SELECT ?surname WHERE {
  SERVICE <x-sparql-anything:https://sparql-anything.cc/example1.csv> {
    _:person xyz:surname ?surname .
    _:person xyz:name "Laura" .
  }
}
\end{lstlisting}
obtaining the following result:

\begin{table}[h!]
    \centering
    \begin{tabular}{| l |}
\hline
\textbf{surname} \\
\hline

\hline
"Grey" \\
\hline

\end{tabular}
\end{table}

\end{ex}

This system (accessible through a web server, CLI, Java, or Python API~\cite{ratta2024pysparql})  has been utilized in numerous practical applications and is gaining recognition from the knowledge graph community in both research and industry~\cite{fx:endorse2023}\footnote{See also the activity on the open-source project page on GitHub: \url{http://github.com/sparql-anything/sparql.anything}}. 
\chadded{The W3C Data Façades Community Group has been recently founded to develop \FX{} as a vendor-neutral, interoperable technology that can be implemented by various developers and organisations, including graph database vendors}\footnote{W3C Data Façades Community Group: \url{https://www.w3.org/community/facade-x/}.}

Currently, the execution of SPARQL queries on \FX{} data sources requires the materialisation of the source (in full or via a slicing approach~\cite{asprino2025materialisation}).
An open question is how to support the implementation of more efficient query engines.
It is noteworthy that \FX{} utilises a fraction of RDF, and, consequently, not all SPARQL queries yield a solution (i.e. are satisfiable) when evaluated over a \FX{} graph. 
\chadded{For example, the query described in the Example~\ref{running_example} is satisfiable, whereas the query in the Listing~\ref{ex:running:unsatisfiable} is not. This is because the \FX{} model does not allow for such a pattern: CSV data sources don't have any \texttt{rdf:type} triples; therefore, no results would be produced from its evaluation. These considerations can be made simply by looking at the query itself (rather than the input) and could save computational resources.}

\begin{lstlisting}[language=sparql,caption={A query for retrieving the surname of people whose first name is 'Laura'.},label={ex:running:unsatisfiable}]
PREFIX rdf: <http://www.w3.org/1999/02/22-rdf-syntax-ns#>

SELECT ?s WHERE {
  SERVICE <x-sparql-anything:https://sparql-anything.cc/example1.csv> {
    ?s rdf:_1 ?o . ?x rdf:type ?s .
  }
}
\end{lstlisting}

In this article, we study the satisfiability of basic graph patterns under \FX{}. 
In addition, we contribute an algorithm to annotate any SPARQL BGP with possible solution patterns, to determine whether a given BGP can have a solution within \FX.
We contribute:
\begin{itemize}
    \item A consolidated formal definition of  \FX{} and associated mappings to RDF;
    %\item Façade-X profiles for CSV, JSON, XML file formats;\todo{check}
    %\item Façade-X profile for the relational data;\todo{check}
    \item Characterisation of the satisfiability of SPARQL basic graph patterns queries within \FX{};
    \item Algorithm for assessing whether a SPARQL basic graph patterns is satisfiable under \FX{}.
    %\item A methodology for applying the algorithm to an open-ended set of Façade-X profiles\todo{check}.
\end{itemize}

The rest of the article is structured as follows.
The next section is dedicated to surveying related literature and contextualise our work.
Section~\ref{sec:preliminaries} introduces the formalisms and languages adopted in this article.
Section~\ref{sec:facadex} presents a consolidated version of \FX{} whose mapping to RDF is formalised in Section~\ref{sec:mapping}.
Finally, Section~\ref{sec:theory} studies the conditions under which a BGP is satisfiable in \FX{} which are exploited in the algorithm presented in Section~\ref{sec:algorithms}. 
Finally, Section~\ref{sec:experiments} provides experiments with a proof-of-concept implementation, which demonstrates the practical feasibility \chadded{and advantages} of our theory\chadded{, including an evaluation with real world queries}.
We conclude the article with Section~\ref{sec:conclusions}.

\section{Related work}\label{sec:relatedwork}We survey literature on Knowledge Graph Construction (KGC), i.e. systems and methods to onboard non-RDF data into RDF databases.
We then look into approaches to extend SPARQL.
Next, we focus on the evaluation of SPARQL queries, particularly basic graph patterns.
Furthermore, we review the problem of satisfiability of graph patterns outside the RDF area.
We conclude with a look at various approaches to query documents and provide a unifying view over heterogeneous resources.
\subsection{Knowledge graph construction}
Mapping languages for transforming heterogeneous files into RDF are represented by RML~\cite{iglesias2023rml,van2023declarative}, also specialised to support data cleaning operations~\cite{slepicka2015kr2rml}, and specific forms of data: relational~\cite{rodriguez2015efficient}, geospatial data~\cite{kyzirakos2014geotriples}, etc.
RML has been adopted as reference mapping language in most of the KGC systems, such as SDM-RDFizer~\cite{iglesias2020sdm}, Morph-CSV~\cite{chaves2021enhancing},and Morph-KGC~\cite{arenas2024morph}. 
Planning techniques have been studied to improve the efficiency of RML engines~\cite{iglesias2023scaling}.
% RML is also representative for general data integration approaches such as OBDA~\cite{xiao2018ontology}.
A recent work provides an algeabric foundation for RML based approaches~\cite{oo2025algebraic}, which establish a \textit{query} operator as a primitive component in the definition of a mapping. \FX{}~could be used as a common language to express queries over heterogeneous data sources \textit{formally}. %~\todo[]{Connect the Query variable to \FX graph patterns, and say that \FX can be used as a common abstraction to formally express queries over heterogeneus data sources.}

Authors of an alternative to RML, based on the RDF constraing language ShEx~\cite{garcia2020shexml}, stress the importance of making mappings usable by end users.
Indeed, literature acknowledges how these languages are built with machine-processability in mind~\cite{heyvaert2018declarative}, and how defining or even understanding the rules is not trivial. 
%When the user is proficient in SPARQL, that is one of the reasons for choosing \FX.
\FX{} provides a common abstraction allowing to express queries to non-RDF sources using basic graph patterns. 
\chadded{It is worth noting that,  although  abstractions targeting individual data formats are already available in the literature, \FX{} keeps a single formalism and a single data model (i.e. RDF) for all data formats.}
Therefore, it paved the way for querying non-RDF data in SPARQL and building KGs, by making plain CONSTRUCT queries with a SPARQL engine supporting~\FX{} data access.
\chadded{Moreover, since \FX{} is an abstract model that has been defined independently of RDF (as shown in Section~\ref{sec:facadex}, its specification is in predicate logic), it can potentially be applied to other languages (e.g. Datalog). Indeed, other query languages could benefit from the homogeneous overview of heterogeneous sources provided by \FX{}.
Although this analysis could be extended to other frameworks for representing and querying data, our focus is on SPARQL.}
%one of the ways knowledge engineers can onboard non-RDF data
%Therefore, it is a first-class citizen in the landscape of KGC methods. 
However, it is an open question \textit{how to efficiently execute SPARQL queries on \FX{} data sources} that cannot be fully loaded in memory via a materialisation approach~\cite{asprino2025materialisation}.
This problem motivates the present work.

\subsection{Approaches to extend SPARQL}
We examine methods for enhancing SPARQL. A common technique for augmenting SPARQL is to introduce custom functions that can be utilized in FILTER or BIND operators\footnote{ARQ includes a library of custom functions that facilitate aggregates, such as calculating the standard deviation of a set of values. ARQ functions: \url{https://jena.apache.org/documentation/query/extension.html} (accessed 15/12/2020).}. Query processing engines can enhance SPARQL through the use of so-called magic properties. This method involves defining custom predicates to guide specific actions during query execution\footnote{For instance, this allows for the specification of intricate fulltext searches over literal values. Query processors can outsource execution to a fulltext engine (e.g., Lucene) and produce a set of query results as triple patterns}. A comparable method is employed to incorporate XPath and XQuery in Virtuoso, which also enables the embedding of SQL functions within SPARQL to directly access the underlying relational model\footnote {Virtuoso Open Source. SPARQL inline in SQL: \url{http://docs.openlinksw.com/virtuoso/rdfsparqlinline/}.}. SPARQL Generate~\cite{lefranccois2017sparql} presents a technique for transforming data from varied sources into RDF by enhancing the SPARQL syntax with a new GENERATE operator~\cite{lefranccois2017sparql}. This method introduces two additional operators, SOURCE and ITERATOR. Custom functions execute ad-hoc operations on supported formats, utilizing XPath or JSONPath. Nevertheless, there are methods to augment SPARQL without modifying the standard syntax. 
For instance, BASIL variables permit the formulation of parametric queries by imposing a convention on SPARQL variable names~\cite{daga2015basilar,merono2016grlc}. 
SPARQL Anything leverages BASIL variables to facilitate parametric queries and file names. 
SPARQL Micro service~\cite{michel2018sparql} offers a framework that, based on an API mapping specification, encapsulates web APIs as SPARQL endpoints and employs a JSON-LD profile to convert the API's JSON responses into RDF. 
\FX{} adopts a similar minimalist strategy and extends SPARQL by \textit{overriding} the functionality of the SERVICE operator.

\subsection{Evaluation of SPARQL graph patterns}
%\todo[inline]{revise + extend}
Foundations of SPARQL query executions can be found in~\cite{hogan2021knowledge}.
Seminal work in the area include studying the semantics of the query language~\cite{perez2006semantics,arenas2009semantics} and optimise the execution of graph pattern matching~\cite{stocker2008sparql,hartig2007sparql,fletcher2008algebra,vidal2010efficiently}.
%In our work, we refer to SPARQL 1.1 as defined in the W3C specifications.
%More recently, SPARQL canonalisation is investigated in~\cite{salas2022semantics}.
Recent work tackles the efficient execution of SPARQL queries via translation in Datalog~\cite{angles2023sparqlog} or via combining logical and physical query plans~\cite{pang2023gfov}.

The satisfiability problem is undecidable in SPARQL 1.0 (and, therefore, in 1.1 that is its extension) but becomes decidable for basic graph patterns~\cite{zhang2016satisfiability}. 
%Although decidability is NP-complete in theory, it is not in practice, since SPARQL BGPs are typically small~\cite{zhang2016satisfiability}.
Although decidability is NP-complete in theory, decidability of SPARQL BGPs is tractable in practices as the SPARQL BGPs are typically small~\cite{zhang2016satisfiability}.
%\todo{un problema NP complemento non cambia classe se l'istanza su cui è risolto è più piccola, semmai diventa trattabile. nel commento lascio la frase precedente}
Our work starts from the observation that \FX{} RDF graphs are less expressive than general RDF data. %\todo{sostituito "different" con "less expressive"} 
This motivates the study of evaluating basic graph patterns in the  context of \FX{} graphs.
However, it is certainly possible that optimization techniques for RDF pattern matching can be also applied to \FX. 
Hence, in the present work, we investigate the nature of \FX{} graphs and their corresponding basic graph patterns in SPARQL.

\chadded{Another line of research analyses the satisfiability of BGPs in the context of conjunctive query (CQ) theory. Indeed, a SPARQL BGP can be viewed as a CQ over an RDF graph, thus inheriting the properties of CQs. A well-known result of Yannakakis~\cite{yannakakis1981algorithms} demonstrates that the evaluation of acyclic  CQs   is solvable in polynomial time. This result shows that the shape of a query fundamentally impacts the complexity of its evaluation.
This perspective is  relevant to Fa\c{c}ade-X. 
Since the Fa\c{c}ade-X data model only allows for rooted, acyclic structures, thus enforcing the tree-like nature of satisfiable BGPs. }

\subsection{Satisfiability of patterns in graphs}
The problem of efficiently performing graph matching operations is significant beyond SPARQL~\cite{gallagher2006matching,livi2013graph,yan2016short,bouhenni2021survey,deutsch2022graph}.
While the general problem of subgraph pattern matching is known to be NP-complete~\cite{bouhenni2021survey}, it seems that for a wide variety of statistics concerning trees, the expected time complexity for pattern matching is linear~\cite{STEYAERT198319}.
The satisfiability problem of tree patterns queries (within XPath and XQuery) is investigated in~\cite{lakshmanan2004testing}, where it is observed that some of the queries are NP-complete, while others can be resolved in polynomial time.

While \FX{} graphs can be seen as trees, when referring to static data sources, such as JSON documents or XML files (where two equal strings have different identity, because it derives from their position), their representation in RDF are not, because identical types and literals will be projected to unique RDF nodes.
Instead, \FX/RDF graphs are \textit{directed, acyclic, and rooted path graphs}.
In the present work, we focus on satisfiability of SPARQL basic graph patterns on \FX{} RDF sources, and contribute a theory and related algorithm for deciding whether a graph pattern can have a solution in a possible \FX{} RDF graph.

% Graph admissibility: Case generation and analysis by learning models \url{https://www.sciencedirect.com/science/article/abs/pii/S1877750324000747} (2024)
% They study the problem of deciding whether a graph is \textit{t-admissible}
% \textit{“A graph G is t-admissible if it has a spanning tree T such that any adjacent vertices in G are at a distance at most t in T, in which case T is called a tree t-spanner of G. We denote as the stretch index of G, or o(G), the smallest value t such that G is t-admissible. Despite having a polynomial decider for t = 2, the problem is NP-complete for t >= 4, while the t = 3 variant remains open after decades since its proposal.”}
% In my understanding, FX graphs are t-admissible by definition.
% Can be useful for formalising FX-admissibility.

% Echahed, Rachid, and Jean-Christophe Janodet. \textit{"Admissible Graph Rewriting and Narrowing."} In IJCSLP, pp. 325-342. 1998. 
% \textit{“Many different notations are used in the literature to investigate graph rewriting [8, 22]. The aim of this section is to recall briefly some key definitions in order to make easier the understanding of the paper.”} 
% Useful to learn how to specify graphs and rewriting strategies.

\subsection{Other approaches}
%Tools for direct data access or transformation.
Tools are available for automatically transforming data sources of several formats into RDF (Any23\footnote{\url{http://any23.apache.org/} (now retired).}, JSON2RDF\footnote{\url{https://github.com/AtomGraph/JSON2RDF}}, CSV2RDF\footnote{\url{http://clarkparsia.github.io/csv2rdf/}} to name a few).
While these  tools have a similar goal (i.e. enabling the user to access the content of a data source as if it was in RDF), the (meta)model used for generating the RDF data  highly depends on the input format, thus limiting the homogeneity of data generated from heterogeneous data formats.
In addition, none of those approaches are based on a common abstraction from heterogeneous formats.

A way to provide generalised data access is the list-of-lists approach. This is used by software libraries to onboard heterogeneus data structures, see, for example, Pandas\footnote{\url{https://hevodata.com/learn/pandas-load-json/}} or the Visidata tool\footnote{\url{https://www.visidata.org/}}. \FX{} follows a similar intuition, while targeting \textit{graphs} as the general meta-model, rather than tabular representations.

\FX{} use cases typically involve documents.
The approaches to querying documents have a tradition in computer science~\cite{clifton1995hyperfile,abiteboul1997querying,10.1007/978-1-4471-1525-0_25}, and many format-specific languages have been developed for that (XPath, JsonPath, and CSS selectors are the most popular).
\FX{} can be seen as a general approach to express queries over heterogeneus files.
Here, we focus on the satisfiability of basic graph patterns against \FX{} RDF graphs.
This analysis is the foundation for studying efficient execution strategies of SPARQL / \FX{} queries against serialized resources such as JSON or XML files.

% In summary, the objective of this work is to provide a foundational reference towards a theory of \FX{} data access and querying, focusing on core definitions and theorems, tackling the SPARQL graph pattern satisfiability problem in \FX{}.

%\todo[inline]{Clarify the use of math/pl notation}

\section{Preliminaries}\label{sec:preliminaries}We introduce the reader to the formalisms on which our theory is based, namely RDF, SPARQL and Predicate Logic.

\subsection{Predicate logic}
As the \fx{} model is defined in predicate logic, we provide the reader with a brief introduction to the main ingredients of this formalism.
However, providing a full description of the syntax and semantic of PL is not the purpose of this section.
For a complete definition of the syntax and semantic of PL we defer the reader to reference books, such as~\cite{Russell2020}.
Predicate Logic (PL) is the logic to speak about objects, which are the domain of discourse (universe $\mathcal{U}$).
PL is concerned about properties and relations about the objects of the universe which are expressed in form of predicates.
Intuitively, a PL formula is combination of PL terms (variables and constants) with logical operators $\neg, \wedge, \vee, \Rightarrow$ and quantifiers $\exists, \forall$.
An interpretation $\mathscr{I}$ of one or multiple PL formulas specifies what each predicate means, and the entities that can instantiate the variables. 
Given a set of predicates $P_1...P_n$ and function $f_1...f_m$ a PL interpretation is 
\begin{center}
    $\mathscr{I}=(\Delta^{\mathscr{I}},P_1^{\mathscr{I}}...P_n^{\mathscr{I}},f_1^{\mathscr{I}}...f_m^{\mathscr{I}})$
\end{center}
where $\Delta^{\mathscr{I}}\subseteq \mathcal{U}$ is the domain, $P_i^{\mathscr{I}}\subseteq \Delta^{\mathscr{I}} \times \cdots \times \Delta^{\mathscr{I}}$ (k-times where k is the arity of  $\mathcal{P}_i$) is the interpretation of the predicate $P_i$, $f_i^{\mathscr{I}}: \Delta^{\mathscr{I}} \times \cdots \times \Delta^{\mathscr{I}}\to\Delta^{\mathscr{I}}$ (k-times where k is the arity of  $f_i$) is the interpretation of the function $f_i$.

A PL open formula, also called \textit{query}, is a PL formula with free variables (i.e. not quantified).
Let \ax{Vars} be the set of variables, given an interpretation $\mathscr{I}$, an assignment is a function $\alpha: \tt{Vars} \to \Delta^{\mathscr{I}}$.

Given a formula $\phi$ with free variables $(x_1...x_n)$ and an interpretation $\mathscr{I}$, the answer to the query is the tuple $(a_1...a_n)$, where $a_i$ is an assignment of a free variable to an object of the domain of discourse $\Delta^{\mathscr{I}}$,   such that $\mathscr{I}, (a_1...a_n) \models \phi$.
We introduce $\mathcal{L}$ as the set of all possible queries, i.e. the query language, and the function \ax{eval} which given a query $\phi$ and an interpretation $\mathscr{I}$ returns the list of assignments that satisfy $\phi$ over $\mathscr{I}$.

\subsection{RDF and SPARQL}
% Given an \FX{} SPARQL query our objective is to decide if its basic graph patterns are satisfiable. % in order to streamline the query evaluation process.
% To this end, this Section identifies all possible RDF graph patterns. 
% Then, in Section~\ref{sec:theory}, we introduce the \FX{} model and we enumerate characterise the conditions under which a basic graph pattern is satisfiable on a \FX{} graph.

% \textbf{Notation.}
In this section, we use relational algebra notation to define the elements of the universe under consideration, indicated by $\mathcal{U}$, namely RDF graphs, SPARQL queries, and resource, data sources and values they contain that will be introduced in Section~\ref{sec:theory}.
We also use predicate logic to define the conditions that apply to the elements of the universe.

%\subsection{RDF and SPARQL}

We start from the basic notions of RDF~\cite{w3c:rdf11} and SPARQL~\cite{w3c:sparql11}.
We defer the reader to~\cite{hogan2021knowledge} and to the corresponding documentation for a complete description of these standards.
Formally, let $\mathcal{I} \subseteq \mathcal{U}$, $\mathcal{B} \subseteq \mathcal{U}$, and $\mathcal{L} \subseteq \mathcal{U}$ be infinite sets of IRIs, blank nodes, and literals.
The sets are assumed to be pairwise disjoint and we will collectively
refer to them  as \emph{RDF terms}. 
A tuple $(s, p, o) \in \mathcal{T}$ (with $\mathcal{T}$ is the set $ (\mathcal{I}\cup \mathcal{B})\times(\mathcal{I})\times(\mathcal{I}\cup \mathcal{B} \cup \mathcal{L})$) is called \emph{(RDF) triple} and we say $s$ is the \emph{subject} of the triple, $p$ the \emph{predicate}, and $o$ the \emph{object}. 
An \emph{RDF graph} is a set of RDF triples, whereas an \emph{RDF dataset} is a collection of named RDF graphs, each one identified by IRI, and a default RDF graph.
\textit{RDF graphs} can be expressed as a set of quads, i.e. elements of the set $\mathcal{I}\times\mathcal{T}$.

SPARQL is based on the idea of defining patterns to be matched against an input RDF dataset.
Formally, considering the set of variables $\mathcal{V} \subseteq \mathcal{U}$, disjoint from the previously defined $\mathcal{I}$, $\mathcal{B}$, and $\mathcal{L}$, a \textit{triple pattern} is a tuple of the form $(s, p, o) \in (\mathcal{I}\cup \mathcal{B}\cup \mathcal{V})\times(\mathcal{I}\cup \mathcal{V})\times(\mathcal{I}\cup \mathcal{B} \cup \mathcal{L}\cup \mathcal{V})$\footnote{We are aware that the BGP specification in~\cite{w3c:sparql11} includes literals in subject position. However, such patterns won't have a solution in RDF graphs, as defined in~\cite{w3c:rdf11}.}. 
All the possible SPARQL triple patterns are enumerated in Table~\ref{tab:triple_patterns}.
A \emph{basic graph pattern (BGP)} is a set of triple patterns.
%We introduce the function $BGP$ that given a query returns the triples contained in its BGP\todo[inline]{Consider the order of the triples?}

A SPARQL query $Q$ is  composed of the following components:
\begin{enumerate*}[label=(\roman*)]
    \item the query type (i.e. SELECT, ASK, DESCRIBE, CONSTRUCT); 
    \item the dataset clause;
    \item the graph pattern (recursively defined as being a BGP or the result of the composition of one or more graph patterns through one of several SPARQL operators that modify and combine the  obtained results);
    \item the solution modifiers (i.e. LIMIT, GROUP BY, OFFSET).
\end{enumerate*}

This work focuses solely on the graph pattern itself, since the other components can be evaluated as defined in SPARQL 1.1.
Consider a graph pattern consisting of two joined unordered triple patterns with a join condition defined on a single node. 
There are 6 possible joins matching nodes of the two triple patterns.
The joins are listed in Table~\ref{tab:sparql_joins}.

More complex join conditions involving multiple nodes can be derived by combining these joins.
Most common graph structures can be classified according to a set of graph \textit{motifs}~\cite{Asprino2023how,Salas2023}.%\todo[]{Where is the idea of motifs coming from? Can we cite anyone? Why these motifs?}
These are reported in Table~\ref{tab:bgp_motifs}.
For brevity, the motifs only contain nodes as variables, but these can be any valid combination of node types.

%\subsection{RDF Graphs}

%\subsection{RDF Basic Graph Pattern (BGP)}
%Triple, variables, costants, and the universal quantifier

%\subsection{Classes of RDF Graph Patterns}

%Joins. Enumerate the different types of joins, considering variables, constants (Literal or IRI), and the universal quantifier []

%Conjunctive queries:

%- Star-patterns

%- Snow-flakes

%- cycles

%- paths (Path graphs)

%- cliques

%\subsection{Satisfiability of RDF Basic Graph Patterns (BGP)}

\begin{table}[t]
    \centering
    \caption{The list of all possible SPARQL triple patterns.}
    \label{tab:triple_patterns}

    \begin{minipage}{.45\linewidth}
\begin{tabular}{cccc}
    \hline
       \textbf{Id}  & \textbf{Pattern} & \textbf{Example} \\\hline\hline

\VAR{}\VAR{}\VAR{} & \VAR{} \VAR{} \VAR{} & \var{s} \var{p} \var{o}  \\
\VAR{}\VAR{}\IRI{} & \VAR{} \VAR{} \IRI{} & \var{s} \var{p} \iri{o}  \\
\VAR{}\VAR{}\BN{} & \VAR{} \VAR{} \BN{} & \var{s} \var{p} \bn{o}  \\
\VAR{}\VAR{}\LIT{} & \VAR{} \VAR{} \LIT{} & \var{s} \var{p} \lit{o}  \\
\VAR{}\IRI{}\VAR{} & \VAR{} \IRI{} \VAR{} & \var{s} \iri{p} \var{o}  \\
\VAR{}\IRI{}\IRI{} & \VAR{} \IRI{} \IRI{} & \var{s} \iri{p} \iri{o}  \\
\VAR{}\IRI{}\BN{} & \VAR{} \IRI{} \BN{} & \var{s} \iri{p} \bn{o}  \\
\VAR{}\IRI{}\LIT{} & \VAR{} \IRI{} \LIT{} & \var{s} \iri{p} \lit{o}  \\
\IRI{}\VAR{}\VAR{} & \IRI{} \VAR{} \VAR{} & \iri{s} \var{p} \var{o}  \\
\IRI{}\VAR{}\IRI{} & \IRI{} \VAR{} \IRI{} & \iri{s} \var{p} \iri{o}  \\
\IRI{}\VAR{}\BN{} & \IRI{} \VAR{} \BN{} & \iri{s} \var{p} \bn{o}  \\
\IRI{}\VAR{}\LIT{} & \IRI{} \VAR{} \LIT{} & \iri{s} \var{p} \lit{o}  \\

\hline
    \end{tabular}
    \end{minipage}
    \begin{minipage}{.45\linewidth}
\begin{tabular}{cccc}
    \hline
       \textbf{Id}  & \textbf{Pattern} & \textbf{Example}  \\\hline\hline

\IRI{}\IRI{}\VAR{} & \IRI{} \IRI{} \VAR{} & \iri{s} \iri{p} \var{o}  \\
\IRI{}\IRI{}\IRI{} & \IRI{} \IRI{} \IRI{} & \iri{s} \iri{p} \iri{o}  \\
\IRI{}\IRI{}\BN{} & \IRI{} \IRI{} \BN{} & \iri{s} \iri{p} \bn{o}  \\
\IRI{}\IRI{}\LIT{} & \IRI{} \IRI{} \LIT{} & \iri{s} \iri{p} \lit{o}  \\
\BN{}\VAR{}\VAR{} & \BN{} \VAR{} \VAR{} & \bn{s} \var{p} \var{o}  \\
\BN{}\VAR{}\IRI{} & \BN{} \VAR{} \IRI{} & \bn{s} \var{p} \iri{o}  \\
\BN{}\VAR{}\BN{} & \BN{} \VAR{} \BN{} & \bn{s} \var{p} \bn{o}  \\
\BN{}\VAR{}\LIT{} & \BN{} \VAR{} \LIT{} & \bn{s} \var{p} \lit{o}  \\
\BN{}\IRI{}\VAR{} & \BN{} \IRI{} \VAR{} & \bn{s} \iri{p} \var{o}  \\
\BN{}\IRI{}\IRI{} & \BN{} \IRI{} \IRI{} & \bn{s} \iri{p} \iri{o}  \\
\BN{}\IRI{}\BN{} & \BN{} \IRI{} \BN{} & \bn{s} \iri{p} \bn{o}  \\
\BN{}\IRI{}\LIT{} & \BN{} \IRI{} \LIT{} & \bn{s} \iri{p} \lit{o}  \\

\hline
    \end{tabular}
    \end{minipage}
    
\end{table}

\begin{table}[t]
    \centering
    
    \caption{The list of all possible single-node joins involving two SPARQL triple patterns.}
    \label{tab:sparql_joins}
    \begin{tabular}{cccc}
        \hline
        \textbf{Id} & \textbf{Join} & \textbf{Example} \\\hline\hline

% \var{s1} \var{p1} \var{o1} . \var{s2} \var{p1} \var{o2}

$S \bowtie P$ & $\subject{1} = \predicate{2}$ &  \var{j} \var{p1} \var{o1} . \var{s2} \var{j} \var{o2}  \\
$P \bowtie O$ & $\predicate{1} = \object{2}$ &  \var{s1} \var{j} \var{o1} . \var{s2} \var{p1} \var{j}  \\
$S \bowtie S$ & $\subject{1} = \subject{2}$ &  \var{j} \var{p1} \var{o1} . \var{j} \var{p1} \var{o2} \\
$P \bowtie P$ & $\predicate{1} = \predicate{2}$ &  \var{s1} \var{j} \var{o1} . \var{s2} \var{j} \var{o2} \\
$S \bowtie O$ & $\subject{1} = \object{2}$ &   \var{j} \var{p1} \var{o1} . \var{s2} \var{p1} \var{j} \\
$O \bowtie O$ & $\object{1} = \object{2}$ &  \var{s1} \var{p1} \var{j} . \var{s2} \var{p1} \var{j} \\

        \hline
         
    \end{tabular}
\end{table}

\begin{table}[t]
    \centering
    
    \caption{A list of possible SPARQL BGP motifs.}
    \label{tab:bgp_motifs}
    \begin{tabular}{c||c|c|c|c}
        \hline
        \textbf{BGP Motif Name} & 
        \textbf{Path} & 
        \textbf{Cycle/Clique} &
        \textbf{Star} &
        \textbf{Snowflake} 
        \\

        \textbf{Motif} & 
        \mline{\tt{?s ?p ?o . ?o ?p1 ?o1}} & 
        \mline{\tt{?s ?p ?o . ?o ?p1 ?s}} &
        \mline{\tt{?s ?p ?o . ?s ?p1 ?o1 . ?s ?p2 ?o2 }} &
        \mline{\tt{?s ?p ?o . ?s ?p1 ?o1 . ?s ?p2 ?o2 . ?o ?p3 ?o3 . ?o ?p4 ?o4 }}
        
        \\
        \hline
    \end{tabular}
\end{table}

\section{A consolidated version of \FX}\label{sec:facadex}\FX{} is the meta-model that results from the \textit{abstraction} of all the basic data structures used to represent the source data formats and the \textit{combination} of these data structures into a unified model. 
In this section, we present a consolidated version of \FX{}. 
%and its mapping to RDF. 
%We then provide insights on the properties of \FX{} graphs.
%Furthermore, we list inferences and axioms derivable from the definition.
%Next, we define mappings to RDF.

\subsection{Resources and data sources}
Before detailing the \FX{}, it is useful to formalise elements of its application context, i.e. the fact that data is serviced and shaped within a digital artifact of sort.
We rely on the generic notion of \textit{resource} to refer to a digital artifact that wraps some data\footnote{To not be confused with the concept of Resource in RDFS (that is the most general type of RDF term)~\cite{w3c:rdfs11}}. 
\chadded{In the Example~\ref{running_example}, the Resource is the file identified by the URL \url{https://sparql-anything.cc/examples/simple.csv}, while the content of the CSV file is the DataSource.}
We say that a resource contains one or more \textit{data sources}, to refer to actual data contained in them. 
Intuitively, a CSV is a file (then, a resource) that includes some tabular data (hence, a single data source). 
A Spreadsheet is a file (then, a resource) containing many tabs, each one of them including tabular data (hence, possibly multiple data sources).

We introduce the unary predicates \texttt{Resource} and  \texttt{DataSource} and the binary predicate \texttt{includes}. 
These predicates must comply with the following conditions: 
\begin{enumerate*}[label=\textit{(\roman*)}]
    \item Resources includes at least one data source;
    %\item Data sources have a unique name (identifier);
    \item A data source must be included in one and only one resource. \chadded{See Listing~\ref{def:fx:resources}.}
\end{enumerate*}
\begin{figure}\begin{lstlisting}[style=logic,caption={Formal definition of  resources and data sources.},label={def:fx:resources}]
$\forall \tt{r,ds}. \tt{includes}(r,ds) \Rightarrow \tt{Resource}(r) \wedge \tt{DataSource}(ds)$|\label{def:fx:resources:1}|
$\forall \tt{r}. \exists \tt{ds}. \tt{Resource}(r) \Rightarrow \tt{includes}(r,ds)$|\label{def:fx:resources:2}|
$\forall\tt{ds}.\exists\tt{r}.\tt{DataSource(ds)} \Rightarrow \tt{includes(r,ds)}$|\label{def:fx:resources:3}|
$\forall\tt{ds,r_1,r_2}.\tt{includes(r_1, ds}) \wedge \tt{includes(r_2, ds}) \Rightarrow r_1 = r_2$|\label{def:fx:resources:4}|
\end{lstlisting}\end{figure}

%$\forall \tt{ds,n}.\tt{hasName}(ds,n)\Rightarrow \tt{DataSource}(ds) \wedge \tt{Name}(n)$
%$\forall \tt{ds}. \exists \tt{n}. \tt{DataSource}(ds) \Rightarrow \tt{hasName}(ds,n)$
%$\forall \tt{r,ds_1,ds_2,n}.\tt{includes(r, ds_1)} \wedge \tt{includes(r, ds_2)} \wedge \tt{hasName(ds_1, n)} \wedge\tt{hasName(ds_2, n)}$$ \Rightarrow \tt{ds_1 = ds_2}$
%$\forall\tt{ds,n_1,n_2}.\tt{hasName(ds, n_1}) \wedge \tt{hasName(ds, n_2}) \Rightarrow n_1 = n_2$

\subsection{\FX}\label{sec:fxmodeldefition}
We overview the basic data structures needed for representing the content of common resources including CSV, JSON, YAML, XML, HTML, Text, Markdown, Relational Databases, and Spreadsheets (e.g. XLS and XLSx)\footnote{See also~\url{https://sparql-anything.readthedocs.io/stable/\#supported-formats}.}.
The framework proposed in~\cite{asprino2023knowledge} represents any data source as a combination of only three knowledge representation primitives, namely: sets of key-data pairs (\textit{maps}), ordered sequence of elements (\textit{lists}), and unary predicates (\textit{types}).
In fact, this framework interprets any data source as a sequence of values, regardless of its format, and then, arranges the values of the sequence in multiple structures according to their serialisation format.
Crucially, all these structures can be captured using the three knowledge representation primitives mentioned above.
In addition, the framework introduces the notion of \textit{Container} as an abstraction of both lists and maps through the notion of \textit{Slots}.
A \textit{Slot} is an allotted place for an object, namely a primitive value or another container, \textit{``contained''} in a \textit{Container}.
\textit{Container} may have a type.
\FX{} is a representation of a data source in terms of  Containers, Types, Slots, and Values.
In this Section we provide a formal definition of the \FX{}.
\chadded{Figure~\ref{fig:fx_model_diagram} shows a diagrammatic representation of the \FX{} that provides an informal overview of the model.}

\begin{figure}[ht]
    \centering
    \includegraphics[width=0.7\linewidth]{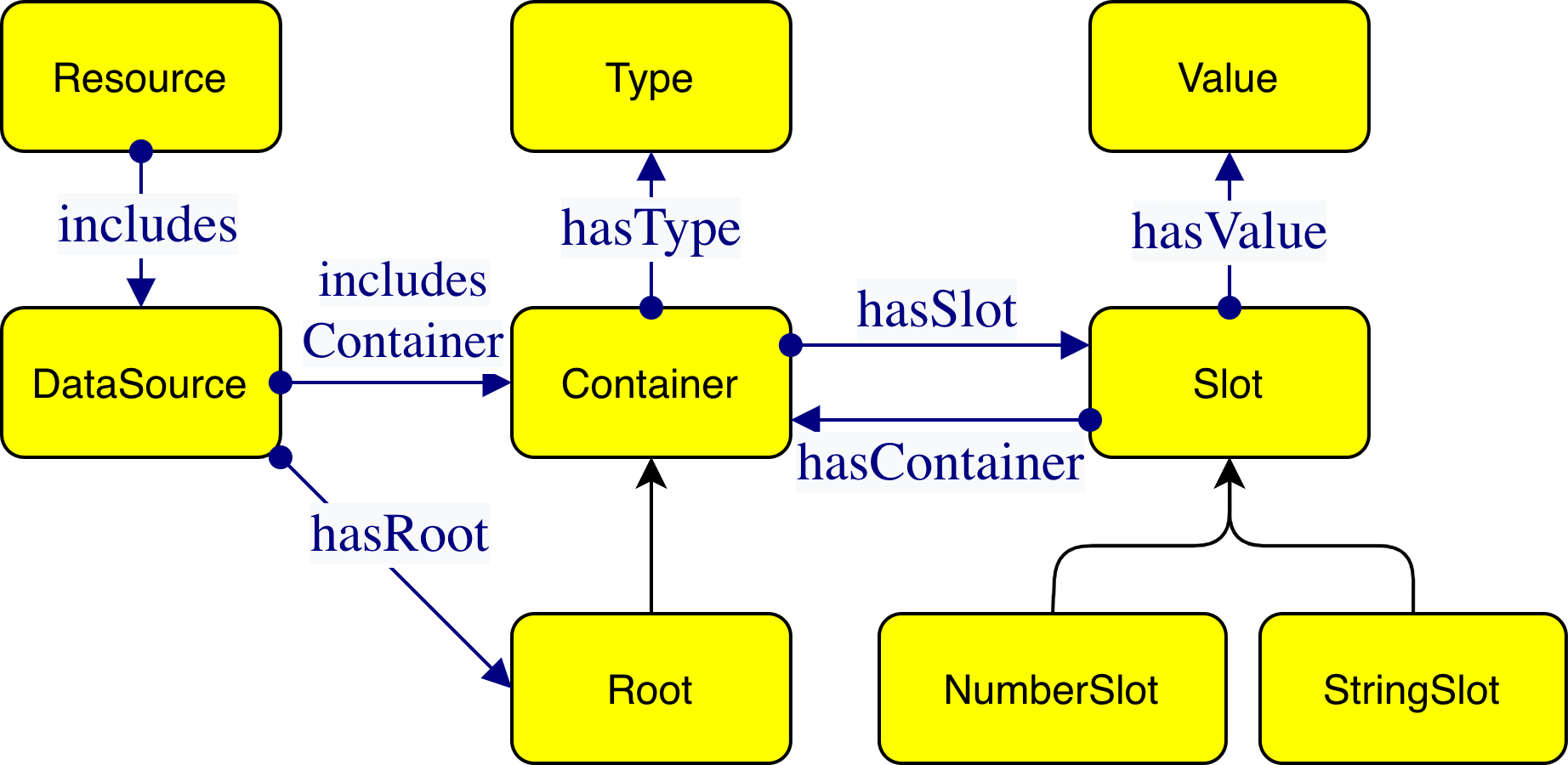}
    \caption{A diagrammatic representation of the \FX{} model. Unlabelled arcs represent subsumtion (e.g. Root is a type of Container).}
    \label{fig:fx_model_diagram}
\end{figure}

\begin{sloppypar}We introduce the unary predicates: \texttt{Container}, \texttt{Slot}, \texttt{StringSlot}, \texttt{NumberSlot} and \texttt{Value}.
We also introduce the binary predicates \texttt{hasSlot}, \texttt{hasValue}, \texttt{includesContainer} and \texttt{hasContainer}.
Containers are collections of unique key/value pairs (slots) where: the key of the slot is unique in the collection and can be either a number (i.e. \texttt{NumberSlot}) or a sequence of alphanumeric characters (i.e. \texttt{StringSlot});  the value can be either a primitive value or another container. \chadded{See Listing~\ref{def:fx:basic}}.\end{sloppypar}

\begin{figure}\begin{lstlisting}[style=logic,caption={Formal definition of containers, slots, and values.},label={def:fx:basic}]
$\tt{\forall c,s. hasSlot(c,s) \Rightarrow Container(c) \wedge Slot(s)}$|\label{def:fx:basic:1}|
$\forall \tt{s,v}. \tt{hasValue(s,v)}\Rightarrow \tt{Slot(s)} \wedge \tt{Value(v)}$|\label{def:fx:basic:2}|
$\forall \tt{s,c}. \tt{hasContainer(s,c)}\Rightarrow \tt{Slot(s)} \wedge \tt{Container(v)}$|\label{def:fx:basic:3}|
$\forall \tt{s.} \tt{NumberSlot(s)} \Rightarrow \tt{Slot(s)}$|\label{def:fx:basic:4}|
$\forall \tt{s.} \tt{StringSlot(s)} \Rightarrow \tt{Slot(s)}$|\label{def:fx:basic:5}|
$\nexists \tt{s.} \tt{StringSlot(s)} \wedge \tt{NumberSlot(s)}$|\label{def:fx:basic:6}|
$\forall \tt{s.} \tt{Slot(s)} \Rightarrow \tt{NumberSlot(s)}\vee\tt{StringSlot(s)}$|\label{def:fx:basic:7}|
$\tt{\forall ds,c. includesContainer(ds,c) \Rightarrow DataSource(ds) \wedge Container(c)}$|\label{def:fx:basic:8}|
$\tt{\forall c. \exists ds. Container(c) \Rightarrow includesContainer(ds,c)}$|\label{def:fx:basic:9}|
$\tt{\forall ds_1, ds_2, c. includesContainer(ds_1, c) \wedge includesContainer(ds_2, c) \Rightarrow ds_1 = ds_2 }$|\label{def:fx:basic:10}|
\end{lstlisting}\end{figure}

\paragraph*{}\textbf{Terminology.} We say that a container \ax{c} \textit{contains} a slot \ax{s} if \ax{(c,s)} is in the interpretation of the predicate \ax{hasSlot}. Moreover, we say that a slot \ax{s} \textit{holds} a value \ax{v} (or a container \ax{c}) if  \ax{(s,v)} (resp. \ax{(s,c)}) is in the interpretation of \ax{hasValue} (resp. \ax{hasContainer}). In this case we can also say that \ax{v} (resp. \ax{c}) is assigned to a slot. 
Finally, we say that a container \ax{c} \textit{recursively contains} a container \ax{c'} if there exists a sequence \ax{s_0,...,s_n} and containers \ax{c_1,...,c_n} such that \ax{(c,s_0),(c_1, s_1)...(c_n,s_n)} is in the interpretation of \ax{hasSlot} and \ax{(s_0,c_1),(s_1, c_2),...,(s_n,c')} is in the interpretation of \ax{hasContainer}.

\chadded{A formal definition of  the properties of slots is provided in Listing~\ref{def:fx:slots}.}
Values must be assigned to a slot (cf. \ref{def:fx:slots}.\ref{lst:fx:slots:1}). %Not single!
Slots must be contained by a single container (cf. \ref{def:fx:slots}.\ref{lst:fx:slots:3}, \ref{def:fx:slots}.\ref{lst:fx:slots:4}).
Slots must hold either a single container or a single value (cf.  \ref{def:fx:slots}.\ref{lst:fx:slots:5}, \ref{def:fx:slots}.\ref{lst:fx:slots:6}, \ref{def:fx:slots}.\ref{lst:fx:slots:2}, \ref{def:fx:slots}.\ref{lst:fx:slots:8}).
There cannot exists a container c and a slot s such that c contains s and s holds c.  (cf.  \ref{def:fx:slots}.\ref{lst:fx:slots:7}).  
Each container cannot recursively contain itself through a sequence of slots and containers (cf.  \ref{def:fx:slots}.\ref{lst:fx:slots:9}).
\chadded{Each container can only be held by a single slot (it is not possible for two slots to hold the same container)(cf.  \ref{def:fx:slots}.\ref{lst:fx:slots:10})}.
%I.e. a container is a collection of slots holding either other containers or values.
%A slot cannot have assigned a container if that container owns the slot.
\begin{figure}\begin{lstlisting}[style=logic,caption={Formal definition of  the properties of slots.},label={def:fx:slots}]
$\tt{\forall v \exists s. Value(v) \Rightarrow hasValue(s,v)}$|\label{lst:fx:slots:1}|
$\tt{\forall s \exists c. Slot(s) \Rightarrow hasSlot(c,s)}$|\label{lst:fx:slots:3}|
$\tt{\forall c_1, c_2, s. hasSlot(c_1,s) \wedge hasSlot(c_2, s) \Rightarrow c_1 = c_2}$|\label{lst:fx:slots:4}|
$\tt{\forall s \exists x. Slot(s) \Rightarrow hasContainer(s,x) \vee hasValue(s,x)}$|\label{lst:fx:slots:5}|
$\tt{\nexists s, x. hasContainer(s,x) \wedge hasValue(s,x)}$|\label{lst:fx:slots:6}|
$\tt{\forall c_1, c_2, s. hasContainer(s, c_1) \wedge hasContainer(s, c_2) \Rightarrow c_1 = c_2}$|\label{lst:fx:slots:8}|
$\tt{\forall v_1, v_2, s. hasValue(s, v_1) \wedge hasValue(s, v_2) \Rightarrow v_1 = v_2}$|\label{lst:fx:slots:2}|
$\tt{\nexists c, s. hasSlot(c, s) \wedge hasContainer(s, c)}$|\label{lst:fx:slots:7}|
$\tt{\forall c \nexists s_1...s_n, c_1...c_{n}.hasSlot(c, s_1) \wedge hasContainer(s_1, c_1)\wedge\cdots}$$\cdots\tt{\wedge hasSlot(c_{n-1}, s_{n}) \wedge hasContainer(s_n, c)}  $|\label{lst:fx:slots:9}|
$\tt{\forall c, s_1, s_2. hasSlot(s_1,c) \wedge hasSlot(s_2, c) \Rightarrow s_1 = s_2}$|\label{lst:fx:slots:10}|
\end{lstlisting}\end{figure}

Containers can have types (cf. \ref{def:fx:types}.\ref{def:fx:types:1}). Every type must be assigned to a container (cf. \ref{def:fx:types}.\ref{def:fx:types:2}). To this end we introduce the unary predicate \texttt{Type} and the binary predicate \texttt{hasType}. 
\chadded{A formal definition of types is provided in Listing~\ref{def:fx:types}.}
~\begin{lstlisting}[style=logic,caption={Formal definition of types.},label={def:fx:types}]
$\tt{\forall c,t. hasType(c, t) \Rightarrow Container(c) \wedge Type(t)}$|\label{def:fx:types:1}|
$\tt{\forall t. \exists c. Type(t) \Rightarrow hasType(c,t)}$|\label{def:fx:types:2}|
\end{lstlisting}

 The unary predicates Container, Slot, Value, and Type are pairwise disjoint.
\chadded{Disjointness axioms are provided in Listing~\ref{def:fx:disjointness}.}

\noindent\begin{minipage}{.47\textwidth}
        \centering
        \begin{lstlisting}[style=logic,caption={Definition of the disjointness relation.},label={def:fx:disjointness}]
$\tt{\forall x. Container(x) \Rightarrow \neg Slot(x)}$ |\label{def:fx:disjointness:1}|
$\tt{\forall x. Container(x) \Rightarrow \neg Value(x)}$ |\label{def:fx:disjointness:2}|
$\tt{\forall x. Container(x) \Rightarrow \neg Type(x)}$ |\label{def:fx:disjointness:3}|
$\tt{\forall x. Slot(x) \Rightarrow \neg Container(x)}$ |\label{def:fx:disjointness:4}|
$\tt{\forall x. Slot(x) \Rightarrow \neg Value(x)}$ |\label{def:fx:disjointness:5}|
$\tt{\forall x. Slot(x) \Rightarrow \neg Type(x)}$ |\label{def:fx:disjointness:6}|
\end{lstlisting}
    \end{minipage}%
    \begin{minipage}{0.5\textwidth}
        \centering
        \vspace{0.65cm}
        \begin{lstlisting}[style=logic,firstnumber=7]
$\tt{\forall x. Value(x) \Rightarrow \neg Container(x)}$ |\label{def:fx:disjointness:7}|
$\tt{\forall x. Value(x) \Rightarrow \neg Slot(x)}$ |\label{def:fx:disjointness:8}|
$\tt{\forall x. Value(x) \Rightarrow \neg Type(x)}$ |\label{def:fx:disjointness:9}|
$\tt{\forall x. Type(x) \Rightarrow \neg Container(x)}$ |\label{def:fx:disjointness:10}|
$\tt{\forall x. Type(x) \Rightarrow \neg Slot(x)}$ |\label{def:fx:disjointness:11}|
$\tt{\forall x. Type(x) \Rightarrow \neg Value(x)}$ |\label{def:fx:disjointness:12}|
\end{lstlisting}
    \end{minipage}

%$\forall(x)$  $\Rightarrow$ Container(x) $\oplus$ Slot(x)  $\oplus$ Value(x) $\oplus$ Type(x)

For each data source, there exists one and only one Root container (cf. \ref{def:fx:root}.\ref{def:fx:root:1}, \ref{def:fx:root}.\ref{def:fx:root:7}, \ref{def:fx:root}.\ref{def:fx:root:3}). 
The Root container cannot be assigned to a slot (cf.~\ref{def:fx:root}.\ref{def:fx:root:6}).
Containers that are not root must be assigned to a slot (cf.~\ref{def:fx:root}.\ref{def:fx:root:hascontainer}).
To identify the root container we introduce the unary predicate \texttt{Root} as specialisation of Container (\ref{def:fx:root}.\ref{def:fx:root:5}) and the binary predicate \texttt{hasRoot} that connects a data source to its root container (cf.~\ref{def:fx:root}.\ref{def:fx:root:2}, \ref{def:fx:root}.\ref{def:fx:root:4}).
\chadded{A formal definition of the Root container is provided in Listing~\ref{def:fx:root}.}
%exists that holds all data, which we name the Root container.
\begin{figure}\begin{lstlisting}[style=logic,caption={Formal defition of Root container.},label={def:fx:root}]
$\tt{\forall ds. \exists r. DataSource(ds) \Rightarrow hasRoot(ds,r)}$|\label{def:fx:root:1}|
$\tt{\forall ds, r. hasRoot(ds,r) \Rightarrow DataSource(ds) \wedge Root(r)}$|\label{def:fx:root:2}|
$\tt{\forall r. \exists ds. Root(r) \Rightarrow hasRoot(ds, r)}$|\label{def:fx:root:atleastoneroot}\label{def:fx:root:3}|
$\tt{\forall ds,r. hasRoot(ds, r) \Rightarrow includesContainer(ds,r)}$|\label{def:fx:root:4}|
$\tt{\forall x. Root(x) \Rightarrow Container(x)}$|\label{def:fx:root:5}|
$\tt{\nexists r,s. Root(r) \wedge hasContainer(s,r) }$|\label{def:fx:root:6}|
$\tt{\forall ds, r_1, r_2. hasRoot(ds, r_1) \wedge hasRoot(ds, r_2) \Rightarrow r_1 = r_2}$ |\label{def:fx:root:singleroot}\label{def:fx:root:7}|
$\tt{\forall c \exists s. Container(c) \wedge \neg Root(c) \Rightarrow  hasContainer(s, c) }$ |\label{def:fx:root:hascontainer}\label{def:fx:root:8}|
\end{lstlisting}
\end{figure}
\chadded{In Example~\ref{running_example},  a container is generated for each row of the file. These containers are held by slots in a root container, which represents the entire table. A value is generated for each cell and associated with the corresponding row via a slot. An excerpt of this is shown in the Example~\ref{tab:runingexample:1}.}

\subsection{Properties of the \FX{}}

This section outlines some properties of the \FX{} model that will be used later to characterise the RDF knowledge graph constructed from the instances of the model.

\begin{proposition}[Connectedness of containers]\label{proposition:connectiveness}
    Every non-root container is recursively contained by  the root container.
\end{proposition}
\begin{proof}
    For each non-root container \ax{c}, either
    \begin{enumerate*}[label=\textit{(\roman*)}]
        \item \label{theorem:connectiveness:case1} \ax{c} is not assigned to any slot;
        \item \label{theorem:connectiveness:case2} or \ax{c} is assigned to a slot.    
    \end{enumerate*}
    We can exclude \ref{theorem:connectiveness:case1} as by definition (\ref{def:fx:root}.\ref{def:fx:root:hascontainer}\footnote{With this notation we refer to the line \ref{def:fx:root:hascontainer} of the definition \ref{def:fx:root}.}) a non-root container must be assigned to a slot. 
    As for \ref{theorem:connectiveness:case2}, we show that for each non-root container \ax{c} there always exist a sequence of slots \ax{s_r,s_1,...,s_n} and containers \ax{r,c_1,...,c_n} such that \ax{(r,s_r),(c_1, s_1)...(c_n,s_n)} is in the interpretation of \ax{hasSlot} and \ax{(s_r,c_1),(s_1, c_2),...,(s_n,c)} is in the interpretation of \ax{hasContainer}, i.e. \ax{r} recursively contains \ax{c}. The following conditions imply that each non-root container \ax{c} must be in a sequence as above:
    \begin{enumerate*}[label=\textit{(\alph*)}]
        \item\label{theorem:connectiveness:casea} each non-root container must be assigned to a slot (\ref{def:fx:root}.\ref{def:fx:root:hascontainer});
        \item\label{theorem:connectiveness:caseb} each slot is always contained by a container (\ref{def:fx:slots}.\ref{lst:fx:slots:3});
        \item\label{theorem:connectiveness:casec} containers  cannot be recursively contain themself (\ref{def:fx:slots}.\ref{lst:fx:slots:9}).
    \end{enumerate*}
    It is worth noting that \ref{theorem:connectiveness:casea} and \ref{theorem:connectiveness:caseb} imply that for each non-root there is always a sequence of containers like that above which it belongs to.  \ref{theorem:connectiveness:casec} implies that the sequence does not form a cycle. 
    Therefore, the first container in the sequence must be the root container, since it cannot be assigned to any slot.
\end{proof}

\begin{proposition}\label{proposition:singlepath}
    For each data source ds, there is one and only one sequence of containers and slots that connects the root of ds to its containers.
\end{proposition}
\begin{proof}
    This follows from the following facts: 
    \begin{enumerate*}[label=\textit{(\roman*)}]
    
        \item\label{theorem:singlepath:i} There is always a single root for each data source (\ref{def:fx:root}.\ref{def:fx:root:singleroot});

        \item\label{theorem:singlepath:ii} Every non-root container is recursively contained by  the root container (\ref{proposition:connectiveness});

        \item\label{theorem:singlepath:iii} Each slot is assigned to one and only one container (\ref{def:fx:slots}.\ref{lst:fx:slots:4});

        \item\label{theorem:singlepath:iv} Each non-root container is assigned to one and only one slot (\ref{def:fx:slots}.\ref{lst:fx:slots:8}).
        
    \end{enumerate*}
    Therefore, \ref{theorem:singlepath:ii} implies that there is always a sequence of slots and containers that connects a container to a root, \ref{theorem:singlepath:ii} that there is always a single root,  \ref{theorem:singlepath:ii} and  \ref{theorem:singlepath:iv} guarantee that the sequence is one and only one (it is convenient to pass through the sequence backward, the last container of the sequence can be assigned to one and only one slot, which in turn can be assigned to one and only container and so on).
\end{proof}

%\begin{theorem}[Ancestors of sibling non-root containers]\label{theorem:siblingcontainers}
    %For any pair of sibling non-root containers (i.e. containers assigned to two slots within the same container), the set of ancestors  of those containers (i.e. the containers that recursively contain them) is the same.
%\end{theorem}
%\begin{proof}
%We observe that:
%    \begin{enumerate*}[label=\textit{(\roman*)}]
        
%        \item\label{theorem:siblingcontainers:ii} Each non-root container is assigned to one and only one slot (see \ref{def:fx:slots}.\ref{lst:fx:slots:8});

%        \item\label{theorem:siblingcontainers:i} Each slot is assigned to one and only one container (see \ref{def:fx:slots}.\ref{lst:fx:slots:4}).

%    \end{enumerate*}

%    Therefore, the backward path from the any pair of non-root containers to the root must go through the same set of containers.
%\end{proof}

Now that we have described the \FX{} model and demonstrated its properties, we introduce a function $fx$, called \textit{fa\c{c}adify}.
\chadded{This function is required to convert a resource into an instance of the \FX{} model.}
Let $\mathfrak{I}$ be the set all possible interpretations and $\mathscr{R}$ the set of all possible resources, the function $fx$, called \textit{fa\c{c}adify}, given a resource $r$ returns an interpretation $\mathscr{I}$ that complies with the \FX{} model, i.e. $fx: \mathscr{R} \to \mathfrak{I}$.

Summary of changes from~\cite{asprino2023knowledge}:

\begin{itemize}
    \item We remove the predicates Key(k), hasKey(s,k), StringKey(k), and NumberKey(k), i.e. collapse them with the definition of slots. We add StringSlot and NumberSlot as a replacement.
    Without loss of generality, we imagine that the key of the slot is included in the identifier of the individual.
    \item We rename Class into Type.
    \item We solve an ambiguity regarding the distinction between data sources and root containers, that are now distinct. We add the predicate hasRoot(ds, r).
    \item We remove Name and hasName that were used to assign a name to a DataSource (we assume that the identifier of the data sources are unique).
    \item We clarify disjointness and connectedness.
    \item We provide formal definitions for the predicates DataSource, Resource and for integrity constraints that were previously only defined informally.
\end{itemize}

\begin{ex}
Table~\ref{tab:runingexample:1} show an example of usage of the introduced predicates for interpreting a CSV file.
\begin{table}[h!]
    \centering
    \caption{An example of usage of the introduced predicates for interpreting a CSV file\footnotemark.}\label{tab:runingexample:1}
    \begin{tabular}{p{0.2\textwidth}p{0.78\textwidth}}
        \hline
         \multicolumn{1}{c}{\textbf{Source: \texttt{ex.csv}}} & \multicolumn{1}{c}{\textbf{\FX{} Interpretation}}  \\\hline
         \mline{name,surname\\Laura,Grey\\Craig,Johnson}& 
         \mline{
         \texttt{includes(file:///ex.csv,ex.csv),} \\
         \texttt{hasRoot(ex.csv,root)},\texttt{hasSlot(root,rs.1)},\\
\texttt{hasSlot(root,rs.2)},\texttt{hasSlot(root,rs.3),} \\
\texttt{NumberSlot(rs.1)},\texttt{NumberSlot(rs.2)},\texttt{NumberSlot(rs.3)}, \\
\texttt{hasContainer(rs.1,row:1)},\texttt{hasContainer(rs.2,row.2)} \\
\texttt{hasContainer(rs.3,row.3),} \\
\texttt{hasSlot(row.1,rsr:11)},\texttt{hasSlot(row.1,rsr.12),} \\
\texttt{hasSlot(row.2,rsr.21)},\texttt{hasSlot(row:2,rsr.22),} \\
\texttt{hasSlot(row.3,rsr.31)},\texttt{hasSlot(row.3,rsr.32),} \\
\texttt{hasValue(rsr.11,"name")},\texttt{hasValue(rsr.12,"surname"),} \\
\texttt{hasValue(rsr.21,"Laura")},\texttt{hasValue(rsr.22,"Grey"),} \\
\texttt{hasValue(rsr.31,"Craig")},\texttt{hasValue(rsr.32,"Jonhnson")} \\
\texttt{NumberSlot(rsr.11)},\texttt{NumberSlot(rsr.12)},\texttt{NumberSlot(rsr.21)}, \\
\texttt{NumberSlot(rsr.22)},\texttt{NumberSlot(rsr.31)},\texttt{NumberSlot(rsr.32)}
         }\\\hline
    \end{tabular}
\end{table}
\end{ex}
\footnotetext{The example shows the CSV as a single container including a list of lists. SPARQL Anything allows to interpret the first line as a schema and generate an alternative representation, using column names as \texttt{StringSlot}s. However, in this work, we do not consider format-specific configurations. More details at:~\url{https://sparql-anything.readthedocs.io/stable/formats/CSV/}.}
\begin{ex}
Table~\ref{tab:runingexample:2} show an example of usage of the introduced predicates for interpreting a JSON document.
\begin{table}[h!]
    \centering
    \caption{An example of usage of the introduced predicates for interpreting a JSON file}
    \label{tab:runingexample:2}
    \begin{tabular}{p{0.2\textwidth}p{0.78\textwidth}}
        \hline
         \multicolumn{1}{c}{\textbf{Source: \texttt{ex.json}}} & \multicolumn{1}{c}{\textbf{\FX{} Interpretation}}  \\\hline
         \mline{\{\\"a":1,\\"b":[1,2,3]\\\}}& 
         \mline{
         \texttt{includes(file:///ex.json,ex.json),} \\
         \texttt{hasRoot(ex.json,root)},\\
\texttt{hasSlot(root,rs.a)},\texttt{StringSlot(rs.a)},\texttt{hasValue(rs.a,1)},\\
\texttt{hasSlot(root,rs.b)},\texttt{StringSlot(rs.b)},\texttt{hasContainer(rs.b,rsb.c)},\\
\texttt{hasSlot(rsb.c,rsbc.1)},\texttt{hasSlot(rsb.c,rsbc.2)},\texttt{hasSlot(rsb.c,rsbc.3)},\\
\texttt{NumberSlot(rsbc.1)},\texttt{NumberSlot(rsbc.2)},\texttt{NumberSlot(rsbc.3)},\\
\texttt{hasValue(rsbc.1,1)},\texttt{hasValue(rsbc.2,2)},\texttt{hasValue(rsbc.3,3)}\\
         }\\\hline
    \end{tabular}
\end{table}
\end{ex}

\begin{ex}
Table~\ref{tab:runingexample:3} show an example of usage of the introduced predicates for interpreting a XML file.
\begin{table}[h!]
    \centering
    \caption{An example of usage of the introduced predicates for interpreting a XML file}
    \label{tab:runingexample:3}
    \begin{tabular}{p{0.4\textwidth}p{0.5\textwidth}}
        \hline
         \multicolumn{1}{c}{\textbf{Source: \texttt{ex.xml}}} & \multicolumn{1}{c}{\textbf{\FX{} Interpretation}}  \\\hline
         \mline{<TEAM~name="Chicago~Bulls">\\~~~~<PLAYER~name="Micheal~Jordan"/>\\</TEAM>}& 
         \mline{
         \texttt{includes(file:///ex.xml,ex.xml),} \\
         \texttt{hasRoot(ex.xml,root)},\texttt{hasType(root,TEAM)},\\
         \texttt{hasSlot(root,r.name),StringSlot(rname)},\\
         \texttt{hasValue(rname,"Chicago~Bulls")},\\        \texttt{hasSlot(root,r.1),NumberSlot(r.1),}\\
         \texttt{hasContainer(r.1,r1c),}\texttt{hasType(r1c,PLAYER),}
         \texttt{hasSlot(r1c,r1c.name),StringSlot(r1c.name)},\\\texttt{hasValue(r1c.name,"Micheal~Jordan")}\\
         }\\\hline
    \end{tabular}
\end{table}
\end{ex}
\section{\FX{}~RDF graphs}
\subsection{Mapping \FX{} to RDF}\label{sec:mapping}
We define a mapping to RDF as a function $\mathcal{F}: \mathcal{L} \times \mathfrak{I} \to 2^{\mathcal{I} \times \mathcal{T}}$ which takes as input a query $q$ over  \FX{}  and an interpretation $\mathcal{I}$ of the model, it evaluates $q$ over $\mathcal{I}$ and it returns a set of quads\footnote{$2^{A}$ denotes the powerset of a set $A$, i.e. the set of all possible subsets of $A$.}.
The function $\mathcal{F}$ uses the assignments of the free variables of $q$ to instantiate the quads.
The input of $\mathcal{F}$ is specified as a PL formula and the output is specified as a RDF graph template that include the free variable to instantiate the RDF terms.
We use the symbol $\rightleftharpoons$ to indicate an input-output pair of $\mathcal{F}$.
We note that such mappings are bidirectional and conversion is loss-less, meaning that one can go back and forth without loosing information.

It is worth noting that even though the purpose of mapping is different, formalising the mapping as a function that associates queries and data sources with RDF graphs is compatible with the framework proposed in \cite{oo2025algebraic} which associates queries with the positions of graphs to be generated.
However, demonstrating the compatibility of the two theories is beyond the scope of this article.

Before providing the details of the function $\mathcal{F}$ we introduce elements that will be used in the transformation.
We define two namespaces and preferred prefix names: \texttt{fx} and \texttt{xyz} \chadded{(see Listing~\ref{def:fx:namespaces})}. 

\begin{figure}\begin{lstlisting}[style=rdf,caption={Definition of \FX{} namespaces.},label={def:fx:namespaces}]
prefix fx: <http://sparql.xyz/facade-x/ns/> 
prefix xyz: <http://sparql.xyz/facade-x/data/>
\end{lstlisting}\vspace{-.5cm}
\end{figure}
%\begin{definition}[RDF construction functions]
%We use the predicate Namespace( ) to derive a namespace from an RDF IRI.
%We add STR( ) as the SPARQL function that converts an IRI to the corresponding string.
%We add the predicates \textit{TypeProperty} for all instances of rdf:type\todo{rdf:type is supposed to be an individual}, and \textit{ContainerMembershipProperty} to represent the set of IRIs with prefix rdf:\_ such that the remaining part is a natural number ($N$).

%~\begin{lstlisting}[style=logic,escapechar=|]
%$\forall$(s,p,o).Triple(s,p,o) $\Rightarrow$ Subject(s) $\wedge$ Property(p) $\wedge$ Object(o)
%$\forall(s)$.Subject(s) $\Rightarrow$ IRI(s) $\oplus$ BNode(s)
%$\forall(p)$.Property(p) $\Rightarrow$ IRI(s) 
%$\forall$(x).IRI(x) $\wedge$ STR(x) = rdf: $\cdot$ _ $\cdot$ n $|$ n $\in$ $N$ $\Rightarrow$ ContainerMembershipProperty(x)
%$\forall$(x).IRI(x) $\wedge$ x = rdf:type $\Rightarrow$ TypeProperty(x)
%$\forall(s)$.Object(s) $\Rightarrow$ IRI(s) $\oplus$ %BNode(s) $\oplus$ Literal(s)
%\end{lstlisting}    
%\end{definition}
%\todo[inline]{Double-check we don't confuse the predicate Root with the RDF term fx:root -- when we refer to the second, we should either use fx:root or the singleton predicate FXRoot (preferrable, because it is an annotation and we cannot infer equivalence in rules (but state that a variable can have role FXRoot}

We define the mappings \chadded{(see Listing~\ref{def:fx:mappings})} between \FX{} terms and RDF. 
These mappings are the input-output (i.e. query, RDF template) pairs of $\mathcal{F}$.
We use the functions \texttt{IRI}, \texttt{Entity} and \texttt{Literal} to construct RDF terms on the basis of the arguments passed to the function.
Depending from the configuration given by the KG engineer, the \texttt{Entity} constructs either an IRI or a Blank Node.
Moreover, we introduce the function \texttt{ContainerMembershipProperty} that given  $n\in \mathbb{N}^+$ constructs the IRI of the corresponding container membership property, i.e. \texttt{rdf:\_n}.
We introduce the primitive IRI \texttt{fx:root} to represent the unary predicate \texttt{Root}.
Without loss of generality, we assume that
\begin{enumerate*}[label=\textit{(\roman*)}]
    \item \texttt{ContainerMembershipProperty} can extract the cardinal order of the slot from its identifier;
    \item  \texttt{IRI} extract the name of the slot from its identifier.
\end{enumerate*}

%|Root(c) $\rightleftharpoons$ Triple(c, p, o) $|$ p = rdf:type $\wedge$ o = fx:root|
%StringSlot(s) $\wedge$ hasSlot(c1,s) $\wedge$ hasContainer(s, c2) $\rightleftharpoons$ Triple(c1, s, c2) 
%NumberSlot(s) $\wedge$ hasSlot(c1,s) $\wedge$ hasContainer(s, c2) $\rightleftharpoons$ Triple(c1, s, c2) $|$ ContainerMembershipProperty(s) $\wedge$ $\lnot$ Literal(o)
%StringSlot(s) $\wedge$ hasSlot(c,s) $\wedge$ hasValue(s,v) $\rightleftharpoons$  Triple(c1, s, v) $|$ Literal(v)
%NumberSlot(s) $\wedge$ hasSlot(c,s) $\wedge$ hasValue(s,v) $\rightleftharpoons$ Triple(c, s, v) $|$ Literal(v) $\wedge$ ContainerMembershipProperty(s)
%hasType(c,t) $\rightleftharpoons$ Triple(c, p, t) $|$ p = rdf:type
%|\label{def:fx:mappings:ds}|DataSource(ds) $\wedge$ hasRoot(ds, r) $\rightleftharpoons$ Quad(ds, r, p, t) $|$ p = rdf:type $\wedge$ t = fx:root

\begin{figure}\begin{lstlisting}[style=rdf,caption={Defintion of mappings of the \FX{} to RDF.},label={def:fx:mappings},escapechar=|]
$\tt{  hasRoot(ds,c)  \rightleftharpoons \{(IRI(ds), Entity(ds, c), \text{rdf:type}, \text{fx:root})  \}}$|\label{def:fx:mappings:1}|
$\tt{  includesContainer(ds,c_1) \wedge includesContainer(ds,c_2) }$$\tt{ \wedge StringSlot(s) \wedge hasSlot(c_1,s) \wedge hasContainer(s, c_2) \rightleftharpoons} $$\tt{\{(IRI(ds), Entity(ds, c_1), IRI(\text{xyz:}, s), Entity(ds, c_2))\}}$|\label{def:fx:mappings:2}|
$\tt{   includesContainer(ds,c_1) \wedge includesContainer(ds,c_2) }$$\tt{ \wedge NumberSlot(s) \wedge hasSlot(c_1,s) \wedge hasContainer(s, c_2) \rightleftharpoons} $$\tt{\{(IRI(ds), Entity(ds, c_1), ContainerMembershipProperty(s), Entity(ds, c_2))\}}$|\label{def:fx:mappings:3}|
$\tt{  includesContainer(ds,c)}$$\tt{ \wedge NumberSlot(s) \wedge hasSlot(c,s) \wedge hasValue(s, v) \rightleftharpoons} $$\tt{\{(IRI(ds), Entity(ds, c), ContainerMembershipProperty(s), Literal(v))\}}$|\label{def:fx:mappings:4}|
$\tt{  includesContainer(ds,c)}$$\tt{ \wedge StringSlot(s) \wedge hasSlot(c,s) \wedge hasValue(s, v) \rightleftharpoons} $$\tt{\{(IRI(ds), Entity(ds, c), IRI(\text{xyz:}, s), Literal(v))\}}$|\label{def:fx:mappings:5}|
$\tt{ \wedge includesContainer(ds,c)}$$\tt{ \wedge hasType(c,t) \rightleftharpoons} $$\tt{\{(IRI(ds), Entity(ds, c), \text{rdf:type}, IRI(\text{xyz:}, t))\}}$|\label{def:fx:mappings:6}|
\end{lstlisting}\vspace{-.5cm}
\end{figure}

\begin{corollary}\label{theorem:bidirectionalmappings}
    The mapping rules of  Definition~\ref{def:fx:mappings} are bidirectional.
\end{corollary}

\begin{proof}
    This is can be demonstrated by showing that the mapping of Definition~\ref{def:fx:mappings} is a bijective, i.e. a one-to-one correspondence of distinct elements.
    By definition the mapping rules define a correspondence of queries and generated quads.
    Therefore, we have only to show that the corresponding elements are distinct.
    Firstly, we note that queries and generated quads are distinct entities (i.e. it is not possible to convert one into the other using implications).
    As queries, we note that in each query (implicitly) involves distinct elements of the model: \ref{def:fx:mappings:1}: Root; \ref{def:fx:mappings:2}: Container and StringSlot; \ref{def:fx:mappings:3}: Container and NumberSlot;
    \ref{def:fx:mappings:4}: Value and StringSlot; \ref{def:fx:mappings:5}: Value and NumberSlot; and
    \ref{def:fx:mappings:6}: Type.
    In a similar fashion, this is also the case for the generated quads.
    While the first two elements are always an IRI derived from the DataSource and a Container, the property and the object are distinct: \ref{def:fx:mappings:1}: \texttt{rdf:type}, \texttt{fx:root}; \ref{def:fx:mappings:2}:  \texttt{xyz:} property and Entity; \ref{def:fx:mappings:3}: container membership property and Entity; 
    \ref{def:fx:mappings:4}: container membership property and Literal;  \ref{def:fx:mappings:3}: \texttt{xyz:} property and Literal; and
    \ref{def:fx:mappings:6}: \texttt{rdf:type} and \texttt{xyz:} IRI.
\end{proof}

\begin{figure}
\begin{lstlisting}[caption={The result of the application of the mapping to the input of Example~\ref{running_example}.},label={ex:running_example_rdf}]
PREFIX fx:     <http://sparql.xyz/facade-x/ns/>
PREFIX rdf:    <http://www.w3.org/1999/02/22-rdf-syntax-ns#>

[ rdf:type  fx:root;
  rdf:_1    [ rdf:_1  "email"; rdf:_2  "name"; rdf:_3  "surname" ];
  rdf:_2    [ rdf:_1  "laura@example.com"; rdf:_2  "Laura";
              rdf:_3  "Grey" ];
  rdf:_3    [ rdf:_1  "craig@example.com"; rdf:_2  "Craig";
              rdf:_3  "Johnson" ];
  rdf:_4    [ rdf:_1  "mary@example.com"; rdf:_2  "Mary";
              rdf:_3  "Jenkins" ];
  rdf:_5    [ rdf:_1  "jamie@example.com"; rdf:_2  "Jamie";
              rdf:_3  "Smith" ]
] .
\end{lstlisting}\vspace{-0.5cm}
\end{figure}

\chadded{Listing~\ref{ex:running_example_rdf} shows the result of applying the mapping rules in ~\ref{def:fx:mappings} to the input of Example~\ref{running_example}.}
Considering the mapping to RDF defined above, we can show that the constructed quads meet a series of conditions.

\begin{corollary}\label{theorem:implicationsfxgraph}
     The following statements apply to the constructed RDF quads:
     \begin{enumerate}[label=\textit{(\roman*)}]
     
        \item\label{theorem:implicationsfxgraph:subjectcontainer} The subject is\footnote{``is'' in this context means ``is constructed on the basis of''.} a \textbf{Container};
        \item\label{theorem:implicationsfxgraph:propertyisslot} If the property  is not  \texttt{rdf:type}, then it is a \textbf{Slot};

         \item\label{theorem:implicationsfxgraph:objecttype}  If the property is  \texttt{rdf:type} and the object is not \texttt{fx:root}, then the object is a \textbf{Type};

         \item\label{theorem:implicationsfxgraph:predicatenumberslot} If the property is a container membership property, then the property is a \textbf{NumberSlot};
         
        \item\label{theorem:implicationsfxgraph:predicateslotvalueobject} If the object is a Literal, then  the property is a \textbf{Slot} and the object from a \textbf{Value}.

      %  \item\label{theorem:implicationsfxgraph:predicate} If the object is a \textbf{Value}, then the property is derived from a \textbf{Slot}.
         
     \end{enumerate}
\end{corollary}
\begin{proof}

    The statement \ref{theorem:implicationsfxgraph:subjectcontainer} is demonstrated by observing that the subjects in the mapping rules of the Definition~\ref{def:fx:mappings} is always an Entity constructed on the basis of a variable bound to a Container.

    The statement \ref{theorem:implicationsfxgraph:propertyisslot} is demonstrated by observing that all the rules that do not produce quads having \texttt{rdf:type} as property, namely \ref{def:fx:mappings:2}, \ref{def:fx:mappings:3}, \ref{def:fx:mappings:4} and \ref{def:fx:mappings:5}, have as property an Entity or a container membership property derived from a slot.

    The statement \ref{theorem:implicationsfxgraph:objecttype} is demonstrated by observing that the only mapping rule that produces a quad having \texttt{rdf:type} as property and an object different from \texttt{fx:root} is the \ref{def:fx:mappings:6}. The quad generated by this rule has as object an Entity derived from a type t.

    The statement \ref{theorem:implicationsfxgraph:predicatenumberslot} is demonstrated by observing that the only mappings comply with this condition are the \ref{def:fx:mappings:3} and \ref{def:fx:mappings:4} and in both rules the predicate is derived from a \textbf{NumberSlot}.

    The statement \ref{theorem:implicationsfxgraph:predicateslotvalueobject} is demonstrated by observing that the only mappings comply with this condition are the \ref{def:fx:mappings:4} and \ref{def:fx:mappings:5} and in both rules the predicate is derived from a \textbf{Slot} and the object from a \textbf{Value}.\end{proof}

Considering the implications demonstrated in the Corollary~\ref{theorem:implicationsfxgraph}, the mapping to RDF provided in the Definition~\ref{def:fx:mappings} and the \FX{} described in Section~\ref{sec:fxmodeldefition}, we can show that further implications hold.
Note that these implications apply to the constructed quads as well as to the facts implied by the statements in Corollary~\ref{theorem:implicationsfxgraph} (this is why they are provided separately).

\begin{corollary}\label{theorem:furtherimplications}
    The following statements hold:

    \begin{enumerate}[label=\textit{(\roman*)}]

        \item\label{theorem:furtherimplications:1} If the object is an IRI and the property is a \textbf{Slot}, then the object is a \textbf{Container};
        
        \item\label{theorem:furtherimplications:2} If the object is a  \textbf{Container}, then the property is a \textbf{Slot};

        \item\label{theorem:furtherimplications:3} If the object is a \textbf{Value}, then the property is a \textbf{Slot}.
        
    \end{enumerate}

\end{corollary}
\begin{proof}
    The statement \ref{theorem:furtherimplications:1} is demonstrated by observing that the only mapping rules that generate a IRI  as object (via the function Entity) are the \ref{def:fx:mappings:2} and \ref{def:fx:mappings:4}. In both, the predicate of the generated quad is a Slot.

    The statement \ref{theorem:furtherimplications:2} is demonstrated by observing that the only mapping rules that generate a Container as object are the \ref{def:fx:mappings:2} and \ref{def:fx:mappings:3}. In both, the predicate of the generated quad is a Slot.

    The statement \ref{theorem:furtherimplications:3} is demonstrated by observing that the only mapping rules that generate a Value as object are the \ref{def:fx:mappings:3} and \ref{def:fx:mappings:4}. In both, the predicate of the generated quad is a Slot.
\end{proof}

Finally, on the basis of the corollaries and the definitions above, we can show that further implications can demonstrated on the triple patterns satisfiable on the \FX{} graphs. 

\begin{corollary}\label{theorem:implicationsontriplepattern}
 The following conditions hold:

    \begin{enumerate}[label=\textit{(\roman*)}]

    \item\label{theorem:implicationsontriplepattern:1} If the object is \texttt{fx:root}, then the subject is derived from the \textbf{Root} and the property is either \texttt{rdf:type}  or a variable that can match only with \texttt{rdf:type}, we call such a term \textbf{TypeProperty};
    \item\label{theorem:implicationsontriplepattern:2} If the object is a \textbf{Type}, then the property is a \textbf{TypeProperty};
    \item\label{theorem:implicationsontriplepattern:3} If the property of a triple pattern is neither a container membership property, nor a variable, nor a \textbf{TypeProperty}, then it must be a \textbf{StringSlot}.
    \end{enumerate}
\end{corollary}
\begin{proof}
     The statement \ref{theorem:implicationsontriplepattern:1} is demonstrated by observing that the only mapping rule in which \texttt{fx:root} appears is at line \ref{def:fx:mappings:1} and the subject of this constructed quad is a Container and the property is \texttt{rdf:type}. Therefore, any triple pattern that has \texttt{fx:root} as object, graph must have a \textbf{TypeProperty} as property to be satisfiable on a \FX{} graph.

     Similarly, regarding the statement \ref{theorem:implicationsontriplepattern:2} we can observe that a Type is generated only by the mapping \ref{def:fx:mappings:6} which generates a quad having \texttt{rdf:type} as property.

     The statement \ref{theorem:implicationsontriplepattern:3} is demonstrated by observing that if we exclude the mapping rules that generate a container membership property are the \ref{def:fx:mappings:3} and \ref{def:fx:mappings:4} and those generating a quad having \texttt{rdf:type} as property are the \ref{def:fx:mappings:1} and \ref{def:fx:mappings:6}. Therefore, if we exclude also that the property of the triple is a variable, then the property must be a \textbf{StringSlot}.
\end{proof}

\subsection{Properties of \FX{} RDF graphs}

RDF graphs are usually interpreted as \textit{directed}, \textit{labelled} (mathematical) graphs, which we refer to as ``m-graphs'' to avoid confusion.
An m-graph can be constructed from an RDF graph by treating the subjects and objects of the triples as nodes and representing the properties as labelled edges.
Considering the properties of the \FX{} model demonstrated in the previous Sections we can derive some topological properties of the corresponding m-graph.

\begin{theorem}\label{theorem:mgraphconnectedess}
    \FX{} m-graphs built from a \FX{} RDF graphs are connected.
\end{theorem}
\begin{proof}
    We can observe that nodes of a \FX{} m-graph can be either Containers, Values, Types or \texttt{fx:root}.
    \chreplaced{Proposition}{Theorem}~\ref{proposition:connectiveness} demonstrates that for each container there is always a sequence of slots and containers that connect it the root.
    Moreover, Containers can be object only of quads having a Slot as property (see \ref{theorem:furtherimplications}.\ref{theorem:furtherimplications:2}), hence a sequence of Containers and Slots in the \FX{} model corresponds a path of Container nodes and Slot edges.
    Therefore, there is always a path that connects the Root node to any Container node.
    Finally, we can observe that Types, Values and \texttt{fx:root} are always directly connected to the Container they belong to (see \ref{def:fx:mappings}.\ref{def:fx:mappings:1}, \ref{def:fx:mappings}.\ref{def:fx:mappings:4}, \ref{def:fx:mappings}.\ref{def:fx:mappings:5} and \ref{def:fx:mappings}.\ref{def:fx:mappings:6}). Hence, the \FX{} m-graphs are connected.
\end{proof}

\begin{theorem}\label{theorem:fxgraphacyclic}
    \FX{} m-graphs built from a \FX{} RDF graphs are acyclic.
\end{theorem}
\begin{proof}
    A cycle involving a container node is not possible, as this would imply that a container contains itself recursively, which would invalidate condition \ref{def:fx:slots}.\ref{lst:fx:slots:9}.
    Types, \texttt{fx:root} and Values are directly connected to the container they belong to only and can not have outgoing edges.
    Therefore, \FX{} m-graphs are acyclic.
\end{proof}

\begin{theorem}\label{theorem:singleroot}
    \FX{} m-graphs have a single root.
\end{theorem}
\begin{proof}
    The only mapping rule matching the root container is the \ref{def:fx:mappings}.\ref{def:fx:mappings:1}.
    For each DataSource there is a single solution for the query of this rule since there is one and only one root for each DataSource, see \ref{def:fx:root}.\ref{def:fx:root:atleastoneroot} and \ref{def:fx:root}.\ref{def:fx:root:singleroot}.
\end{proof}

\begin{theorem}\label{theorem:mgraphsinglepath}
    In \FX{} m-graphs, there always exist a single directed path connecting the root node to any container node.
\end{theorem}
\begin{proof}
    \chreplaced{Proposition}{Theorem}~\ref{proposition:singlepath} demonstrates that for each Container c there is always a single sequence of containers and slots that connects the root of the data source to c. 
    Proof of Theorem~\ref{theorem:mgraphconnectedess} demonstrates that a sequence of Containers and Slots in the \FX{} model corresponds a directed path of Container nodes and Slot edges.
    Therefore, each container is connected to the root of the its data source by one and only one path.
\end{proof}

\begin{corollary}
\label{theorem:mgraphsinglepathcontainers}
    In \FX{} m-graphs, there is either a single directed path connecting any pair of container nodes, or there is none.
\end{corollary}
\begin{proof}
    This is a consequence of the Theorem~\ref{theorem:mgraphsinglepath}. We observe that for a given pair of nodes x and y, either y lies on the path from the root to x, or it does not.
\end{proof}

%Façade-X RDF graphs are:
%\begin{itemize}
    %\item Directed
    %\item Acyclic
    %\item Connected 
    %\item Labelled (...RDF)
%    \item Simple directed: no loops and no multiple arcs on the same two nodes [ARE WE SURE? 
%    G1: 
%        <a> a fx:root . 
%        <a> <b> "Enrico" . <a> <b1> "Enrico" .
%   G2:
%        <a> a fx:root . 
%        <a> rdf:\_1 <b> .
%        <a> rdf:\_2 <c> .
%        <b> a <Person> . 
%        <c> a <Person> .
%    ]
%    \item Non symmetric (edges only appear once in one direction) 
%    \item Non-transitive reducible (the transitive reduction of a FX graph is the same as the original FX graph)
    % \item NO Rooted tree: out-tree, all arcs are oriented against the root. Then, it is also bipartite and planar (edges do not intersect).
    %\item Single path to Container and FXRoot.
    %\item Rooted graph: there exists a directed path to every vertex from a distinguished root vertex
    %\item \FX graphs have at least one node typed Root.
%\end{itemize}

\section{Satisfiability of Basic Graph Patterns under \FX}\label{sec:theory}We characterise the Basic Graph Patterns that admits a solution on a \FX{} RDF graph (\sfx{BGP}).
Although the number of satisfiable \sfx{BGP} is infinite, we can characterise a \sfx{BGP} from the Triple Patterns (TPs) that admit a solution on a \FX{} RDF Graph and the possible joins among them.
Moreover, although the \FX{} model include their definition, Resource and DataSource are not involved in BGP expressions. 
In fact, a DataSource is mapped to a named graph (see Definition~\ref{def:fx:mappings}). 
A \FX{} Resource is mapped to an RDF Dataset, as intended by~\cite{w3c:datasets}.
In what follows, we study basic graph patterns.

\subsection{Satisfiable Triple Patterns} 
We remind some properties of a \FX{} RDF graph that can be derived by the theory outlined in the previous Sections.
\begin{itemize}
    \item The subject must be a Container;
    \item The set of all possible predicates is the union of TypeProperty, SlotString and SlotNumber;
    \item The set of all possible objects is the union of Value, Container, FXRoot and Type;
    \item Container, TypeProperty, SlotString, SlotNumber, Value, Root and Type, of these sets are mutually disjoint by definition.
\end{itemize}
Given these conditions, we can have only 6 TPs compatible with a \FX{} graph (out of the 12 possible combinations).
Table~\ref{tab:fxtripleppatterns} shows the TPs resulting all \chadded{the 12} possible  combinations of \FX{} components, together with an indication of their satisfiability.
\chadded{Consider the \FX{} triple pattern ``\texttt{Container} \texttt{SlotNumber} \texttt{Value}'' (first row of the table). This is satisfiable, since the \FX{} graph can contain such a triple. 
The last column of the table indicates the SPARQL triple patterns that can match the corresponding \FX{} triple pattern. For example, the SPARQL triple pattern  \VAR{}\IRI{}\LIT{} (e.g. \texttt{?s} \texttt{rdf:\_1} \texttt{"some value".}) can match   ``\texttt{Container} \texttt{SlotNumber} \texttt{Value}'' as the variable \texttt{?s} can be resolved to a container, \texttt{rdf:\_1} is a \texttt{SlotNumber} and \texttt{"some value"} is a Literal. The \FX{} triple pattern ``\texttt{Container} \texttt{SlotNumber} \texttt{Type}'' (second row of the table) is not satisfiable as  types in \FX{} occur only in triples with a \texttt{TypeProperty} as predicate.}

\begin{table}[th]
    \centering
    
    \caption{Triple Patterns result in all possible combinations of \FX{} components, together with an indication of their satisfiability on a \FX{} RDF graph.}
    \label{tab:fxtripleppatterns}
    \resizebox{\textwidth}{!}{\begin{tabular}{cccp{6.5cm}}
        \hline
        \textbf{Id} & \textbf{\FX{} Triple Pattern} & \chadded{\textbf{Satisfiability}}\chdeleted{\textbf{Validity}} & \textbf{Corresponding \chadded{SPARQL}\chdeleted{RDF} Triple Patterns} \\
        \hline
        \hline

\texttt{C.SN.V} & \texttt{Container} \texttt{SlotNumber} \texttt{Value} & \checkmark{} & \VAR{}\VAR{}\VAR{}, \VAR{}\VAR{}\BN{}, \VAR{}\VAR{}\LIT{}, \VAR{}\IRI{}\VAR{}, \VAR{}\IRI{}\BN{}, \VAR{}\IRI{}\LIT{}, \IRI{}\VAR{}\VAR{}, \IRI{}\VAR{}\BN{}, \IRI{}\VAR{}\LIT{}, \IRI{}\IRI{}\VAR{}, \IRI{}\IRI{}\BN{}, \IRI{}\IRI{}\LIT{}, \BN{}\VAR{}\VAR{}, \BN{}\VAR{}\BN{}, \BN{}\VAR{}\LIT{}, \BN{}\IRI{}\VAR{}, \BN{}\IRI{}\BN{}, \BN{}\IRI{}\LIT{}  \\
\texttt{C.SN.T} & \texttt{Container} \texttt{SlotNumber} \texttt{Type} & - & -  \\
\texttt{C.SN.R} & \texttt{Container} \texttt{SlotNumber} \texttt{FXRoot} & - & -  \\
\texttt{C.SN.C} & \texttt{Container} \texttt{SlotNumber} \texttt{Container} & \checkmark{} & \VAR{}\VAR{}\VAR{}, \VAR{}\VAR{}\IRI{}, \VAR{}\VAR{}\BN{}, \VAR{}\IRI{}\VAR{}, \VAR{}\IRI{}\IRI{}, \VAR{}\IRI{}\BN{}, \IRI{}\VAR{}\VAR{}, \IRI{}\VAR{}\IRI{}, \IRI{}\VAR{}\BN{}, \IRI{}\IRI{}\VAR{}, \IRI{}\IRI{}\IRI{}, \IRI{}\IRI{}\BN{}, \BN{}\VAR{}\VAR{}, \BN{}\VAR{}\IRI{}, \BN{}\VAR{}\BN{}, \BN{}\IRI{}\VAR{}, \BN{}\IRI{}\IRI{}, \BN{}\IRI{}\BN{}  \\
\texttt{C.SS.V} & \texttt{Container} \texttt{SlotString} \texttt{Value} & \checkmark{} & \VAR{}\VAR{}\VAR{}, \VAR{}\VAR{}\BN{}, \VAR{}\VAR{}\LIT{}, \VAR{}\IRI{}\VAR{}, \VAR{}\IRI{}\BN{}, \VAR{}\IRI{}\LIT{}, \IRI{}\VAR{}\VAR{}, \IRI{}\VAR{}\BN{}, \IRI{}\VAR{}\LIT{}, \IRI{}\IRI{}\VAR{}, \IRI{}\IRI{}\BN{}, \IRI{}\IRI{}\LIT{}, \BN{}\VAR{}\VAR{}, \BN{}\VAR{}\BN{}, \BN{}\VAR{}\LIT{}, \BN{}\IRI{}\VAR{}, \BN{}\IRI{}\BN{}, \BN{}\IRI{}\LIT{}  \\
\texttt{C.SS.T} & \texttt{Container} \texttt{SlotString} \texttt{Type} & - & -  \\
\texttt{C.SS.R} & \texttt{Container} \texttt{SlotString} \texttt{FXRoot} & - & -  \\
\texttt{C.SS.C} & \texttt{Container} \texttt{SlotString} \texttt{Container} & \checkmark{} & \VAR{}\VAR{}\VAR{}, \VAR{}\VAR{}\IRI{}, \VAR{}\VAR{}\BN{}, \VAR{}\IRI{}\VAR{}, \VAR{}\IRI{}\IRI{}, \VAR{}\IRI{}\BN{}, \IRI{}\VAR{}\VAR{}, \IRI{}\VAR{}\IRI{}, \IRI{}\VAR{}\BN{}, \IRI{}\IRI{}\VAR{}, \IRI{}\IRI{}\IRI{}, \IRI{}\IRI{}\BN{}, \BN{}\VAR{}\VAR{}, \BN{}\VAR{}\IRI{}, \BN{}\VAR{}\BN{}, \BN{}\IRI{}\VAR{}, \BN{}\IRI{}\IRI{}, \BN{}\IRI{}\BN{}  \\
\texttt{C.TP.V} & \texttt{Container} \texttt{TypeProperty} \texttt{Value} & - & -  \\
\texttt{C.TP.T} & \texttt{Container} \texttt{TypeProperty} \texttt{Type} & \checkmark{} & \VAR{}\VAR{}\VAR{}, \VAR{}\VAR{}\IRI{}, \VAR{}\VAR{}\BN{}, \VAR{}\IRI{}\VAR{}, \VAR{}\IRI{}\IRI{}, \VAR{}\IRI{}\BN{}, \IRI{}\VAR{}\VAR{}, \IRI{}\VAR{}\IRI{}, \IRI{}\VAR{}\BN{}, \IRI{}\IRI{}\VAR{}, \IRI{}\IRI{}\IRI{}, \IRI{}\IRI{}\BN{}, \BN{}\VAR{}\VAR{}, \BN{}\VAR{}\IRI{}, \BN{}\VAR{}\BN{}, \BN{}\IRI{}\VAR{}, \BN{}\IRI{}\IRI{}, \BN{}\IRI{}\BN{}  \\
\texttt{C.TP.R} & \texttt{Container} \texttt{TypeProperty} \texttt{FXRoot} & \checkmark{} & \VAR{}\VAR{}\VAR{}, \VAR{}\VAR{}\IRI{}, \VAR{}\VAR{}\BN{}, \VAR{}\IRI{}\VAR{}, \VAR{}\IRI{}\IRI{}, \VAR{}\IRI{}\BN{}, \IRI{}\VAR{}\VAR{}, \IRI{}\VAR{}\IRI{}, \IRI{}\VAR{}\BN{}, \IRI{}\IRI{}\VAR{}, \IRI{}\IRI{}\IRI{}, \IRI{}\IRI{}\BN{}, \BN{}\VAR{}\VAR{}, \BN{}\VAR{}\IRI{}, \BN{}\VAR{}\BN{}, \BN{}\IRI{}\VAR{}, \BN{}\IRI{}\IRI{}, \BN{}\IRI{}\BN{}  \\
\texttt{C.TP.C} & \texttt{Container} \texttt{TypeProperty} \texttt{Container} & - & -  \\

\hline
         
    \end{tabular}}
\end{table}

\subsection{Satisfiable Join conditions} 

\chadded{Without loss of generality, we can study the satisfiability of \sfx{BGP} consisting of two TPs.
In fact, more complex BGPs can be derived by joining BPGs of one or two TPs.}
\chdeleted{Without loss of generality, we can enumerate only \sfx{BGP} consisting of two TPs as more complex BGP can derived by joining BPGs of one or two TPs that have a single node in common.}
\chadded{Since properties never occur as the subject or object in \FX{}, joins that equate the subject or object of one TP with the predicate of another TP are not permitted.}
\chdeleted{We have that the join condition of two \FX{} TPs cannot equate subject (or object) of a triple pattern with the predicate of the other, as the \FX{} types are pairwise disjoint (see \ref{def:fx:disjointness}).}
Therefore, the joins $ S \bowtie P$ and  $ P \bowtie O$ will always result in a empty solution.
As a result, \FX{} TPs can be involved in four types of joins, namely: $ S \bowtie S$, $ P \bowtie P$, $ S \bowtie O$, and $ O \bowtie O$.

However, not all \FX{} TPs can be involved in all of these joins.
Consider the TPs \texttt{C.SN.V} and \texttt{C.SN.C}. 
The only joins allowed are $ S \bowtie S$, $ P \bowtie P$, and $ S \bowtie O$, but not $ O \bowtie O$ since Values and Containers are disjoint.

\chadded{For example, the following SPARQL BGPs are satisfiable:}
\begin{itemize}
    \item \chadded{\texttt{?s} \texttt{rdf:\_1} \texttt{?v} . \texttt{?s} \texttt{rdf:\_2} \texttt{?c} .  which is example of $ S \bowtie S$ join;}
    \item \chadded{\texttt{<i1>} \texttt{?p} \texttt{"some value"} . \texttt{<i2>} \texttt{?p} \texttt{<i3>} .  which is example of $ P \bowtie P$ join;}
    \item \chadded{\texttt{<i>} \texttt{rdf:\_1} \texttt{?c} . \texttt{?c} \texttt{rdf:\_2} \texttt{"some value"} .   which is example of $ S \bowtie O$ join;}
\end{itemize}
\chadded{Moreover, consider the following BGP.}
\\
\chadded{\texttt{?s1} \texttt{rdf:\_1} \texttt{?o} .} \\
\chadded{\texttt{?s2} \texttt{rdf:\_2} \texttt{?o} .}
\\
\chadded{We have that either \texttt{?o} is a value or it is a container; it cannot be both, as these are mutually exclusive. Therefore, if other constraints from different TPs require that \texttt{?o} be both a container and a value, the pattern is not satisfiable. Moreover, we will see that \texttt{?o} can only be a \texttt{Value}.}

In general, assuming the order of the triple pattern is meaningful, we have 6 valid \FX{} TPs, hence 36 \FX{} basic graph patterns of two joined \FX{} TPs.
If we do not consider the order of the TPs,  there are 21 possible \FX{} graph patterns\footnote{This is the number of all possible combinations without repetition of 2 elements from a set of 6, i.e. $\binom{6}{2}$, plus the number of \FX{} graph patterns where the same pattern is repeated twice, i.e. 6.}.
The join condition among these patterns can be one of the 4 above, resulting in 84 combinations.

To narrow down the number of satisfiable combinations, we can make the following considerations.
%Assuming that the other members of the triples are not distinct (otherwise both triple patterns would be bound to the same triple), we have that\todo{Not true, eg. spoXsp1o  }:
\begin{itemize}

    \item $S \bowtie S$ is always satisfiable as it equates containers (see \ref{theorem:implicationsfxgraph}.\ref{theorem:implicationsfxgraph:subjectcontainer});

    \item $S \bowtie O$ is satisfiable when the joined node of the first and the joined node of the second \FX{} triple pattern are of the same \FX{} type (otherwise the graph would not have a solution as the \FX{} types are pairwise disjoint, see Definition \ref{def:fx:disjointness}).
    
    %\item $ O \bowtie O$ is admissible when the objects of the two \FX{} patterns are both a Value or a Type;
    
    %\item $P \bowtie P$ is admissible when the predicates of the two triple patterns belong to the same \FX{} primitive.
\end{itemize}
As a result, we have 42 possible \FX{} graph patterns, which are showed in Table~\ref{tab:fxgraphpatterns}. 
%It is easy to see that these graph patterns can be combined to construct more complex structures.

\begin{table}[th]
    \centering
    
    \caption{The list of all possible joins of two \FX{} triple patterns with their corresponding admissible joining condition.}
    \label{tab:fxgraphpatterns}
    \begin{tabular}{cc}
        \hline
         \textbf{\FX{} Graph Pattern} & \textbf{Satisfiable Joins}  \\
         \hline\hline
         
 \texttt{C.SN.V} $\bowtie$ \texttt{C.SN.V} & $ S \bowtie S$ $ P \bowtie P$   $ O \bowtie O$  \\ 
 \texttt{C.SN.V} $\bowtie$ \texttt{C.SS.V} & $ S \bowtie S$   $ O \bowtie O$  \\ 
 \texttt{C.SN.V} $\bowtie$ \texttt{C.SS.C} & $ S \bowtie S$  $ S \bowtie O$   \\ 
 \texttt{C.SN.V} $\bowtie$ \texttt{C.TP.T} & $ S \bowtie S$     \\ 
 \texttt{C.SN.V} $\bowtie$ \texttt{C.TP.R} & $ S \bowtie S$     \\ 
 \texttt{C.SN.C} $\bowtie$ \texttt{C.SN.V} & $ S \bowtie S$ $ P \bowtie P$     \\ 
 \texttt{C.SN.C} $\bowtie$ \texttt{C.SN.C} & $ S \bowtie S$ $ P \bowtie P$  $ S \bowtie O$ $ O \bowtie O$  \\ 
 \texttt{C.SN.C} $\bowtie$ \texttt{C.SS.V} & $ S \bowtie S$     \\ 
 \texttt{C.SN.C} $\bowtie$ \texttt{C.SS.C} & $ S \bowtie S$  $ S \bowtie O$ $ O \bowtie O$  \\ 
 \texttt{C.SN.C} $\bowtie$ \texttt{C.TP.T} & $ S \bowtie S$     \\ 
 \texttt{C.SN.C} $\bowtie$ \texttt{C.TP.R} & $ S \bowtie S$     \\ 
 \texttt{C.SS.V} $\bowtie$ \texttt{C.SS.V} & $ S \bowtie S$ $ P \bowtie P$   $ O \bowtie O$  \\ 
 \texttt{C.SS.V} $\bowtie$ \texttt{C.TP.T} & $ S \bowtie S$     \\ 
 \texttt{C.SS.V} $\bowtie$ \texttt{C.TP.R} & $ S \bowtie S$     \\ 
 \texttt{C.SS.C} $\bowtie$ \texttt{C.SS.V} & $ S \bowtie S$ $ P \bowtie P$     \\ 
 \texttt{C.SS.C} $\bowtie$ \texttt{C.SS.C} & $ S \bowtie S$ $ P \bowtie P$  $ S \bowtie O$ $ O \bowtie O$  \\ 
 \texttt{C.SS.C} $\bowtie$ \texttt{C.TP.T} & $ S \bowtie S$     \\ 
 \texttt{C.SS.C} $\bowtie$ \texttt{C.TP.R} & $ S \bowtie S$ \\ 
 \texttt{C.TP.T} $\bowtie$ \texttt{C.TP.T} & $ S \bowtie S$ $ P \bowtie P$   $ O \bowtie O$  \\ 
 \texttt{C.TP.R} $\bowtie$ \texttt{C.TP.T} & $ S \bowtie S$ $ P \bowtie P$     \\ 
 \texttt{C.TP.R} $\bowtie$ \texttt{C.TP.R} & $ S \bowtie S$ $ P \bowtie P$   $ O \bowtie O$  \\
 
\hline
    \end{tabular}
\end{table}

\subsection{Satisfiable Basic Graph Patterns}\label{sec:satbgps}
Combining these patterns yields to more complex graph structures as those reported in Table~\ref{tab:bgp_motifs}.
%An unsatisfiable \sfx{BGP} is  a BGP violating at least a formula of the \FX{} model.%\todo{No! Perche' prima sono definiti come quelli che hanno una soluzione. Quindi non possono includere bgp non-satisfiable. Cioe', $BGP^{FX} \subset BGP$ cioe' i bgp che hanno almeno un potenziale grafo FX come soluzione.}
It is worth noting that  we can consider a BGP as a conjunction of its sub-BGPs, if a sub-BGP violates the \FX{} model, then the BGP is unsatisfiable.
Therefore, given a BGP, we could extract every possible sub-BGP from it and check whether any of these violate the \FX{} model.
If we find a sub-BGP that violates \FX{}, then the BGP is unsatisfiable.

We observe that there are formulas of the model that can be violated without compromising the satisfiability of the BGP.
In fact, a BGP is satisfiable if there is a \textit{partial} mapping for its variables under which the BGP is isomorphic with a \textit{portion} of the \FX{} RDF graph. 
Therefore, formulas requiring the existence of triples can be violated without compromising the satisfiability of the BGP, because the mandatory triples could be in a different portion of the graph.
For example, the formula \ref{def:fx:slots}.\ref{lst:fx:slots:1} implies that every value must be contained in a slot. 
However, BGPs that do not involve slots or values (e.g. \texttt{?s a ?t}) can violate formula \ref{def:fx:slots}.\ref{lst:fx:slots:1} while still being satisfiable.

This approach, i.e. analysing the BGP by its sub components to assess the satisfiability, raises three questions: 
\begin{enumerate}[label=\textit{(\roman*)}]
    \item\label{procedure:question1} Which formulas, if violated, would make a BGP unsatisfiable?
    \item\label{procedure:question2} How can we check whether a (sub-)BGP violates a model formula?
    \item\label{procedure:question3} Which (sub-)BGPs do we need to extract?
\end{enumerate}

\textbf{Answer \ref{procedure:question1}}. We identify the categories for the model formulas: \textit{type/disjoint implications} which are the formulas that allow us to reason about the types of the individuals; \textit{min-cardinality constraints} (resp. \textit{max-cardinality constraints}) which are the  formulas that require that there must be a minimum (resp. maximum) of individuals/pairs in the interpretation of a model; and \textit{integrity constraints} which are the formulas that impose that one or multiple individuals/pairs do not exist.

%\todo{Questo e' il periodo piu' bello dell'articolo. :-)}
It is worth noticing that min-cardinality constraints can be violated without compromising the satisfiability of the BGP.
These formulas imply the (implicit) existence of triples  that can exist in a different portion of the \FX{} RDF graph.
Instead, violating a maximum cardinality constraint, an integrity constraint, a type/disjoint implications  renders a BGP unsatisfiable, since these constraints prevent triples from existing, whereas BGP requires them to exist.
These call these kind of formulas \textit{inviolable}.

Moreover, we note that there are formulas that, by definition of the mapping rules \textit{alone} (see~\ref{def:fx:mappings}), can not be violated. 
For example, formula \ref{def:fx:basic}.\ref{def:fx:basic:6} can not be violated as NumberSlot and StringSlot are always generated by distinct rules.
Instead, for formula \ref{def:fx:basic}.\ref{def:fx:basic:1} to be violated, we need to know that StringSlots or NumberSlots are both Slots, which is knowledge that cannot be derived from the mappings alone.
Finally, it is worth noting that type implications themselves can only lead to the unsatisfiability of the BGP when considered alongside a disjoint implication.

\begin{table}[t]
    \centering

    \caption{Classification of the \FX{} model formulas by the categories identified: Min-Cardinality Constraint (Min), Max-Cardinality Constraint (Max), Type Implications (Type), Disjointness (Dis), Integrity Constraints (Int), BGP Inviolable (Inv) and BGP Inviolable By Construction of the Mapping Rules (Inv BC). %Formulas introduced in Definition~\ref{def:fx:disjointness} are not included as they are all disjoint implications classified as \textit{Inviolable}.
    }
    \label{tab:formula_classification}
    
    \resizebox{\linewidth}{!}{\begin{tabular}{c|cccc|cccccccccc|cccccccccc}
    \hline
    
         & 
        \multicolumn{4}{c|}{Definition \ref{def:fx:resources}} &
        \multicolumn{10}{c|}{Definition \ref{def:fx:basic}} &
        \multicolumn{10}{c}{Definition \ref{def:fx:slots}}\\
        
        Form  & 
        \ref{def:fx:resources:1} & \ref{def:fx:resources:2} & \ref{def:fx:resources:3} & \ref{def:fx:resources:4} &
        
        \ref{def:fx:basic:1} & \ref{def:fx:basic:2} & \ref{def:fx:basic:3} & \ref{def:fx:basic:4} & \ref{def:fx:basic:5} &\ref{def:fx:basic:6} & \ref{def:fx:basic:7} & \ref{def:fx:basic:8} & \ref{def:fx:basic:9} & \ref{def:fx:basic:10} &

        \ref{lst:fx:slots:1} & \ref{lst:fx:slots:3} & \ref{lst:fx:slots:4} & \ref{lst:fx:slots:5} & \ref{lst:fx:slots:6} & \ref{lst:fx:slots:8} & \ref{lst:fx:slots:2} & \ref{lst:fx:slots:7} & \ref{lst:fx:slots:9} & \ref{lst:fx:slots:10}
        
        \\\hline
        Min &
        & \checkmark & \checkmark & &
        & & & & & & & & \checkmark & &  
        \checkmark & \checkmark && \checkmark &&&&&&
        \\
        
        Max & 
        & & & \checkmark &
        & & & & & & & & & \checkmark &
        && \checkmark &&& \checkmark & \checkmark &&& \checkmark
        \\
        
        Type & 
        \checkmark & & & &
        \checkmark & \checkmark & \checkmark & \checkmark & \checkmark & & \checkmark & \checkmark & &  &
        &&&&&&&&&
        \\
        
        Disj & 
        & & & &
        & & & & & & & & &  &
        &&&&&&&&
        \\

        Int  & 
        & & & &
        & & & & & \checkmark & & & & &
        &&&& \checkmark &&& \checkmark & \checkmark &
        \\\hline
        
        BGP Inv & 
        & & & &
        \checkmark & & & & & & & & & &
        &&&& \checkmark & \checkmark & \checkmark & \checkmark & \checkmark & \checkmark
        \\
        
        BGP Inv BC & 
        
         &  &  &  &
        
        & \checkmark & \checkmark & \checkmark & \checkmark & \checkmark & \checkmark & \checkmark &  \checkmark  &  \checkmark &

        \checkmark & \checkmark & \checkmark & \checkmark &&&&&&
        \\
        
        \hline
    \end{tabular}}

    \resizebox{\linewidth}{!}{\begin{tabular}{c|cc|cccccccccccc|cccccccc}
    \hline
    
         & 
        \multicolumn{2}{c|}{Definition \ref{def:fx:types}} &
        \multicolumn{12}{c|}{Definition \ref{def:fx:disjointness}} &
        \multicolumn{8}{c}{Definition \ref{def:fx:root}} \\
        
        Form  & 
        \ref{def:fx:types:1} & \ref{def:fx:types:2} & 
        &\ref{def:fx:disjointness:1} &\ref{def:fx:disjointness:2} &\ref{def:fx:disjointness:3} &\ref{def:fx:disjointness:4} &\ref{def:fx:disjointness:5}&\ref{def:fx:disjointness:6}&\ref{def:fx:disjointness:7}&\ref{def:fx:disjointness:8}&\ref{def:fx:disjointness:9}&\ref{def:fx:disjointness:10}&\ref{def:fx:disjointness:11}&
        \ref{def:fx:root:1} & \ref{def:fx:root:2} & \ref{def:fx:root:3} & \ref{def:fx:root:4} & \ref{def:fx:root:5} & \ref{def:fx:root:6} & \ref{def:fx:root:7} & \ref{def:fx:root:8} 
        
        \\\hline
        
        Min &
        & & 
        &&&&&&&&&&&&
        \checkmark  &  & \checkmark  & & & & & \checkmark 
        \\
        
        Max & 
        & & 
        &&&&&&&&&&&&
        & & & & & &\checkmark  & 
        \\
        
        Type & 
        \checkmark & \checkmark  & 
        &&&&&&&&&&&&
        & \checkmark  & &\checkmark  & \checkmark  & & & 
        \\
        
        Disj & 
        & & 
        &\checkmark &\checkmark &\checkmark &\checkmark &\checkmark &\checkmark &\checkmark &\checkmark &\checkmark &\checkmark &\checkmark &\checkmark 
        & & & & & & & 
        \\

        Int  & 
        & & 
        &&&&&&&&&&&&
        & & & & & \checkmark  & & 
        \\\hline
        
        BGP Inv & 
        & & 
        &&&\checkmark &&&&&&&\checkmark&&
        & & & & & \checkmark & \checkmark & 
        \\
        
        BGP Inv BC & 
        \checkmark  &\checkmark  & 
        &\checkmark &\checkmark & &\checkmark &\checkmark &\checkmark & \checkmark & \checkmark &\checkmark & &\checkmark &\checkmark
        & \checkmark & &\checkmark & \checkmark & & & 
        \\
        \hline
    \end{tabular}}

\end{table}

Table~\ref{tab:formula_classification} classifies the model formulas according to the above-defined categories and indicates wether the formula is inviolable or inviolable by construction.
%Formulas introduced in Definition~\ref{def:fx:disjointness} are not included as they are all disjoint implications classified as \textit{Inviolable}. 

%Formulas in Definition   \ref{def:fx:resources} can not be violated due to the way the mapping rules are constructed.
%Formulas of Definition \ref{def:fx:basic} are either type implications, min-cardinality constraints or formulas that can not be violated due to the mapping rules \ref{def:fx:basic}.\ref{def:fx:basic:6} and \ref{def:fx:basic}.\ref{def:fx:basic:10}.
%In Definition \ref{def:fx:slots}, \ref{lst:fx:slots:1}, \ref{lst:fx:slots:3}, and \ref{lst:fx:slots:5} are min-cardinality constraints; \ref{lst:fx:slots:4}, \ref{lst:fx:slots:8} and \ref{lst:fx:slots:2} are max-cardinality constraints; \ref{lst:fx:slots:6}, \ref{lst:fx:slots:7} and \ref{lst:fx:slots:9} are  integrity constraints.
%Definition \ref{def:fx:types} introduces a type implication \ref{def:fx:types:1} and min-cardinality constraint \ref{def:fx:types:1}.
%Formulas in Definition \ref{def:fx:disjointness} are all disjoint implications.
%Definition \ref{def:fx:root} introduces: min-cardinality constraints, i.e. \ref{def:fx:root:1}, \ref{def:fx:root:3} and \ref{def:fx:root:8}; type induction formulas \ref{def:fx:root:2}, \ref{def:fx:root:4} and \ref{def:fx:root:5}; an integrity constraint \ref{def:fx:root:6} ; and a max-cardinality constraint \ref{def:fx:root:7}.
%To sum up, most of the formulas 

\textbf{Answer \ref{procedure:question2}}. 
A BGP violates a model formula if there is a mapping between its nodes and the formula's variables that renders the formula invalid (i.e. if it makes the negation of the formula true).
For example, consider the formula \ref{def:fx:slots}.\ref{lst:fx:slots:7}, the triple pattern \texttt{?x ?p ?x} violate it since we can assign to both \texttt{c} and \texttt{s} the variable \texttt{?x}.
Then we can reformulate question \ref{procedure:question2} as ``how can we assign variables of the inviolable formulas to the nodes of the BGP?''

For each \textit{inviolable} formula, we can construct, via the mapping rules in Definition \ref{def:fx:mappings}, a \textit{proxy graph} that violates the formula.
This is the minimal graph that violates a formula and exemplifies the pattern that must be avoided in BGP to ensure satisfiability.
Since mapping rules are bidirectional (see Theorem~\ref{theorem:bidirectionalmappings}) assigning BGP nodes to formula variables is equivalent to assigning BGP nodes to proxy-graph nodes.
Therefore, BGPs that are isomorphic to a proxy pattern violate an inviolable formula.

Table \ref{tab:proxy_patterns} shows the inviolable formulas and the corresponding proxy graphs.
For brevity, Table \ref{tab:proxy_patterns} omits disjoint formulas.
Such formulas can be violated by forcing a node to be of two disjoint types.
Most of the formulas cannot be violated due to the construction of the mapping rules, either because the two types of a disjoint occur in different roles (e.g. slots and containers are placed as properties and subjects) or because they are materialised using different RDF node types (e.g. values and containers are literals and IRIs).
These conditions can be validated by checking if the triple patterns of the BGP match the valid triple pattern indicated in Table~\ref{tab:fxtripleppatterns}. 
Therefore, the only disjoint that can be violated is between Type and Container.
%Formulas that are inviolable by construction of the mapping rules  as well as violable formulas (i.e. min-cardinality constraints) are not reported in the Table \ref{tab:proxy_patterns}.
%We provide a single proxy graph for multiple formulas wherever possible.
%Since formulas of the same type have a similar shape, we provide a single proxy graph for each formula type that apply to the formulas of the same category.

\begin{table}[t]
    \centering
    \caption{Inviolable formulas and corresponding proxy graph}
    \label{tab:proxy_patterns}
    \begin{tabular}{cl}
        \hline
        \textbf{Formula} & \multicolumn{1}{c}{\textbf{Proxy Graph} } \\\hline\hline

        \ref{def:fx:basic}.\ref{def:fx:basic:1}, \ref{def:fx:disjointness}.\ref{def:fx:disjointness:3}, \ref{def:fx:disjointness}.\ref{def:fx:disjointness:10} &
        \texttt{<i1>}~\texttt{xyz:s}~\texttt{<i2>}~.  \texttt{<i3>}~\texttt{rdf:type}~\texttt{<i1>} \\

        \ref{def:fx:slots}.\ref{lst:fx:slots:6} &
        \texttt{<i1>}~\texttt{rdf:\_1}~\texttt{``a''}~.  \texttt{<i1>}~\texttt{rdf:\_1}~\texttt{<i2>} \\

        \ref{def:fx:slots}.\ref{lst:fx:slots:8} &
        \texttt{<i1>}~\texttt{rdf:\_1}~\texttt{<i2>}~.  \texttt{<i1>}~\texttt{rdf:\_1}~\texttt{<i3>} \\

        \ref{def:fx:slots}.\ref{lst:fx:slots:2} &
        \texttt{<i1>}~\texttt{rdf:\_1}~\texttt{``1''}~.  \texttt{<i1>}~\texttt{rdf:\_1}~\texttt{``2''} \\

        \ref{def:fx:slots}.\ref{lst:fx:slots:7} &
        \texttt{<i1>}~\texttt{rdf:\_1}~\texttt{<i1>}\\
        
        \ref{def:fx:slots}.\ref{lst:fx:slots:9} &
        \texttt{<i1>}~\texttt{rdf:\_1}~\texttt{<i2>} . \texttt{<i2>}~\texttt{rdf:\_1}~\texttt{<i3>} . \texttt{<i3>}~\texttt{rdf:\_1}~\texttt{<i1>} 
        \\

        \ref{def:fx:slots}.\ref{lst:fx:slots:10} &
        \texttt{<i1>}~\texttt{rdf:\_1}~\texttt{<i2>} . \texttt{<i1>}~\texttt{rdf:\_2}~\texttt{<i2>} 
        \\

        \ref{def:fx:root}.\ref{def:fx:root:6} & 
    \texttt{<i1>}~\texttt{rdf:type}~\texttt{fx:root}~. \texttt{<i2>}~\texttt{rdf:\_1}~\texttt{<i1>}\\
        
        \ref{def:fx:root}.\ref{def:fx:root:7} & 
    \texttt{<i1>}~\texttt{rdf:type}~\texttt{fx:root}~. \texttt{<i2>}~\texttt{rdf:type}~\texttt{fx:root}

        \\\hline
    \end{tabular}
\end{table}

\begin{remark}\label{remark:multiplepaths}
     In addition to the above conditions, there is an extra constraint that must be met to ensure the BGP is satisfiable.
    This constraint is treated separately from the others as it implicitly emerges from the structural properties of the \FX{} m-graph.
    Since there is a single path connecting the root to its containers (see \ref{theorem:mgraphsinglepath}) and there is either a single path connecting any pair of container nodes, or there is none (see \ref{theorem:mgraphsinglepathcontainers}), backward paths of containers to the root must be compatible. 
    For example, the following BGP is not satisfiable since \texttt{<iri>} which must be a container, and it is contained by two different containers: \texttt{<j1>} and \texttt{<j2>} (in case that \texttt{?j3=?j5}, \texttt{?p3=?p4}, \texttt{?p2=?p4}), or, \texttt{<j3>} and \texttt{<j5>} (if $\tt{j3}\neq \tt{j5}$):
    \begin{verbatim}
<j1> ?p2 ?j3 . ?j3 ?p3 <iri> .
<j4> ?p4 ?j5 . ?j5 ?p4 <iri> . <iri> ?p ?o .    \end{verbatim}
\end{remark}
\textbf{Answer \ref{procedure:question3}}. 
It is worth noting that, in order to determine whether a BGP violates an inviolable formula, we must first establish whether the BGP contains an invalid FX triple pattern, and secondly, whether it is isomorphic to a proxy graph in Table  \ref{tab:proxy_patterns}. %\todo{is sufficient? do we need to check joins?}.
If we exclude proxy graph associated to formula \ref{def:fx:slots}.\ref{lst:fx:slots:9} that imposes that \FX{} RDF graphs are acyclic, conditions that would make the BGP unsatisfiable would span across one or two triple patterns.
Therefore, assuming that the BGP is acyclic, we can determine its validity by extracting sub-BGPs of one or two triple patterns and checking the above-defined conditions.

\section{Detecting satisfiable Basic Graph Patterns under \FX{}}\label{sec:algorithms}
We apply the \FX{} theory defined in the previous sections to determine the satisfiability of basic graph patterns in \FX{}.
We first summarize the components of the theory and their role in addressing the problem and then present two algorithms to annotate basic graph patterns and determine their satisfiability.

\subsection{Implementing the theory}

We will introduce some preliminary definitions and supporting functions that will be used in later algorithms.

\textbf{Satisfiable BGPs under \FX{}.}
We define $G$ as the set of all possible RDF graphs, and \sfx{G} as the set of possible RDF graphs that derive from the mappings of definition~\ref{def:fx:mappings}.
Furthermore, we define BGP\footnote{We remind that a BGP is a set of triple patterns, then $BGP:=2^{(\mathcal{I}\cup \mathcal{B}\cup \mathcal{V})\times(\mathcal{I}\cup \mathcal{V})\times(\mathcal{I}\cup \mathcal{B} \cup \mathcal{L}\cup \mathcal{V})}$} as the set of SPARQL Basic Graph patterns and $\sfx{BGP} \subset BGP$ as the set of BPGs having  at least one solution against at least one possible $g \in \sfx{G}$. Note that there is a strict inclusion between \sfx{BGP} and BGP as we already demonstrated that not all triple patterns/joins are satisfiable under \FX{} (see Section~\ref{sec:satbgps}).

Any $bgp \in \sfx{BGP}$ must comply with the conditions determined by our theory, namely:
\begin{enumerate}[label=C\arabic*]
    
    \item No cycles are allowed in BGPs (see Theorem~\ref{theorem:fxgraphacyclic}); \label{condition:nocycles}\label{conditionfirst}

    \item Join conditions of the BGP must comply with those reported in Table~\ref{tab:fxgraphpatterns}; \label{condition:joins}

    \item BGPs cannot have two alternative paths that lead from any node to a container (see remark~\ref{remark:multiplepaths}); \label{condition:multiplepaths}
    
    \item Mapping Rules (i.e. \ref{def:fx:mappings}) and related corollaries \ref{theorem:implicationsfxgraph}, \ref{theorem:furtherimplications} and \ref{theorem:implicationsontriplepattern};\label{condition:mappingrules}

    \item There must be \textit{at most} one root (see \ref{theorem:singleroot}); \label{condition:singleroot}

    \item Triple patterns of the BGP must comply with those reported in Table~\ref{tab:fxtripleppatterns}; \label{condition:triplepattern}
    
    \item None of the sub-BGPs included in bgp must be isomorphic to the proxy graphs defined for the inviolable formulas as reported in the Table~\ref{tab:proxy_patterns}. \label{condition:proxygraph}\label{conditionlast}
    
    %\item The constraints defined by the theory needs to be satisfied, specifically:
    %\begin{enumerate}
     %   \item definition~\ref{def:fx:mappings},
      %  \item related corollaries~\ref{def:fx:mappings:1}-\ref{def:fx:mappings:6}, and
       % \item their implications in definition~\ref{def:fx:inferences}.\todo{review according to new classification of rules}
       
    %\item The solution graph cannot have two alternative paths that lead from any node to a object Container (as per theorem~\ref{theorem:furtherimplications}) or to the node FXRoot (as per~\ref{??}). We note that when such paths exist, the shortest one cannot start from a root container, as that would invalidate~\ref{XXX}. When the object node is FXRoot, since the root container can only be one, if two triple patterns both have the object \texttt{fx:root}, their subject and predicate must be compatible.\todo{This invalidates other conditions defined. so maybe we can avoid it.}
       % \end{enumerate}
\end{enumerate}

We ignore the following properties of \FX{} RDF graphs:
\begin{itemize}
    \item Connectedness (cf. \ref{theorem:mgraphconnectedess} and \ref{proposition:singlepath}): because BGP don't need to, as non-joined triple patterns will be executed as their product;
    \item \FX{} graphs are rooted (cf. \ref{theorem:singleroot}): because the BGP can safely match a portion of the \FX{} graph;
    \item Min-Cardinality constraints and conditions defined in formulas that are inviolable by construction of the mapping rules (cf. Section~\ref{sec:satbgps}).
\end{itemize}

Next, we introduce functions to avoid unsupported joins, cycles and multiple paths to container nodes. 
These functions guarantee that the conditions \ref{condition:joins}, \ref{condition:nocycles} and \ref{remark:multiplepaths} are satisfied.

We proceed by describing the procedure to remove unsupported joins.

\begin{algorithm}[Detection of unsupported joins]
\sfx{BGP} cannot have Subject/Property or Object/Property joins within or across triples.
%Denoting the BGP with cycles as $BGP^{C}$, we know that $BGP^{C} \subseteq BGP \wedge BGP^{FX} \subseteq BGP \wedge BGP^{C} \oplus BGP^{fx}$.
\chadded{Listing~\ref{alg:hasUnsupportedJoin}}\chdeleted{We} define\chadded{s} the function $hasUnsupportedJoin(bgp) \rightarrow Bool$ to detect unsupported joins. 
{\normalfont\begin{lstlisting}[style=algorithm,caption={Pseudocode of the  hasUnsupportedJoin function.},label={alg:hasUnsupportedJoin}]
function hasUnsupportedJoin(bgp): Bool
    subjects = {}, properties = {}, objects = {}
    for (s, p, o) $\in$ bgp:
        subjects = subjects $\cup$ {s}
        properties = properties $\cup$ {p}
        objects = objects $\cup$ {o}
        if s $\in$ properties $\vee$ o $\in$ properties $\vee$ p $\in$ subjects $\vee$ p $\in$ objects :
            return True
    return False
\end{lstlisting}}
\end{algorithm}

\textbf{Time Complexity of \texttt{hasUnsupportedJoin}.}
   The time complexity depends on the number of triples in the bgp, hence $O(|bgp|)$.

We proceed to detect cycles in basic graph patterns.

\begin{algorithm}[Detection of cycles]
Since \FX~graphs are acyclic, \sfx{BGP} also cannot have cycles.
%Denoting the BGP with cycles as $BGP^{C}$, we know that $BGP^{C} \subseteq BGP \wedge BGP^{FX} \subseteq BGP \wedge BGP^{C} \oplus BGP^{fx}$.
We define the function $hasCycle(bgp) \rightarrow Bool$ to detect cycles in BGPs. 
This function explores the bgp by performing a depth-first search  (DFS, see \texttt{detectCycle}) for each triple pattern of the bgp.
If a node is visited two times by the search, the function returns True.
\chadded{hasCycle and detectCycle are defined in Listing~\ref{alg:detectcycle}.}
{\normalfont\begin{lstlisting}[style=algorithm,caption={Pseudocode of functions for detecting cycles.},label={alg:detectcycle}]
function hasCycle(bgp): Bool
    for (s, p, o) $\in$ bgp:
        if(detectCycle (bgp, {s}, s)):
            return True
    return False

function detectCycle(bgp, visited, last): Bool
    for (last, p, next) $\in$ bgp:
        if next $\in$ visited:
            return True
        else:
            return detectCycle(bgp, visited $\cup$ {last}, next)
    return False
\end{lstlisting}}
\end{algorithm}

\textbf{Time Complexity of cycle detection.}
    The time complexity of a depth-first search on a graph of n nodes and e edges is $O(n+e)$.
    We observe that the function \texttt{hasCycle} performs a DFS for each triple of the BGP (number of edges of the graph).
    Therefore in the worst case the time complexity is $O(|bgp|^2+|bgp|*n)$.

Once the cycling detection is verified, we move to verify the constraints of our theory. 
To do that, we proceed to annotate the nodes in the pattern with the predicates resulting from the \FX{} mappings to RDF. 
We list the possible predicates in Table~\ref{tab:fx:hierarchy}, which also summarises the following characteristics.

\textbf{Ground predicates.}
Nodes in BGPs have roles related to the RDF triples (Subject, Predicate, Object) or roles resulting from mappings to the \FX{} model: Container, FXRoot (to represent the role of the term \texttt{fx:root}), Slot, SlotString, SlotNumber, Value, Type and TypeProperty. 
We  organise these roles in a predicate hierarchy, where the top role is the RDF Node and the first layer Subject, Predicate, Object, which are in turn specialised by \FX{} predicates.
Ultimately, the nodes in $BGP^{FX}$ will necessarily be one of Container (Subject or Object), Type (Object), FXRoot (Object), TypeProperty (Property), SlotString (Slot which specialises Property), SlotNumber (Slot which specialises Property), or Value (Object).
We name those \textit{ground predicates} ($GP$), i.e. the set of \FX{} predicates that cannot be specialised.
These are marked with a * in Table~\ref{tab:fx:hierarchy}.
Table~\ref{tab:fx:hierarchy} also summarises the disjointness between predicates (derived from Definition~\ref{def:fx:disjointness} and the analysis of satisfiable joins, see Table~\ref{tab:fxgraphpatterns}).
%Similarly, from the analysis of admissible triple patterns and joins (Table~\ref{}), we derive a node roles disjointness matrix.
%This is also summarised in Table~\ref{tab:fx:hierarchy}.

\begin{table}[t]
\caption{\FX ~RDF predicates hierarchy and disjointness. \textit{Ground predicates} are marked with a *}
\label{tab:fx:hierarchy}
\centering
\begin{tabular}{lll}  
\hline\multicolumn{1}{c}{\textbf{Predicate}} & \multicolumn{1}{c}{\textbf{Specialises}} & \multicolumn{1}{c}{\textbf{Disjoint with}}\\\hline\hline
Subject & & Property \\
Property & & Subject, Object  \\
Object & & Property \\
TypeProperty* & Property & Subject, Object, Slot, Type, Container\\
Type* & Object  & Slot, Container, Value, FXRoot \\
Container* & Subject, Object & Property, Slot, Value, Type, FXRoot \\
Slot & Property & Subject, Object,TypeProperty \\
Value* & Object & Property, Subject,Type,Container,FXRoot,Slot\\
FXRoot* & Object & Subject, Property,Container,Value,Type\\
SlotNumber* & Slot & SlotString, Subject, Object\\
SlotString* & Slot & SlotNumber, Subject, Object \\\hline

\end{tabular}
\end{table}
%\todo[inline]{FXRoot is not the same as fx:root: the first is an annotation/role that a node can have; the second is the concrete IRI fx:root. A rule may need to refer to one or the other.}
\begin{definition}[Annotated Basic Graph Patterns]\label{def:fx:annotate}
We indicate the set of nodes of a BGP as $\mathcal{N}$, i.e. $\mathcal{N}:=\mathcal{V} \cup\mathcal{I} \cup \mathcal{B} \cup \mathcal{L}$
We define the function $\tt{nodes}: BGP \rightarrow 2^{\mathcal{N}}$ which given a BGP returns the set of nodes it contains.
%We define $O \subseteq (I\cup~B\cup~V\cup~L)$ as the set of nodes in any BGP, and define a function $nodes(bgp \in BGP)$ to obtain $\{n_{1},n_{\dots}\} | n\in O$.
We define a function $\tt{annotate}: BGP  \rightarrow  Map(\mathcal{N}, \tt{GP})$\footnote{\texttt{Map(X,Y)} indicates the set of all possible dictionaries having elements in X as keys and Y as values.} which associates every bgp with a set of pairs (node, ground predicate).
%In other words, $annotate()$ generate annotations of bgp, each annotation binding each node to a \textit{predicate} from Table~\ref{tab:fx:hierarchy}.
%\todo[inline]{Maybe to be removed}
%We denote $ann_{bgp} \in ANN_{bgp}$, where $ANN_{bgp}$ is the set of all possible annotations on a BGP, and $ANN$ all possible annotated BGPs, so that $ANN_{bgp} \subset~ANN$.
%Furthermore, establish a one-to-one correspondence between $map_{bgp}^{n\rightarrow~p} \leftrightarrow ann_{bgp}$.
\end{definition}

\textbf{Relations between $BGP$ and annotated BPGs.} Note that a BGP can be annotated in many ways, i.e. $|annotate({bgp})| \geq~1. \forall bgp \in BGP$
Note that the existence of an annotation for a BGP does not guarantee that it is satisfiable.
% However, only a subset of them will be valid, ie. will represent a set of solutions on a possible \FX~graph.

% \begin{remark}[Minimal number of distinct nodes in a \FX~BGP] For each non-empty $bgp^{FX} \in BGP^{FX}$, $|bgp^{FX}_{nodes}| \geq 3$, because there cannot be repetitions of nodes in a single \FX~triple pattern (subject, predicate, and object must be all different)\todo{reference to theory}.
% \end{remark}

%\begin{remark}[Grounded Annotations of Graph Patterns]
%We note that a subset of $ANN$ will only include \textit{grounded predicates}, and denote such set $ANN^{GP} \subseteq~ANN$. 
%\end{remark}

%\begin{remark}[Satisfiability of annotated BGPs]
%We know $ANN^{GP} \subset ANN$ but we don't know whether the annotated BGP satisfies the theory, yet.
%We denote the set of annotations including only ground predicates that also satisfies the theory as %$ANN^{FX} \subseteq ANN^{GP} \subseteq ANN$. 
%We observe that there can still be annotations that don't violate the theory but are not part of %$ANN^{GP}$, since include predicates that are not in $GP$.
%We say that $ANN^{\neg~GP} \oplus ANN^{GP} \subseteq ANN$.
%\end{remark}

Later, we present two algorithms that, given a generic $bgp \in BGP$, can generate all possible annotations that satisfy the theory.
%If no valid annotations  exist, the basic graph pattern does not have a solution in \FX, ie. it is not satisfiable.
In the remaining of this section, we are going to summarise the elements of the theory and introduce pseudo-code that will be later applied in the algorithms.

In what follows, we distinguish the BGP node types (Variable, IRIs, Literal) from the predicates, i.e. the \textit{annotated roles} that each node may have in a \FX{} graph matching scenario.
We omit blank nodes, as they operate as existential qualifiers and thus behave like variables do~\cite{w3c:sparql11}. 
%Thus, the following rules  include the additional predicates: 
%\begin{itemize}
%    \item IRI( node): the node is a concrete IRI
%    \item Variable( node): the node is a variable or a blank node -- we collapse Blank nodes with variables  because they act as existential qualifier in BGPs (as variables that do not capture values).
%    \item Literal( node ) -- the node is a concrete literal.
%\end{itemize}

%\paragraph{\FX~RDF predicates hierarchy table}

\begin{definition}\label{def:inferencerules}
    According to the above theory, the algorithm implements the following rule for annotating a BGP.
Note that if a rule raises an inconsistency when verified, the NS term is raised, which stands for 'no solution'. 

\begin{enumerate}[label=R\arabic*]
    \item If a node is a subject, then it is a Container; \label{crule:1}\label{crule:first}
    \item If a node is the  IRI fx:root, then it is FXRoot;\label{crule:2}
    \item If a node is the  IRI rdf:type, then it is TypeProperty;\label{crule:3}
    \item If a Property is not a variable or the IRI rdf:type, then it is a Slot;\label{crule:4}
    \item If the Object of a triple is not a Variable and not the  IRI fx:root but Property is rdf:type, then the Object is a Type;\label{crule:5}
    \item If the Object is FXRoot, then the Property is TypeProperty;\label{crule:6}
    \item If the Object is Type, then Property is TypeProperty;\label{crule:7}
    \item If the Property is a ContainerMembershipProperty, then the Property is SlotNumber;\label{crule:8}
    \item If the Property is neither a Variable, rdf:type, or ContainerMembershipProperty, then is a SlotString;\label{crule:9}
    \item If a node is a Literal, then it is a Value;\label{crule:10}
    \item If the Object is a IRI and Property is Slot, then Object is a Container;\label{crule:11}
    \item If the Object is a Value, then the Property is Slot;\label{crule:lasttp}\label{crule:12}
    %\item Object cannot be a IRI that is not fx:root and be FXRoot;
    \item If the Object is an IRI not equal to \texttt{fx:root} and it is annotated as FXRoot, then raise NS;\label{crule:13}
    %\item Nodes annotated as Type cannot be the concrete fx:root;
    \item If the Object is a Type and is \texttt{fx:root}, then raise NS;\label{crule:14}
    %\item Object cannot be IRIs and Value;
    \item If the Object is an IRI and is a Value, then raise NS;\label{crule:15}
    \item If the Object is a Container, then Predicate is Slot;\label{crule:16}
    \item If two triples have the same subject and predicate, and the predicate is a Slot, they also need to have matching objects - where either or are variables or are the same node;\label{crule:17}
    \item If there is a subject/object join when the object is FXRoot, then raise NS;\label{crule:18}
    \item If a node is Object and Container or FXRoot and there are disjoint paths (S$\bowtie$O joins) leading to it, then raise NS.\label{crule:19}\label{crule:last}
\end{enumerate}
\end{definition}

% Any BGP$^{fx}$  must comply with the rules above.
% We pose the following research questions:
% \begin{itemize}
%     \item[RQ1] How to determine all possible interpretations of a BGP in \FX?
%     \item[RQ1] Is it always possible to determine whether a FXGB is satisfiable? 
%     \item[RQ2] Is a satisfiability algorithm sound and complete? 
%     \item[RQ3] What is the complexity of such algorithm?
% \end{itemize}
% Is it possible to determine the number of expected interpretations, theoretically?
% From here, determine the two complexity issues: number of BGPs to evaluate, number of grounded BGPs.

\textbf{Data structures used in the algorithms.}
We introduce the data structures which will be referred to by both algorithms. These are: 
\begin{enumerate*}[label=(\roman*)]
    \item the input Basic Graph Pattern, i.e. a sequence of triples denoted by the variable \texttt{bgp};
    \item the set of predicates of \FX{} indicated as \texttt{terms};
    \item the \FX{} RDF terms hierarchy, as presented in Table~\ref{tab:fx:hierarchy}, i.e. a map, named \texttt{specialisations}, which associates each \FX{} predicate to its direct super type;
    \item the \FX{} RDF terms disjointness Table ~\ref{tab:fx:hierarchy}, a map, named \texttt{disjointness} which associates each \FX{} predicate to the set of predicates it is disjoint with;
    \item ground terms, i.e. \FX{} set of terms that cannot be specialised, indicated as \texttt{groundTerms};
    \item top terms, i.e. terms that do not specialise other terms, indicated as \texttt{topTerms}.
    % \item the \FX ~ implications as per Definition~\ref{def:fx:inferences}, and
    % \item a consistency check function to apply them (cf. lines~9-16).
\end{enumerate*}
%{\normalfont 
%\begin{lstlisting}[style=algorithm,escapechar=|]
%bgp $\leftarrow$ [(s$_1$,p$_1$,o$_1$), (s$_2$,p$_2$,o$_2$), $\dots$ (s$_n$,p$_n$,o$_n$)]
%terms $\leftarrow$ {'Subject', 'Predicate', 'Object', 'Container', 'Slot', 'SlotNumber', 'SlotString', 'Value', 'TypeProperty', 'Type', 'Root'} 
%specialisations $\leftarrow$ { type $\rightarrow$ supertype } # See Table|~\ref{tab:fx:hierarchy}|
%disjointness $\leftarrow$ { term $\rightarrow$ { disjoints } } # See Table|~\ref{tab:fx:disjointness}|
%groundTerms $\leftarrow$ { term $|$ term $\rightarrow$ superTerm $\in$ specialisations $and$ subTerm $\rightarrow$ term $\notin$ specialisations }
%topTerms $\leftarrow$ { term | subterm $\rightarrow$ term $\in$ specialisations $and$ term $\rightarrow$ superterm $\notin$ specialisations }
%\end{lstlisting}}

\textbf{Implementation of inference rules.}
    We assume that the consistency rules presented above (see \ref{crule:first}-\ref{crule:last}) are implemented in form of implications ($\tt{head}\leftarrow \tt{body}$), meaning that the head is applied if the body is true.
    These rules are accessible through the variable \texttt{rules}.
    We also introduce the functions:
    \begin{enumerate*}[label=\textit{(\roman*)}]
    \item $\tt{body}(rule, node, map):Bool$ which takes a rule a node of the BGP and a $map \in \tt{Map(\mathcal{N},terms)}$, it evaluates the body of the rule and returns True if the condition holds, False otherwise;
    \item $\tt{head(rule, node, map):Bool}$ which applies the head of a rule (returns the implied FX term or NS).
    \end{enumerate*}

\begin{algorithm}[Checking constraints from inference rules]
We introduce the function \texttt{check} \chadded{(Listing~\ref{alg:check})}, which verifies the consistency of a bgp against the set of rules presented before. 
The function iterates through the nodes of the bgp, applying each rule in \texttt{rules} to each node. 
If the body of the rule is true, then it obtains the inferred term from the head of the rule.
If the the inferred term is inconsistent with the other terms already associated with the node, it returns False; otherwise, it continues with the next rule.
The function returns True if no inferred term raises an inconsistency.
%These are a list of \FX ~ rules as per Definition~\ref{def:fx:inferences} (cf. line~1), two functions to access the head and body of each rule (cf. lines 2 and 3), and
%a function to apply them (cf. lines~4-11).
%rules $\leftarrow$ [head $\Leftarrow$ body] # See list of rules in Definition|~\ref{def:fx:inferences}|
%function body(rule, node, map) # evaluates the body of a rule (true if condition holds, false otherwise)
%function head(rule) # evaluates the head of a rule (returns the implied FX term or NS)
{\normalfont 
\begin{lstlisting}[style=algorithm,escapechar=|,caption={Pseudocode of  check function.},label={alg:check}]
function check(map, bgp): Bool
    for node $\rightarrow$ term in map:
        for rule in rules:
            if body (rule, node, map) then:
                result $\leftarrow$ head(rule, node, map)
                if result = NS or result in disjointness[term]:
                    return False
    return True
\end{lstlisting}}
\end{algorithm}

\textbf{Time complexity of \texttt{check}.}
    The time complexity of the function \texttt{check} is $O(n * r )$ where n is the number of nodes in the bgp and r is the number of rules, i.e. \ref{crule:last}. Since we can consider r a constant parameter, then  we can conclude that \texttt{check} is linear in the number of nodes of the bgp.

\textbf{match(n1,n2).} We define a function $\tt{match(n1, n2):Bool}$ that implements the SPARQL specification (SPARQL Basic Graph Pattern Matching)~\cite{w3c:sparql11}. The function compares the nodes' types, if any of the two is a variable, returns True. If both aren't, it checks node equality (see Table~\ref{tab:match}).

\begin{table}[htpb]
    \centering
    \caption{The \textit{match(left,right)} function compares two nodes}
    \label{tab:match}
    \begin{tabular}{ccc}
    \hline
       \textbf{left} & \textbf{right} & \textbf{matches(left,right)} \\\hline\hline
       \VAR{} & \VAR{} & True\\
       \VAR{} & \IRI{} & True\\
       \VAR{} & \LIT{} & True\\
       \IRI{} & \IRI{} & =\\
       \IRI{} & \LIT{} & False\\
       \LIT{} & \LIT{} & =\\
       \hline
    \end{tabular}
\end{table}

\begin{algorithm}[Function satisfiesUniquePathToRoot]\label{al:satisfiesUniquePathTo}
\begin{sloppypar}We present the function \texttt{satisfiesUniquePathToRoot(bgp, map, node):Bool} \chadded{(see Listing~\ref{alg:satisfiesUniquePathToRoot} for the pseudocode)} which guarantees that the constraint \ref{crule:last} is satisfied.
The function takes as argument (i) $bgp$ - a basic graph pattern; (ii) $map \in \tt{Dict(\mathcal{N}, terms)}$ - nodes with annotations (e.g. `?x -> Slot`); and (iii) a focus node $n$ so that $(n\rightarrow~Container) \in map$ or $(n\rightarrow~FXRoot) \in map$. 
The algorithm operates as follows:\end{sloppypar}
\begin{enumerate}
    
    \item From bgp, construct the set \texttt{triplePairs} containing pairs of distinct triples that have \texttt{node} as object (cf. line \ref{line:constructTriplePairs}). This can be an empty set. 
    
    \item We collect backward triple paths. For each pair of triples so that $Triple(\_,\_,node)$, and each triple in the pair, get the lists of triples having recursively S$\bowtie$O joins from the BGP (cf. function \texttt{walkBackwards}).
    $backwardsTriplePathPairs$ is a list of tuples, holding two backwards paths for each triple pair.
    
    \item Next, for each triple path pair, we transform each triple path in a sequence of backward node paths starting from $node$ (cf. function \texttt{nodesPath}).
    $backwardsNodePathPairs$ is a list of tuples, holding two backwards node paths for each pair.
    
    \item We walk each paths' pair.
    
    \item Paths can have the same length or not. When there is a node in both paths, we check whether they are compatible, ie. we match the nodes in the same position, relying on the $match$ function defined before. If two nodes don't match, the bgp does not satisfy the  constraint at \ref{crule:last}.
    
    \item If length does not match (lines~\ref{line:begindifferentlength}-\ref{line:enddifferentlength}), we verify that the shortest path of the two ends with a node that is not being annotated as FXRoot. Since both paths must be compatible for the bgp to be satisfiable, if the shortest path ends with a root container, the bgp is not satisfiable (cf. line~26).
    
    \item If the algorithm proceeds, all conditions are satisfied: the solution graph does not need to have alternative unmatching paths to a container or root node.
    
\end{enumerate}
\begin{figure}
{\normalfont\begin{lstlisting}[style=algorithm,escapechar=|,caption={Pseudocode of the satisfiesUniquePathToRoot function.},label={alg:satisfiesUniquePathToRoot}]
function satisfiesUniquePathToRoot(bgp, map, node): Bool
    triplePairs $\leftarrow$ {((s$_1$, p$_1$, node), (s$_2$, p$_2$, node)) $|$ (s$_1$, p$_1$, node) $\in$ bgp $\wedge$ (s$_2$, p$_2$, node) $\in$ bgp $\wedge$ (( s$_1$ $\neq$ s$_2$) $\vee$ ( p$_1$ $\neq$ p$_2$))}|\label{line:constructTriplePairs}|
    backwardsTriplePathPairs $\leftarrow$ {(walkBackwards([left], bpg), walkBackwards([right], bgp)) $|$ (left, right) $\in$ triplePairs}
    backwardsNodePathPairs $\leftarrow$ {(nodesPath(left), nodesPath{right) $|$ (left, right) $\in$ backwardsTriplePathPairs}
    for (leftList, rightList) in backwardsNodePathPairs:
        for i $\in$ [0$\cdots$max($|$leftList$|$,$|$rightList$|$)]
            if $|$leftList$|$=$|$rightList$|$
                if match(leftList[i], rightList[i]) = False: 
                    return False
            else: |\label{line:begindifferentlength}|
                if i < $|$leftList$|$:
                    shortestEndNode = rightList[$|$rightList$|$]
                else:
                    shortestEndNode = leftList[$|$leftList$|$]
                if (shortestEndNode, _, r) \in bgp $\wedge$ map[r] = FXRoot:
                    return False |\label{line:enddifferentlength}|
    return True

function walkBackwards(list, bgp): List
    $\tt{(s, p, o)} \leftarrow \tt{list}[|\tt{list}|]$
    if $\tt{(s_p, p_p, s)} \in bgp$:
        $\tt{list} \leftarrow \tt{list} + [\tt{(s_p, p_p, s)}]$
        return walkBackwards(list, bgp)
    else:
        return list
        
function nodesPath(triplePath): List
    path = []
    for (s, p, o) $\in$ triplePath:
        if path = []: 
            path = [o]
        else: path = 
            path + [p, s]
    return path
    \end{lstlisting}}\end{figure}
\end{algorithm}

\paragraph*{}\textbf{Time complexity of \texttt{satisfiesUniquePathToRoot}.}
    In the worst case, the function \textbf{satisfiesUniquePathToRoot} constructs a path pair for each node. 
    Since the dominant part of the function involves going through pairs of paths and checking that each node does not invalidate any conditions, the time complexity of the function is $O(n^2)$.

\subsection{Correctness and Completeness with respect to \FX{}}

To guarantee the satisfiability of a BGP, it is necessary and sufficient to verify the conditions \ref{conditionfirst}-\ref{conditionlast}.
Here, we summarise how all the conditions are captured by our framework.

\textbf{ \ref{condition:nocycles}, \ref{condition:joins} and \ref{condition:multiplepaths}.}
As discussed earlier, \ref{condition:nocycles}, \ref{condition:joins} and \ref{condition:multiplepaths} are verified by executing the functions \texttt{hasCycle}, \texttt{hasUnsupportedJoin}, and \texttt{satisfiesUniquePathToRoot}.

The other conditions (\ref{condition:mappingrules}, \ref{condition:triplepattern}, \ref{condition:singleroot}, \ref{condition:proxygraph}) are verified by the \texttt{check} function via the inference rules \ref{crule:first}-\ref{crule:last}, as follows.

\textbf{\ref{condition:mappingrules} and \ref{condition:triplepattern}.}
\ref{crule:first}-\ref{crule:lasttp} guarantee that the BPG's triple patterns comply with the \FX{} triple patterns and mapping rule: %, see \ref{condition:triplepattern} and \ref{condition:mappingrules}:
\begin{itemize}
    \item \ref{crule:1} implies that all the subject must be a Container;
    \item \ref{crule:3}, \ref{crule:5},  \ref{crule:6}, \ref{crule:7}, \ref{crule:8} and \ref{crule:9} imply that the predicate of a BGP's TP is either a Slot or TypeProperty.
    \item \ref{crule:5}, \ref{crule:10}, \ref{crule:11} imply that the object of a TP is a Container or a Type or a Value or FXRoot.
\end{itemize}
Note that if multiple incompatible predicates are associated with a node, the function \texttt{check} raises NS.

So far, satisfying the conditions guarantees that satisfiable \FX{} TPs are identified.

Next, we must demonstrate that unsatisfiable TPs are excluded.
To this end, we go through the unsatisfiable patterns of Table~\ref{tab:fxtripleppatterns}, to demonstrate that they are excluded by the rules.
\texttt{C.SN.T}, \texttt{C.SN.R}, \texttt{C.SS.T}, \texttt{C.SS.R} are excluded by the \ref{crule:11}, which requires the Object of a Slot to be a Container that is disjoint with Type and FXRoot.
\texttt{C.TP.V} is excluded by the \ref{crule:12}, which requires the property of a triple having Value as Object to be a Slot (and Slot is disjoint with TypeProperty).
\texttt{C.TP.C} is excluded by the rule \ref{crule:16}, which requires the property of a triple having Container as Object to be a Slot.
Thus, we have shown how all unsatisfiable triple patterns are identified by the rules.

\textbf{\ref{condition:singleroot}.} Condition \ref{condition:singleroot} is guaranteed by the \ref{crule:2} and \ref{crule:19}. In fact, \ref{crule:2} guarantees that the IRI \texttt{fx:root} is annotated as FXRoot and \ref{crule:19} guarantees that there can not be multiple path leading to it.

\textbf{\ref{condition:proxygraph}.} Last, we go through the Proxy Graph reported in Table~\ref{tab:proxy_patterns} and we demonstrate that are all excluded by the rules.
Proxy Graph for Formulas \ref{def:fx:basic}.\ref{def:fx:basic:1}, \ref{def:fx:disjointness}.\ref{def:fx:disjointness:3}, \ref{def:fx:disjointness}.\ref{def:fx:disjointness:10} raises a NS since \texttt{<i1>} is annotate as both Container (cf. \ref{crule:1}) and Type (cf. \ref{crule:3} and \ref{crule:5}) and they are disjoint predicates.
\ref{def:fx:slots}.\ref{lst:fx:slots:6}, \ref{def:fx:slots}.\ref{lst:fx:slots:8}, and \ref{def:fx:slots}.\ref{lst:fx:slots:2} are excluded by  \ref{crule:17} since we have two triples have the same subject and predicate, the predicate is a Slot, but the objects do not match.
\ref{def:fx:slots}.\ref{lst:fx:slots:7} and \ref{def:fx:slots}.\ref{lst:fx:slots:9} are cycles, then excluded by the function \texttt{hasCycle}.
\ref{def:fx:slots}.\ref{lst:fx:slots:10} is excluded by the function \texttt{satisfiesUniquePathToRoot}.
\ref{def:fx:root}.\ref{def:fx:root:6} is excluded by \ref{crule:18} since it is  a subject/object join when the object is FXRoot.
Finally,  \ref{def:fx:root}.\ref{def:fx:root:7} guarantees that the BGP has a single root and, as discussed earlier, this is enforced by \ref{crule:2} and \ref{crule:19}.

\subsection{Annotating graph patterns}
We investigate two approaches to annotate basic graph patterns, ie. implementing the function of definition~\ref{def:fx:annotate}: (1) a top-down algorithm, addressing the annotation process as a search, and (2) a bottom-up algorithm, addressing the annotation process as constraint satisfaction one.
% Two approaches: Search (top-down) or CSP (bottom-up).

\subsubsection{Top-down algorithm}

A top-down algorithm could address the annotation of the BGP problem as searching the space of annotated patterns for any (or all) annotations that comply with the theory. If none is found, the input $bgp$ is not satisfiable. Clearly, it would be sufficient to find one viable solution for the BGP to be satisfiable. However, the implementation allows for finding all possible valid annotated BGPs (by setting the argument complete as True).

\begin{algorithm}[Top-down algorithm: annotation as search]\label{al:search} The algorithm \chadded{(whose pseudocode is provided in Listing~\ref{alg:annotate_topdown})} operates as follows:
\begin{enumerate}
    
    %\item A search function (cf. line 1) selects all nodes from the bgp. 
    
    \item It initialise an initial annotation $start \in Map(\mathcal{N}, terms)$, where each node is associated with its position in the triple it occurs in. For instance, given the bgp $ \{?a~?b~?c~.~?c~?d~?e\} $, it would generate the map $\{(~?a\rightarrow~Subject~)$ $(~?b\rightarrow~Property~)$ $(~?c\rightarrow~Object~)$ $(~?d\rightarrow~Property~)$ $(~?e\rightarrow~Object~) \}$;
    
    \item Next, the recursive function \texttt{search} is called (cf. line~3);
    
    \item For each node in the map, we first lookup for possible specialisations, and generate a copy of the map with a modified annotation (cf. line~9);
    
    \item Next, we verify that the set of annotations is consistent with the constraints defined in~\ref{def:inferencerules} (cf. line 10);
    
    \item If the bgp is consistent, and only ground predicates are mentioned in the annotation, the bgp is satisfiable. The current valid map is added to the solutions (cf. line 12). If $complete$ is $False$, the solution is returned and the process interrupted, otherwise, we keep generating alternative solutions (cf. lines 6-7);
    
    \item If the bgp is consistent but some annotations can still be specialised, we keep searching within the current space (cf. line 16).
\end{enumerate}

{\normalfont 
\begin{lstlisting}[style=algorithm, caption={Pseudocode of annotate\_topdown function.}, label={alg:annotate_topdown}]
function annotate_topdown(bgp, complete) : List   
    map $\leftarrow$ { node $\rightarrow$ Subject |  (node,_,_)$\in$bgp } $\cup$ { node $\rightarrow$ Predicate | (_,node,_)$\in$bgp } $\cup$ { node $\rightarrow$ Object |  (_,_,node)$\in$bgp }
    return search ( map , bgp, {}, complete)
    
function search (map , bgp, solutions, complete): List
    for node $\rightarrow$ term in map:
        for (subterm $\rightarrow$ term) in specialisations:
            newMap $\leftarrow$ copy(map)
            newMap[node] $\leftarrow$ subterm
            if check(newMap, bgp):
                if {term | node $\leftarrow$ term $\in$ newMap} $\subseteq$ GP:
                    solutions $\leftarrow$ solutions + [newMap]
                    if complete:
                        return solutions
                else:
                    solutions $\leftarrow$ solutions + search(newMap, bgp, solutions, complete)
    return solutions
\end{lstlisting}
}

\end{algorithm}

\textbf{Time complexity of \texttt{annotate\_topdown}.}
    The dominant part of the algorithm is the search function. This function specialises multiple times the nodes of the bgp and checks the validity of the annotation every time a term is specialised.
    Therefore, the complexity of the \texttt{annotate\_topdown} depends on the number of nodes of the bgp (n), the number of specialisation relation in the predicate hierarchy |specialisations| and the complexity of \texttt{check}, i.e. $O(n*r)$. 
    Furthermore, if the solution contains predicates other than ground ones, the function specialises the solution further.
    The process is repeated at most $|terms| - |GP|$ times.
    Therefore, the complexity of \texttt{annotate\_topdown} is quadratic in the number of nodes in the bgp and linear in the number of specialisations and inference rules, i.e. $O(n^2*r*|specialisations|*(|terms| - |GP|))$.

\subsubsection{Bottom-up algorithm}
The second algorithm uses the same components (inference rules) of the previous one, but operates by first annotating bgp triples with all \textit{grounded predicates} specialising Subject, Property, or Object, and then filters out inconsistent combinations.

\begin{algorithm}\label{algorithm2}
[Bottom-up algorithm] 
A \textit{generate} function produces valid annotations of a bgp. When the parameter \textit{complete} equals to True, it stops when one valid solution is found. 
\begin{enumerate}

    \item We extract the list of unique nodes from the bgp (cf. line~2);
    
    \item We populate a list with, for each node, the set of \textit{ground predicates} that could match the nodes' position in any of the triples of the bgp (cf. lines~5-13). Note that $|nodesTerms|=|nodes|$ and for any index $i$ in the list, $nodes[i] \rightarrow~nodeTerms[i]$. \chadded{We assume that $nodes[n]$ where n is a node of the BGP returns the set of ground predicates associated with n.}

    \item \chadded{We refine the ground predicates by some of  the inference rules (cf. Definition~\ref{def:inferencerules}) whose application conditions are at the triple level (lines~15-22).}
    
    \item We generate possible annotations in line~16. $hypotheses$ is a set of lists (tuples) of size $|nodes|$, each one with a possible combinations of ground predicates (line~24). 
    
    \item In the next steps, we verify each combination against our conditions with the check function (cf. lines~26-32).
    
    \item If the argument \textit{complete} is set to False, the process ends when the first valid solution is found (cf. lines~29-32).
    
    \item Finally, all valid solutions are returned, or an empty set, if no solution was found (cf. line~33).
\end{enumerate}

\begin{figure}
{\normalfont 
\begin{lstlisting}[style=algorithm,caption={Pseudocode of the bottom-up alogorithm.}]
function generate_bottomup(bgp, complete):
    nodes $\leftarrow$  [{s|(s,_,_)$\in$bpg} $\cup$ {p|(_,p,_)$\in$bpg} $\cup$ {o|(_,_,o)$\in$bpg}]
    nodesTerms $\leftarrow$  []
    
    for node in nodes:
        t $\leftarrow$ {}
        if (node, p, o) $\in$ bgp:
            t $\leftarrow$ { term | term $\in$ groundTerms $\wedge$ topTerm(term) = 'Subject' }
        if (s, node, o) $\in$ bgp:
            t $\leftarrow$ t + { term | term $\in$ groundTerms $\wedge$ topTerm(term) = 'Predicate' }
        if (s, p, node) $\in$ bgp:
            t $\leftarrow$ t + { term | term $\in$ groundTerms $\wedge$ topTerm(term) = 'Object' } 
        nodesTerms $\leftarrow$ nodesTerms + t
        
    for (s, p, o) $\in$ bgp:
        nodes[s] = {'Container'}
        if p is an IRI:
            if p $\in$ TypeProperty:
                nodes[p] = {'TypeProperty'}
                nodes[o] = nodes[o] \ {'Container', 'Value'}
            else:
                nodes[o] = nodes[o] \ {'Type', 'Root'}
            
    hypotheses $\leftarrow$ nodeTerms[1] $\times$ nodeTerms[2]  $\times$ $\cdots$ $\times$ nodeTerms[|nodes|]
    
    solutions = {}
    for h $\in$ hypotheses:
        map = {nodes[i] $\rightarrow$ h[i] |1$\leq$i$\leq$|nodes|}
        if check(map, bgp):
            solutions $\leftarrow$ solutions $\cup$ {map}
            if complete = False:
                return solutions
    return solutions
\end{lstlisting}
}\end{figure}
\end{algorithm}

\textbf{Time complexity of \texttt{generate\_bottomup}}
    The dominant part of the algorithm is the iteration at lines 27-32. This goes through the hypotheses, which are $n * |GP|$. Then, for each hypothesis, it generates an annotation (map) and checks its validity. Therefore, the time complexity of \texttt{generate\_bottomup} is $O(n^2 * |GP|*r)$.

%\subsection{The problem of Façade-X BGP satisfiability is decidable}

%To demonstrate that our algorithms are complete:
%can we say that all possible valid BGP$^{fx}$ are combinations of the allowed joins of Table~\ref{tab:fxgraphpatterns}?
%If yes, can we simply show that all remaining joins from Table~\ref{tab:sparql_joins} are invalid according to any of the definitions/theorems we presented before?

% \section{Façade-X profiles \FX}\label{sec:profiles}\input{sections/profiles}

\section{Experiments}\label{sec:experiments}\chadded{We report on three experiments. 
The first compares the two algorithms described in Section~\ref{sec:algorithms} with a curated set of BGPs of increasing complexity, and provide empirical evidence of their performance. 
The second experiment focuses on real world queries, and has the objective of testing the relevance and practical feasibility of the approach.
The third experiment shows the advantages of checking whether a query can be satisfied before evaluating it, considering the case of a knowledge graph construction benchmark.}

\subsection{Setting}%\todo{Moved before the experiments as it is the same for both}
We prepare a set of BGPs with different characteristics and consider two regimes: (a) find the first satisfiable annotation (ensure satisfiability); or (b) generate all satisfiable annotations (describe solution patterns).

The algorithms are implemented in Java. 
Reported results were performed on a MacBook Pro \chdeleted{[MacBookPro16,1]} with 8-Core Intel Core i9 (2.3 GHz) processor and Java 17 and Maven 3.9.
We report on results of ten executions, aggregated with average duration (in milliseconds) and standard deviation. 
Since the satisfiability assessment will be part of a process performed at query time, we set the upper limit of 5 seconds, after which we stop the execution and consider the algorithm not viable.
We leave possible optimisations of the algorithms and experimenting with a broader class of BGPs to future work.

\chadded{The proof-of-concept implementation and the experimental code is published and experiments can be reproduced~\cite{fxbgp}.}

\subsection{Comparing two algorithms}
We report on experiments with the two algorithms described in Section~\ref{sec:algorithms}. Both implementations perform all checks described in Section~\ref{sec:algorithms} (unsupported joins and cycles), before proceeding with the annotation of the BGPs.
With these experiments, we want to demonstrate the practical feasibility of the approach to solving \FX~BGP satisfiability problem.

\begin{table}[t]
    
\centering
    \caption{List of BGPs used in experiments}\label{tab:files}
    \begin{tabular}{lccp{5cm}}
    \hline 
\textbf{BGP label} & \textbf{Satisfiable?} &
\textbf{Triple patterns} &\textbf{Features} \\\hline \hline 
N\_1T   & No & 1 & only variables, no joins  \\
N\_2J   & No & 2 & only variables, with joins  \\
N\_2P\_R & No & 2 & multiple paths to fx:root \\
N\_2T   & No & 2 & only variables, no joins  \\
N\_3J   & No & 3 & only variables, no joins  \\
N\_3P\_C & No & 3 & multiple paths to a container \\
N\_3P\_R & No & 3 & multiple paths to fx:root \\
N\_3T & No & 3 & only variables, no joins  \\
N\_4J & No & 3 & only variables, with joins  \\
N\_4P\_C & No & 4 & multiple paths to a container \\
N\_4T & No & 4 & only variables, no joins  \\
N\_5J & No & 5 & only variables, with joins  \\
N\_5P\_C & No & 5 & multiple paths to a container \\
N\_5T & No & 5 & only variables, no joins  \\
S\_1T & Yes & 1 & only variables, no joins  \\
S\_2J & Yes & 2 & only variables, with joins  \\
S\_2P\_R & Yes & 2 & multiple paths to fx:root \\
S\_2T & Yes & 2 & only variables, no joins  \\
S\_3J & Yes & 3 & only variables, with joins  \\
S\_3P\_C & Yes & 3 & multiple paths to a container \\
S\_3T & Yes & 3 & only variables, no joins  \\
S\_4J & Yes & 4 & only variables, with joins  \\
S\_4P\_C & Yes & 4 & multiple paths to a container \\
S\_4T & Yes & 4 & only variables, no joins  \\
S\_5P\_C & Yes & 5 & multiple paths to a container \\
S\_5T & Yes & 5 & only variables, no joins  \\\hline
\end{tabular}
\end{table}

\subsubsection{Data}
We generated 27 test BGPs with different characteristics.
These are listed in Table~\ref{tab:files}.
BGP labels are informative:
\begin{itemize}
    \item S: the bgp is satisfiable
    \item N: the bgp is not satisfiable
    \item \textit{number}: the number of triples in the bgp
    \item T: the bgp contains only variables without joins
    \item J: the bgp contains one or more joins
    \item P: the bgp contains a join on an object and must satisfy the single-path-to-root condition
    \item R: the bgp contains a join on an object that is fx:root
    \item C: the bgp contains a join on an object that is a container
\end{itemize}

\begin{table}

\centering
\caption{Results for algorithm Top down (Search): only satisfiability}\label{tab:ex:sea:sat}

\begin{tabular}{l c c c c r r r r}
\hline
\textbf{name} & \textbf{satisfiable?} & \textbf{found} & \textbf{type} & \textbf{size} & \textbf{\thead{ms\\ (avg)}} & \textbf{\thead{ms\\ (std)}} & \textbf{\thead{tested\\ (avg)}} & \textbf{\thead{tested\\ (std)}} \\
\hline\hline
N\_1T & FALSE & 0 & T & 1 & 1.7 & 0.67 & 12 & 0 \\
\hline
N\_2J & FALSE & 0 & J & 2 & 0 & 0 & 0 & 0 \\
\hline
N\_2P\_R & FALSE & 0 & P & 2 & 71 & 20.11 & 1957 & 0 \\
\hline
N\_2T & FALSE & 0 & T & 2 & 478.8 & 84.63 & 16364 & 0 \\
\hline
N\_3J & FALSE & 0 & J & 3 & 0.1 & 0.32 & 0 & 0 \\
\hline
N\_3P\_C & FALSE & 0 & P & 3 & 2446.7 & 211.66 & 109601 & 0 \\
\hline
N\_3P\_R & FALSE & 0 & P & 3 & 266.7 & 16.77 & 12330 & 0 \\
\hline
\rowcolor{lightgray}N\_3T & FALSE & -1 & T & 3 & -1 & 0 & -1 & 0 \\
\hline
N\_4J & FALSE & 0 & J & 4 & 0 & 0 & 0.9 & 0.74 \\
\hline
\rowcolor{lightgray}N\_4P\_C & FALSE & -1 & P & 4 & -1 & 0 & -1 & 0 \\
\hline
\rowcolor{lightgray}N\_4T & FALSE & -1 & T & 4 & -1 & 0 & -1 & 0 \\
\hline
N\_5J & FALSE & 0 & J & 5 & 0 & 0 & 5.4 & 1.65 \\
\hline
\rowcolor{lightgray}N\_5P\_C & FALSE & -1 & P & 5 & -1 & 0 & -1 & 0 \\
\hline
\rowcolor{lightgray}N\_5T & FALSE & -1 & T & 5 & -1 & 0 & -1 & 0 \\
\hline
S\_1T & TRUE & 1 & T & 1 & 0.6 & 0.52 & 31 & 3.2 \\
\hline
S\_2J & TRUE & 1 & J & 2 & 0.5 & 0.53 & 16.3 & 1.49 \\
\hline
S\_2P\_R & TRUE & 1 & P & 2 & 17.9 & 1.73 & 636.3 & 15.81 \\
\hline
S\_2T & TRUE & 1 & T & 2 & 0.3 & 0.48 & 14.8 & 2.49 \\
\hline
S\_3J & TRUE & 1 & J & 3 & 1146.4 & 68.34 & 193106 & 1468.36 \\
\hline
S\_3P\_C & TRUE & 1 & P & 3 & 73.4 & 7.35 & 13808.7 & 250.37 \\
\hline
S\_3T & TRUE & 1 & T & 3 & 0.2 & 0.42 & 57.9 & 5.74 \\
\hline
S\_4J & TRUE & 1 & J & 4 & 2671.5 & 196.85 & 439371.1 & 3103.24 \\
\hline
\rowcolor{lightgray}S\_4P\_C & TRUE & -1 & P & 4 & -1 & 0 & -1 & 0 \\
\hline
S\_4T & TRUE & 1 & T & 4 & 0.3 & 0.48 & 90.4 & 9.58 \\
\hline
\rowcolor{lightgray}S\_5P\_C & TRUE & -1 & P & 5 & -1 & 0 & -1 & 0 \\
\hline
S\_5T & TRUE & 1 & T & 5 & 1 & 0 & 173.3 & 16.75 \\
\hline

\end{tabular}

\end{table}

\begin{comment}
%%%%% First round of experiments
\begin{table}

\begin{tabular}{lc  c  c  rrrrr}
\hline
\textbf{name} & \textbf{satisfiable} & \textbf{found} & \textbf{type} & \textbf{size} & \textbf{\thead{ms\\(avg)}} & \textbf{\thead{ms\\(std)}} & \textbf{\thead{gen.\\(avg)}} & \textbf{\thead{gen.\\(std)}} \\\hline\hline
S\_1T & TRUE & 6 & T & 1 & 4.5 & 0.53 & 159.7 & 12.75 \\

S\_2J & TRUE & 36 & J & 2 & 527.1 & 117.26 & 66061.7 & 3164.61 \\

S\_2P\_R & TRUE & 1 & P & 2 & 26.2 & 1.55 & 5778.1 & 334.17 \\

S\_2T & TRUE & 36 & T & 2 & 2930 & 151.26 & 585926.3 & 8903.79 \\

\rowcolor{lightgray}S\_3J & TRUE & -1 & J & 3 & -1 & >5s & -1 & - \\

\rowcolor{lightgray}S\_3P\_C & TRUE & -1 & P & 3 & -1 & >5s & -1 & - \\

\rowcolor{lightgray}S\_3T & TRUE & -1 & T & 3 & -1 & >5s & -1 & - \\

\rowcolor{lightgray}S\_4J & TRUE & -1 & J & 4 & -1 & >5s & -1 & - \\

\rowcolor{lightgray}S\_4P\_C & TRUE & -1 & P & 4 & -1 & >5s & -1 & - \\

\rowcolor{lightgray}S\_4T & TRUE & -1 & T & 4 & -1 & >5s & -1 & - \\

\rowcolor{lightgray}S\_5P\_C & TRUE & -1 & P & 5 & -1 & >5s & -1 & - \\

\rowcolor{lightgray}S\_5T & TRUE & -1 & T & 5 & -1 & >5s & -1 & - \\
\hline
\end{tabular}

\centering
\caption{Results for algorithm Top down (Search): generate all satisfiable annotations}\label{tab:ex:sea:all}

\end{table}
\end{comment}

%%%% Second round of experiments

\begin{table}

\centering
\caption{Results for algorithm Top down (Search): generate all satisfiable annotations}\label{tab:ex:sea:all}

\begin{tabular}{l c c c c r r r r}
\hline
\textbf{name} & \textbf{satisfiable?} & \textbf{found} & \textbf{type} & \textbf{size} & \textbf{\thead{ms\\ (avg)}} & \textbf{\thead{ms\\ (std)}} & \textbf{\thead{tested\\ (avg)}} & \textbf{\thead{tested\\ (std)}} \\
\hline\hline
S\_1T & TRUE & 6 & T & 1 & 6 & 0.67 & 187.4 & 17.88 \\
\hline
S\_2J & TRUE & 36 & J & 2 & 541.9 & 62.94 & 72895.7 & 795.27 \\
\hline
S\_2P\_R & TRUE & 1 & P & 2 & 33.8 & 9.38 & 6445.1 & 291.25 \\
\hline
S\_2T & TRUE & 32.3 & T & 2 & 2917.4 & 1055.36 & 603611 & 212133.29 \\
\hline
\rowcolor{lightgray}S\_3J & TRUE & -1 & J & 3 & -1 & 0 & -1 & 0 \\
\hline
\rowcolor{lightgray}S\_3P\_C & TRUE & -1 & P & 3 & -1 & 0 & -1 & 0 \\
\hline
\rowcolor{lightgray}S\_3T & TRUE & -1 & T & 3 & -1 & 0 & -1 & 0 \\
\hline
\rowcolor{lightgray}S\_4J & TRUE & -1 & J & 4 & -1 & 0 & -1 & 0 \\
\hline
\rowcolor{lightgray}S\_4P\_C & TRUE & -1 & P & 4 & -1 & 0 & -1 & 0 \\
\hline
\rowcolor{lightgray}S\_4T & TRUE & -1 & T & 4 & -1 & 0 & -1 & 0 \\
\hline
\rowcolor{lightgray}S\_5P\_C & TRUE & -1 & P & 5 & -1 & 0 & -1 & 0 \\
\hline
\rowcolor{lightgray}S\_5T & TRUE & -1 & T & 5 & -1 & 0 & -1 & 0 \\
\hline

\end{tabular}

\end{table}

\begin{table}
\centering
\caption{Results for algorithm Bottom up (CSP): only satisfiability}\label{tab:ex:csp:sat}

\begin{tabular}{ lcccrrrr}
\hline
\textbf{name} & \textbf{satisfiable?} & \textbf{found} & \textbf{type} & \textbf{size} & \textbf{\thead{ms\\ (avg)}} & \textbf{\thead{ms\\ (std)}} & \textbf{tested} \\
\hline\hline
N\_1T & FALSE & 0 & T & 1 & 17.5 & 1.84 & 2 \\
\hline
N\_2J & FALSE & 0 & J & 2 & 0.3 & 0.48 & 2 \\
\hline
N\_2P\_R & FALSE & 0 & P & 2 & 6.2 & 0.79 & 36 \\
\hline
N\_2T & FALSE & 0 & T & 2 & 3.4 & 0.7 & 24 \\
\hline
N\_3J & FALSE & 0 & J & 3 & 0 & 0 & 24 \\
\hline
N\_3P\_C & FALSE & 0 & P & 3 & 2.9 & 0.88 & 18 \\
\hline
N\_3P\_R & FALSE & 0 & P & 3 & 8 & 3.16 & 36 \\
\hline
N\_3T & FALSE & 0 & T & 3 & 16.6 & 2.46 & 288 \\
\hline
N\_4J & FALSE & 0 & J & 4 & 0 & 0 & 288 \\
\hline
N\_4P\_C & FALSE & 0 & P & 4 & 2.8 & 0.92 & 54 \\
\hline
N\_4T & FALSE & 0 & T & 4 & 93.8 & 9.35 & 3456 \\
\hline
N\_5J & FALSE & 0 & J & 5 & 0.1 & 0.32 & 3456 \\
\hline
N\_5P\_C & FALSE & 0 & P & 5 & 7.4 & 1.58 & 162 \\
\hline
N\_5T & FALSE & 0 & T & 5 & 536.1 & 78.69 & 41472 \\
\hline
S\_1T & TRUE & 1 & T & 1 & 0.5 & 0.53 & 12 \\
\hline
S\_2J & TRUE & 1 & J & 2 & 0.1 & 0.32 & 144 \\
\hline
S\_2P\_R & TRUE & 1 & P & 2 & 0.6 & 0.52 & 36 \\
\hline
S\_2T & TRUE & 1 & T & 2 & 0.4 & 0.52 & 144 \\
\hline
S\_3J & TRUE & 1 & J & 3 & 1.4 & 0.7 & 432 \\
\hline
S\_3P\_C & TRUE & 1 & P & 3 & 0.8 & 0.42 & 18 \\
\hline
S\_3T & TRUE & 1 & T & 3 & 0.9 & 0.32 & 1728 \\
\hline
S\_4J & TRUE & 1 & J & 4 & 12.1 & 2.42 & 5184 \\
\hline
S\_4P\_C & TRUE & 1 & P & 4 & 0.5 & 0.53 & 54 \\
\hline
S\_4T & TRUE & 1 & T & 4 & 5.1 & 0.99 & 20736 \\
\hline
S\_5P\_C & TRUE & 1 & P & 5 & 0.3 & 0.48 & 162 \\
\hline
S\_5T & TRUE & 1 & T & 5 & 361.4 & 46.9 & 248832 \\
\hline

\end{tabular}

\end{table}

%%%%%% First round of experiments
\begin{comment}
\begin{table}
\begin{tabular}{lc  c  c  rrrr}
\hline
\textbf{name} & \textbf{satisfiable} & \textbf{found} & \textbf{type} & \textbf{size} & \textbf{\thead{ms\\(avg)}} & \textbf{\thead{ms\\(std)}} & \textbf{\thead{gen.}}\\\hline\hline
S\_1T & TRUE & 6 & T & 1 & 0 & 0 & 12 \\

S\_2J & TRUE & 36 & J & 2 & 0.9 & 0.57 & 144 \\

S\_2P\_R & TRUE & 1 & P & 2 & 0.4 & 0.52 & 36 \\

S\_2T & TRUE & 36 & T & 2 & 0.9 & 0.32 & 144 \\

S\_3J & TRUE & 60 & J & 3 & 6.1 & 1.45 & 1728 \\

S\_3P\_C & TRUE & 4 & P & 3 & 1.5 & 0.53 & 432 \\

S\_3T & TRUE & 216 & T & 3 & 9.7 & 2 & 1728 \\

S\_4J & TRUE & 300 & J & 4 & 76.5 & 5.72 & 20736 \\

S\_4P\_C & TRUE & 8 & P & 4 & 22.4 & 5.27 & 5184 \\

S\_4T & TRUE & 1296 & T & 4 & 133.7 & 24.73 & 20736 \\

S\_5P\_C & TRUE & 16 & P & 5 & 290.5 & 20.88 & 62208 \\

S\_5T & TRUE & 7776 & T & 5 & 1564.9 & 172.77 & 248832 \\
\hline
\end{tabular}
\centering
\caption{Results for algorithm Bottom up (CSP): generate all satisfiable annotations}\label{tab:ex:csp:all}
\end{table}
\end{comment}

%%%%%%%%% Second round of experiments Bottom-Up all solutions

\begin{table}
\centering
\caption{Results for algorithm Bottom up (CSP): generate all satisfiable annotations}\label{tab:ex:csp:all}

\begin{tabular}{lcccrrrr}
\hline
\textbf{name} & \textbf{satisfiable} & \textbf{found} & \textbf{type} & \textbf{size} & \textbf{\thead{ms\\(avg)}} & \textbf{\thead{ms\\(std)}} & \textbf{\thead{gen.}}\\\hline\hline
S\_1T & TRUE & 6 & T & 1 & 1 & 0.47 & 12 \\
\hline
S\_2J & TRUE & 36 & J & 2 & 1.1 & 0.32 & 144 \\
\hline
S\_2P\_R & TRUE & 1 & P & 2 & 0.4 & 0.52 & 36 \\
\hline
S\_2T & TRUE & 36 & T & 2 & 1.1 & 0.32 & 144 \\
\hline
S\_3J & TRUE & 60 & J & 3 & 3 & 0.47 & 432 \\
\hline
S\_3P\_C & TRUE & 4 & P & 3 & 0.5 & 0.53 & 18 \\
\hline
S\_3T & TRUE & 216 & T & 3 & 15.1 & 1.85 & 1728 \\
\hline
S\_4J & TRUE & 300 & J & 4 & 36.7 & 5.48 & 5184 \\
\hline
S\_4P\_C & TRUE & 8 & P & 4 & 0.9 & 0.32 & 54 \\
\hline
S\_4T & TRUE & 1296 & T & 4 & 171.9 & 31.8 & 20736 \\
\hline
S\_5P\_C & TRUE & 16 & P & 5 & 2.2 & 0.63 & 162 \\
\hline
S\_5T & TRUE & 7776 & T & 5 & 1754.5 & 270.5 & 248832 \\
\hline

\end{tabular}
\end{table}

%%%%%%%%%

\subsubsection{Results}
%\todo[inline]{Results of new experiments added. Check text coherency with new tables.}
Results are shown in \chreplaced{tables~\ref{tab:ex:sea:sat}-\ref{tab:ex:csp:all}.}{}
Table~\ref{tab:ex:sea:sat} show results for the Top down (Search) algorithm when doing satisfiability check, ie. the algorithm stops when the first satisfiable annotation is found.
We report the speed in average milliseconds and the standard deviation.
Furthermore, we show the number of solution patterns that are evaluated. When reporting intermediate solutions evaluated, we include both average and standard deviation, as the algorithm is not fully deterministic, as the temporary search states are not ordered and the process stops when the first solution is found (without completing the search).
%The number in parenthesis after the duration in ms shows the number of annotations in the search space.
It can be seen that the algorithm failed to return a result in five seconds in several cases.
Table~\ref{tab:ex:sea:all} shows results for the Top-down (Search) algorithm when generating all satisfiable annotations, ie. the algorithm proceeds to find all solution patterns.
BGPs with more than two triples required more than five seconds.
Results show how with just 2 triples (S\_2T), the algorithm generated 36 viable solutions after evaluating an average of 585926.3 $\pm$ 8903.79 possible annotations. We can conclude that this algorithm is inefficient.

Table~\ref{tab:ex:csp:sat} shows results for the Bottom up (CSP) algorithm when doing satisfiability check, ie. the algorithm stops when the first satisfiable annotation is found. 
Average and standard deviation for the number of solutions generated is fixed as the algorithm is fully deterministic (although the process stops when the first viable solution pattern is found).
Table~\ref{tab:ex:csp:all} shows results for the Bottom-up (CSP) algorithm when generating all satisfiable annotations, ie. the algorithm proceeds to find all solution patterns.
With this algorithm, we managed to evaluate all BGPs for satisfiability and, in addition, to generate all possible solution patterns, in the worst case, in less than two seconds.
However, the worst results (N\_5T) goes still above the second, which opens the question whether it would be possible to perform further optimisations.
We can conclude that the Bottom-up (CSP) algorithm is more efficient and practically sustainable. 

\chdeleted{The proof-of-concept implementation and the experimental code is published~\cite{fxbgp} and experiments can be reproduced:~\url{https://github.com/SPARQL-Anything/fxbgp}.}

\begin{figure}[t]
    \centering
    \includegraphics[trim=20 0 0 42, clip, width=0.8\linewidth]{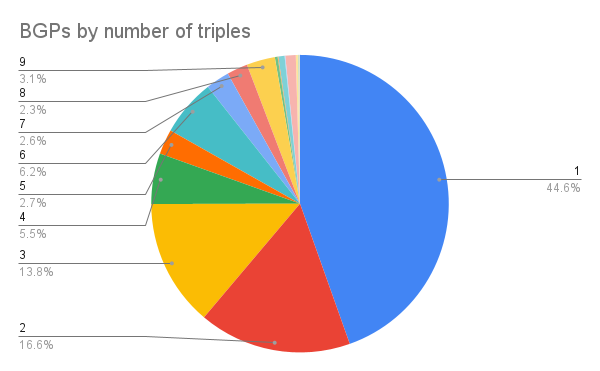}
    \caption{Dimension of BGPs from real-world queries: number of triples.}
    \label{fig:rwq:triples}
\end{figure}

\subsection{Evaluation with real world queries}\label{sec:ex:rwq}
\chadded{We evaluate the practical feasibility with basic graph patterns derived from real world queries. 
We only perform experiments with the most efficient, bottom up algorithm (Listing~\ref{algorithm2}).
Here, we summarise the data and results. 
We refer the reader to Appendix~\ref{app:rwq} for extensive details on data, results, and reproducibility.}

\subsubsection{Data}
\chadded{We extract 1,309 BGPs from 449 queries collected from public GitHub repositories. 228 of those queries have only a single BGP (50.7\%). 
As illustrated in figures~\ref{fig:rwq:triples} and~\ref{fig:rwq:variables}, 83.7\% of the queries have less than 5 BGPs, 44.6\% of the collected BGPs have a single triple, while 75.9\% of them have up to 5 variables. We can conclude that the majority of real world queries and related BGPs have few triples and variables, and have similar sizes to the curated BGPs designed to evaluate the two algorithms. See further details in Appendix~\ref{app:rwq}.}

\begin{figure}[t]
    \centering
    \includegraphics[trim=20 0 0 42, clip, width=0.8\linewidth]{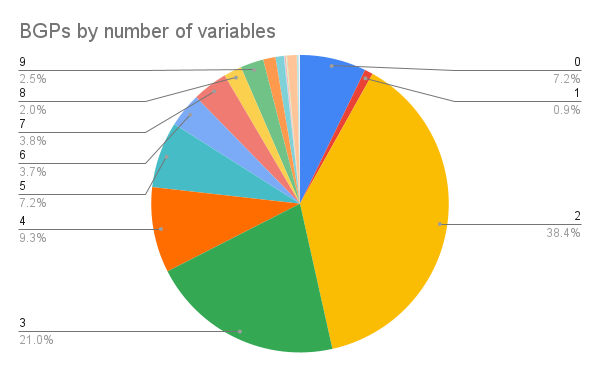}
    \caption{Dimension of BGPs from real-world queries: number of variables.}
    \label{fig:rwq:variables}
\end{figure}
\subsubsection{Results}
\chadded{Satisfiability (finding the first valid solution pattern) can be verified in less than 100ms for all BGPs with up to 10 triples or 14 variables.
The only problematic BGP is one with 19 triples and 19 variables
(‘query-77.rq‘ in the supplemental material), which cannot be verified in less than 5 seconds. 
Similarly, all solution patterns can be generated efficiently in all but two cases (below half second). 
The only two problematic BGPs are one with 19 triples and 19 variables (`query-77.rq` in the supplemental material), and another with 20 triples and 23 variables (`query-244.rq` in the supplemental material), for which it is not possible to generate all solution patterns in less than 5 seconds. See further details in Appendix~\ref{app:rwq}.}

\subsection{Evaluation of benefits with benchmark queries}
\chadded{We evaluate the advantages of assessing the satisfiability of a query before its evaluation. We  perform experiments with the most efficient, i.e. the bottom up algorithm.}

\subsubsection{Data}
\chadded{The evaluation relied on the  of the GTFS-Madrid-Bench  benchmark~\cite{chaves2020gtfs}, a virtual knowledge graph access benchmark aimed at assessing the performance of Ontology-based Data Access (OBDA) engines. More specifically, we use a specialisation of the benchmark~\cite{asprino2025materialisation} for  façade-based data access engines.}

\chadded{The benchmark comprises a set of 18 SPARQL queries that interrogate a collection of data sources.}
\chadded{Table~\ref{tab:gtfs_queries} (see Appendix~\ref{app:bench-queries}) reports the characteristics of the queries of the benchmark in terms of the  number of BGPs, the number of variables per BGP, and the number of triple patterns per BGP.}
    
\chadded{Furthermore, the benchmark offers a data generator that can increase the size of the original data and distribute the datasets across various formats.}

\chadded{In this experiment, we use JSON to represent hierarchical data formats and CSV for tabular data sources. 
We generate datasets of sizes 1 and 10 and execute the 18 queries on them using the SPARQL Anything engine (version v1.1.0)}\footnote{\url{https://github.com/SPARQL-Anything/sparql.anything/releases/tag/v1.1.0}}.
\chadded{We execute each query three times and record the time taken to construct and load the \FX{} RDF graph in memory for each execution. Finally, we calculate the average of these three runs.
Moreover, for each query of the benchmark, we run the bottom-up algorithm for assessing the query satisfiability.
Table~\ref{tab:gtfs_queries} (see Appendix~\ref{app:bench-queries}) reports construction and loading time and time for assessing the satisfiability of the queries.}

\subsubsection{Results}
\chadded{The satisfiability of the GTFS-Madrid-Bench queries can be assessed in less than 15 ms. 
Excluding the first query reduces the maximum time to 2.4.
In contrast, construction and loading times are at least two orders of magnitude longer (the shortest recorded time was 161.5 ms for query 6, size 1, with CSV input). This  demonstrates the benefit of assessing satisfiability prior to query execution, since construction and loading time can be avoided for unsatisfiable queries, thus saving computational resources.
Naturally, the larger the input, the greater the benefit: while construction and loading times increase with the size of the input, assessment of satisfiability is independent from it.}

\section{Conclusions}\label{sec:conclusions}%\todo[inline]{Add discussion on limit cases and future work on optimizations...}
In this article, we studied the satisfiability of SPARQL basic graph patterns in \FX{}.
\chadded{We introduced a consolidated version of \FX{} including formalised mappings to RDF. 
We studied the conditions under which a BGPs are satisfiable as the basis for an algorithm to verify that efficiently. 
Extensive experiments demonstrate the feasibility of satisfiability check. Two extreme cases of large BGPs motivate future work on further improving the performance, focusing on reducing the space of potential \FX{} annotations before verifying the constraints.}
In the future, \chadded{we }will extend this study to consider specific \FX{} profiles, studying how those can bring additional constraints to verify (e.g. \FX{} graphs from CSVs will never have types).
Crucially, identifying solution patterns is a preliminary step for investigating classes of queries (e.g. star queries) whose solutions can be streamed directly from the file sources, without needing to load intermediate results in-memory.
Finally, our work paves the way for studying query execution strategies for Façade-X data access with SPARQL and supporting developers in building more efficient data integration systems for knowledge graphs.

%%
%% Bibliography
%%

%% Please use bibtex, 
\bibliography{main}

\appendix
\section{Experiments with real-world queries}\label{app:rwq}
We report on experiments with real world SPARQL queries tailored to facade based data access.
The queries were collected from the GitHub API on 18 December 2025 and include all files published in public repositories with extensions `.rq` and `.sparql`, which include the string `x-sparql-anything` in the text. 
The code and data are published\footnote{\url{https://github.com/SPARQL-Anything/facade-x-queries-on-github/}} and archived~\cite{fxqueries}.
The analysis was performed on Google drive and can be viewed at~\url{https://docs.google.com/spreadsheets/d/11n7c101WLQBL8aApZs-fXNtHAYTIBIP9nZUDxD6TUa8/edit}.

We perform experiments for both satisfiability and to generate all possible solution patterns on 449 collected queries.
%In our experiments, we test our algorithm with BGPs extracted from real world queries. 
We extract 1430 BGPs from the 449 queries.
Table~\ref{tab:rwq:queriesbybgp} shows a summary of the queries grouped by number of BGPs.
Of these queries, 58 cannot be parsed because they include syntax errors.
228 of the remaining queries have a single BGP.
The majority of queries have less than five basic graph patterns.
Table~\ref{tab:rwq:bgpbytriples} summarizes the BGPs extracted from the queries, to which we apply our satisfiability problem.
103 BGPs with 0 triples refer to BGPs that included SPARQL Anything configuration options, that we exclude.
Figure~\ref{fig:rwq:bgpbyvar} shows the number of BGPs by the number of variables included.
Clearly, the more the variables, the more complex it is to verify satisfiability, since variables are annotated with more potential annotations.
The majority of BGPs have between 1 and 5 variables.
In what follows, we show the average values of 10 executions.

\begin{table}[t]
\centering
\caption{Number of queries by BGPs.}\label{tab:rwq:queriesbybgp}
\begin{tabular}{l | r}
\hline
\textbf{Queries} & \textbf{BGPs} \\\hline\hline
1 & 40 \\
1 & 31 \\
1 & 29 \\
2 & 27 \\
2 & 26 \\
1 & 25 \\
1 & 24 \\
2 & 22 \\
1 & 20 \\
1 & 17 \\
2 & 13 \\
1 & 12 \\
11 & 11 \\
9 & 10 \\
4 & 9 \\
1 & 8 \\
6 & 7 \\
4 & 6 \\
22 & 5 \\
32 & 4 \\
37 & 3 \\
79 & 2 \\
228 & 1 \\
\hline
\end{tabular}
\end{table}

\begin{table}[t]
\centering
\caption{Basic graph patterns grouped by number of triples (that are not SPARQL Anything configuration options).}\label{tab:rwq:bgpbytriples}
\begin{tabular}{l l}
\hline
\textbf{Number of triples} & \textbf{Number of BGPs} \\\hline\hline
0 & 103 \\
1 & 591 \\
2 & 220 \\
3 & 183 \\
4 & 73 \\
5 & 36 \\
6 & 82 \\
7 & 34 \\
8 & 30 \\
9 & 41 \\
10 & 4 \\
11 & 1 \\
12 & 9 \\
13 & 1 \\
14 & 15 \\
16 & 3 \\
18 & 1 \\
19 & 1 \\
20 & 1 \\
\hline

\end{tabular}
\end{table}

\subsection{Satisfiability check}
We perform the same experiment setting as with the designed queries (Section~\ref{sec:experiments}), focusing on the most performing bottom-up algorithm only, with the 449 real world queries.

Table~\ref{tab:rwq:satonlyt} shows the average duration of the more efficient bottom-up satisfiability check algorithm presented in Section~\ref{sec:algorithms}. 
Data is grouped by BGPs with the same number of triples. 
BGPs with 0 triples are the ones that only included triples with SPARQL Anything configuration, and can be ignored.
Table~\ref{tab:rwq:satonlyv} shows the average duration grouped by BGPs with the same number of variables.
We can see how for BGPs with up to 16 triples and 18 variables (except one with 15 variables), satisfiability can be verified in less than 100 ms.
The only problematic BGP is one with 19 triples and 19 variables (`query-77.rq` in the supplemental material), which cannot be verified in less than 5 seconds.
We ignore errors from two other BGPs, one within a DBPedia service clause (`query-66.rq`), and another coming from a negative test within a SPARQL Anything code base (`query-199.rq`).

\subsection{Generating all solutions}
We perform the same experiment setting as with the designed queries (Section~\ref{sec:experiments}), focusing on the most performing bottom-up alogorithm only, with the 449 real world queries.

Table~\ref{tab:rwq:solutionst} shows the average duration of the more efficient bottom-up algorithm presented in Section~\ref{sec:algorithms} to generate all possible solution patterns. 
Data is grouped by BGPs with the same number of triples. 
In addition, Table~\ref{tab:rwq:solutionsv} shows the average duration grouped by BGPs with the same number of variables.
Solutions can be generated efficiently in all but two cases (below half second). 
The only two problematic BGPs are one with 19 triples and 19 variables (`query-77.rq` in the supplemental material), and another with 20 triples and 23 variables (`query-244.rq` in the supplemental material), for which it is not possible to generate all solution patterns in less than 5 seconds.

%%%%%%%%

% \begin{table}
% \centering

% \begin{tabular}{l | r}
% \textbf{Variables} & \textbf{BGPs} \\
% 0 & 103 \\
% 1 & 13 \\
% 2 & 548 \\
% 3 & 300 \\
% 4 & 133 \\
% 5 & 103 \\
% 6 & 53 \\
% 7 & 55 \\
% 8 & 28 \\
% 9 & 36 \\
% 10 & 19 \\
% 11 & 13 \\
% 12 & 1 \\
% 13 & 2 \\
% 14 & 1 \\
% 15 & 1 \\
% 16 & 15 \\
% 18 & 3 \\
% 19 & 1 \\
% 23 & 1 \\

% \end{tabular}
% \caption{BGPs by number of variables.}\label{tab:rwq:bgpbyvar}
% \end{table}

\begin{figure}
    \centering\includegraphics[trim=40 40 0 42, clip, width=1\linewidth]{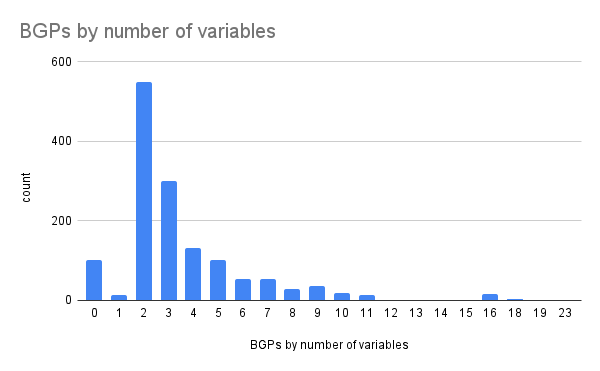}
    \caption{BGPs by number of variables.}
    \label{fig:rwq:bgpbyvar}
\end{figure}

\begin{table}[t]
\centering
\caption{Performance of the bottom-up satisfiability check algorithm on real world queries. 
We group the data by size of BGPs (number of triples). 
There is only 1 BGP with 11,13,18,19, and 20 triples.
One query with 19 triples cannot be evaluated in less the 5 seconds (our limit).}\label{tab:rwq:satonlyt}
\begin{tabular}{r r r}
\hline
\textbf{Triples} & \textbf{AVG} & \textbf{STD} \\\hline\hline
1 & 0.10 & 0.3 \\
2 & 0.11 & 0.37 \\
3 & 0.14 & 0.27 \\
4 & 0.14 & 0.36 \\
5 & 0.41 & 0.29 \\
6 & 0.28 & 0.17 \\
7 & 0.48 & 0.15 \\
8 & 1.36 & 0.67 \\
9 & 0.48 & 0.25 \\
10 & 1.15 & 0.58 \\
11 & 0.20 & - \\
12 & 6.73 & 1.81 \\
13 & 8.30 & - \\
14 & 0.33 & 0.07 \\
16 & 0.43 & 0.02 \\
18 & 218.30 & - \\
\rowcolor{lightgray}19 & >5s & - \\
20 & 8.20 & - \\\hline
\end{tabular}
\end{table}

\begin{table}[t]
\centering
\caption{Performance of the bottom-up satisfiability check algorithm on real world queries. 
We group the data by the size of BGPs - number of variables. 
There are only single BGPs with 12,14,15,19, and 23 variables.
One query with 19 variables cannot be evaluated in less the 5 seconds (our limit).}\label{tab:rwq:satonlyv}
\begin{tabular}{r r r}
\hline
\textbf{Variables} & \textbf{AVG} & \textbf{STD} \\\hline\hline
1 & 0.05 & 0.18 \\
2 & 0.11 & 0.31 \\
3 & 0.11 & 0.26 \\
4 & 0.15 & 0.42 \\
5 & 0.25 & 0.22 \\
6 & 0.27 & 0.41 \\
7 & 0.37 & 0.13 \\
8 & 0.22 & 0.18 \\
9 & 1.65 & 0.69 \\
10 & 0.42 & 0.26 \\
11 & 4.09 & 1.61 \\
12 & 0.20 & - \\
13 & 4.30 & 0.77 \\
14 & 0.20 & - \\
15 & 218.30 & - \\
16 & 0.33 & 0.07 \\
18 & 0.43 & 0.02 \\
\rowcolor{lightgray}19 & >5s & - \\
23 & 8.20 & - \\
\hline
\end{tabular}

% \todo[inline]{@Enrico, double check table headers}

\end{table}
\begin{table}
\centering
\caption{Performance of the bottom-up algorithm on real world queries, when generating all solution patterns. 
We group the data by size of BGPs (number of triples). 
There is only 1 BGP with 11,13,18,19, and 20 triples.
For two queries with 19 and 20 triples, the complete set of solution patterns cannot be evaluated in less the 5 seconds (our limit).}\label{tab:rwq:solutionst}
\begin{tabular}{r r r}
\hline
BGP Size (number of triples) & AVG & STD \\\hline\hline
1 & 0.05 & 0.37 \\
2 & 0.08 & 0.19 \\
3 & 0.13 & 0.25 \\
4 & 0.23 & 0.16 \\
5 & 8.05 & 6.27 \\
6 & 1.36 & 0.82 \\
7 & 2.23 & 1.01 \\
8 & 7.79 & 1.75 \\
9 & 5.51 & 1.19 \\
10 & 38.58 & 3.21 \\
11 & 21.50 & - \\
12 & 58.10 & 20.53 \\
13 & 163.70 & - \\
14 & 191.97 & 12.58 \\
16 & 1,250.00 & 40.42 \\
18 & 1,127.40 & - \\\hline
\rowcolor{lightgray}19 & >5s & - \\
\rowcolor{lightgray}20 & >5s & - \\
\end{tabular}
\end{table}

\begin{table}
\centering
\centering
\caption{Performance of the bottom-up algorithm on real world queries. 
We group the data by the size of BGPs - number of variables. 
There are only single BGPs with 12,14,15,19, and 23 variables.
For two queries with 19 and 23 variables the set of solution patterns cannot be generated in less the 5 seconds (our limit).}\label{tab:rwq:solutionsv}
\begin{tabular}{r r r}
\hline
BGP Size (number of variables) & AVG & STD \\\hline\hline
1 & 0.02 & 0.14 \\
2 & 0.04 & 0.36 \\
3 & 0.08 & 0.29 \\
4 & 0.14 & 0.32 \\
5 & 0.24 & 0.14 \\
6 & 0.46 & 0.16 \\
7 & 1.21 & 0.58 \\
8 & 1.19 & 0.23 \\
9 & 8.47 & 2.01 \\
10 & 25.58 & 8.36 \\
11 & 21.15 & 3.1 \\
12 & 21.50 & - \\
13 & 256.55 & 35.91 \\
14 & 67.40 & - \\
15 & 1,127.40 & - \\
16 & 191.97 & 12.58 \\
18 & 1,250.00 & 40.42 \\
\rowcolor{lightgray}19 & -1.00 & - \\
\rowcolor{lightgray}23 & -1.00 & - \\\hline
\end{tabular}
\end{table}

\section{Experiments with benchmark queries}\label{app:bench-queries}
Table~\ref{tab:gtfs_queries} describes the characteristics of the queries of the GTFS-Madrid-Bench benchmark in terms of number of BGPs for each query, number of variables and number of triple patterns per BGP.

\begin{table}[t]
\centering

\caption{Characteristics of the queries contained in the GTFS-Madrid-Bench benchmark.}
\label{tab:gtfs_queries}
\begin{tabular}{p{0.2\textwidth}p{0.2\textwidth}p{0.2\textwidth}p{0.2\textwidth}}
\hline
\textbf{Query ID} & \textbf{Number of BGPs} & \textbf{Number of variables per BGP} & \textbf{Number of Triple Patterns per BGP} \\
\hline\hline
1  & 1  & 5 & 4 \\
2  & 4  & 2,2,2,3 & 1,1,1,2 \\
3  & 5  & 2,2,2,3,2 & 1,1,1,2,1 \\
4  & 7  & 2,2,2,2,2,4,2 & 1,1,1,1,1,3,1 \\
5  & 1  & 4 & 3 \\
6  & 1  & 3 & 2 \\
7  & 10 & 2,2,2,3,2,2,3,2,2,3 & 1,1,1,2,1,1,2,2,1,2 \\
8  & 11 & 2,2,2,3,2,2,4,2,2,2,3 & 1,1,1,2,1,1,3,1,1,1,2 \\
9  & 3  & 5,2,4 & 4,1,3 \\
10 & 2  & 2,3 & 1,2 \\
11 & 4  & 3,4,1,2 & 2,3,2,1 \\
12 & 4  & 3,3,3,3 & 2,2,2,3 \\
13 & 3  & 3,2,4 & 2,1,3 \\
14 & 4  & 5,3,2,2 & 4,2,1,1 \\
15 & 1  & 4 & 2 \\
16 & 2  & 4,3 & 3,3 \\
17 & 4  & 2,2,4,4 & 1,1,3,3 \\
18 & 5  & 2,4,2,2,2 & 2,3,1,1,1 \\
\hline
\end{tabular}
\end{table}

Table ~\ref{table:loading_time_sat_time_gtfs_queries} shows how long it takes to construct and load the RDF graph required to answer each query into memory.
Additionally, the time required to compute the satisfiability of all BGPs in each query is reported.

\begin{table}[t]

\caption{Construction Loading time and time for assessing the satisfiability of the queries in the GTFS-Madrid-Bench benchmark. }
\label{table:loading_time_sat_time_gtfs_queries}
\begin{tabular}{c|cccc|p{5cm}}

\hline
\textbf{Query}     & 
\multicolumn{4}{c|}{\textbf{Construction + Loading time (ms)}}  & 
\textbf{Time for computing the satisfiability of all BGPs (ms)} \\

     & 
\multicolumn{2}{c}{\textbf{CSV}} & 
\multicolumn{2}{c|}{\textbf{JSON}}  &   \\

 & \textbf{Size 1} & \textbf{Size 10} & \textbf{Size 1} & \textbf{Size 10} & \\\hline\hline
 
1                         & 1049.5                    & 7476.5                     & 1221.5                    & 6252                       & 14.40                                                                                      \\
2                         & 310.5                     & 607.5                      & 284                       & 603.5                      & 0.80                                                                                       \\
3                         & 402.5                     & 595                        & 301                       & 673.5                      & 0.70                                                                                       \\
4                         & 283                       & 251                        & 218                       & 241.5                      & 1.70                                                                                       \\
5                         & 166.5                     & 230                        & 215.5                     & 222                        & 0.20                                                                                       \\
6                         & 161.5                     & 177                        & 166                       & 218                        & 0.10                                                                                       \\
7                         & 390                       & 1031.5                     & 390                       & 1151.5                     & 2.40                                                                                       \\
8                         & 449                       & 1598.5                     & 410                       & 1627                       & 1.10                                                                                       \\
9                         & 1390                      & 6078.5                     & 1260                      & 5609.5                     & 0.70                                                                                       \\
10                        & 292                       & 656.5                      & 292                       & 609.5                      & 0.90                                                                                       \\
11                        & 240.5                     & 331.5                      & 218.5                     & 306                        & 1.80                                                                                       \\
12                        & 379.5                     & 851                        & 305                       & 841                        & 1.70                                                                                       \\
13                        & 275.5                     & 566.5                      & 246.5                     & 515                        & 1.60                                                                                       \\
14                        & 327                       & 1107                       & 348.5                     & 1173                       & 1.20                                                                                       \\
15                        & 270                       & 1340                       & 289                       & 1291                       & 0.40                                                                                       \\
16                        & 202.5                     & 299.5                      & 213                       & 294.5                      & 0.40                                                                                       \\
17                        & 233                       & 438                        & 264.5                     & 422.5                      & 0.80                                                                                       \\
18                        & 214                       & 309                        & 237                       & 298                        & 0.80                                                                                      \\\hline
\end{tabular}
\end{table}

% \break\listofchanges

\end{document}